\documentclass{amsart}
\usepackage[latin1]{inputenc}
\usepackage[english]{babel}
\usepackage{appendix}
\usepackage{latexsym}
\usepackage{amsmath}
\usepackage{amsfonts}
\usepackage{amssymb}
\usepackage{amstext}
\usepackage{mathrsfs}               % \mathscr
\usepackage{bbm}                    % \mathbbm 
\usepackage{bbold}                  % \mathbb 
\usepackage{textcomp}               
\usepackage{enumerate}
\usepackage{enumitem}
\usepackage{units}
\usepackage{stmaryrd}
\newtheorem{thm}{Theorem}[section]
\newtheorem{cor}[thm]{Corollary}
 \newtheorem{lem}[thm]{Lemma}
 \newtheorem{prop}[thm]{Proposition}
 \theoremstyle{definition}
 \newtheorem{defn}[thm]{Definition}
 \theoremstyle{remark}
 \newtheorem{rem}[thm]{Remark}
 \newtheorem{ex}[thm]{Example}
 \numberwithin{equation}{section}
\newcommand{\CC}{\mathbb{C}}
\newcommand{\EE}{\mathbb{E}}

\newcommand{\NN}{\mathbb{N}}
\newcommand{\PP}{\mathbb{P}}

\newcommand{\RR}{\mathbb{R}}

\newcommand{\supp}{\mathrm{supp}}
\newcommand{\dist}{\mathrm{dist}}
\newcommand{\Ran}{\mathrm{Ran}}

\newcommand{\loc}{\mathrm{loc}}

\newcommand{\ren}{\mathrm{ren}}
\newcommand{\Id}{\mathrm{d}}
\newcommand{\SPn}[2]{\langle #1|#2\rangle} 
\newcommand{\SPb}[2]{\big\langle #1\big|#2\big\rangle} 
\newcommand{\SPB}[2]{\Big\langle #1\Big|#2 \Big\rangle}
\newcommand{\ol}[1]{\overline{#1}} 
 
\newcommand{\mr}[1]{\mathring{#1}}

\newcommand{\wt}[1]{\widetilde{#1}}
\newcommand{\restr}{\mathord{\upharpoonright}}
\newcommand{\nf}[2]{{#1/#2}}
\newcommand{\eh}{{1/2}}
\newcommand{\mh}{{-1/2}}
\newcommand{\cO}{\mathcal{O}} 
\newcommand{\cC}{\mathcal{C}}
\newcommand{\cD}{\mathcal{D}}\newcommand{\cP}{\mathcal{P}} 
\newcommand{\cQ}{\mathcal{Q}}

\newcommand{\cG}{\mathcal{G}}

\newcommand{\cU}{\mathcal{U}}

\newcommand{\cK}{\mathcal{K}}

\newcommand{\sA}{\mathscr{A}}
 
\newcommand{\sC}{\mathscr{C}}
\newcommand{\sD}{\mathscr{D}} 
\newcommand{\sE}{\mathscr{E}}\newcommand{\sQ}{\mathscr{Q}}
\newcommand{\sF}{\mathscr{F}}
\newcommand{\sG}{\mathscr{G}}

\newcommand{\sI}{\mathscr{I}}\newcommand{\sU}{\mathscr{U}}
\newcommand{\sV}{\mathscr{V}}
\newcommand{\sK}{\mathscr{K}}\newcommand{\sW}{\mathscr{W}}
\newcommand{\sL}{\mathscr{L}}

\newcommand{\fB}{\mathfrak{B}}

\newcommand{\fQ}{\mathfrak{Q}}
\newcommand{\fF}{\mathfrak{F}}
\newcommand{\fG}{\mathfrak{G}}
\newcommand{\fH}{\mathfrak{H}}

\newcommand{\V}[1]{\boldsymbol{#1}}
\newcommand{\ve}{\varepsilon}
\newcommand{\vp}{\varphi}

\newcommand{\vk}{\varkappa}
\newcommand{\vr}{\varrho}
\newcommand{\vt}{\vartheta}
\newcommand{\vs}{\varsigma}
\newcommand{\id}{\mathbbm{1}}                   % Identity 
\newcommand{\dom}{\cD}                          % domain of definition
\newcommand{\fdom}{\cQ}                         % form domain 
\newcommand{\HR}{\mathscr{H}}                   % Hilbert space
\newcommand{\LO}{\mathcal{B}}                   % set of bounded, linear operators
\newcommand{\ad}{a^\dagger}                     % creation op.
\newcommand{\ee}{\mathfrak{g}}

\newcommand{\Geb}{\cG}
\newcommand{\UV}{\Lambda}
\newcommand{\Wi}{\mathrm{{\textsc{w}}}}
\newcommand{\NV}{H}  
\newcommand{\nv}{\mathfrak{h}}
\newcommand{\HV}{\widetilde{H}}
\newcommand{\hv}{\tilde{\mathfrak{h}}}
\newcommand{\Thh}{\Sigma}
\newcommand{\EV}{E}
\newcommand{\SGV}{S_{\Geb}}
\newcommand{\QGV}{\mathscr{Q}_{\Geb}}
\newcommand{\cfG}{\mathfrak{q}_{\Geb}}
\newcommand{\TV}{T}

\newcommand{\GM}{\nu}
\newcommand{\Th}[2]{E_{\Geb_{#1},#2}}
\begin{document}

\title[Ground states in the renormalized Nelson model]{Ground States and Associated Path Measures in the Renormalized Nelson Model}

\author{Fumio Hiroshima \and Oliver Matte}

\address{Fumio Hiroshima, Faculty of Mathematics, Kyushu University, Fukuoka Motooka 744, 
819-0395, Japan}

\email{hiroshima@math.kyushu-u.ac.jp}

\address{Oliver Matte, Institut for Matematiske Fag, Aalborg Universitet, Skjernvej 4A, 9220 Aalborg, Denmark}

\email{oliver@math.aau.dk}

\begin{abstract}
We prove the existence, uniqueness, and strict positivity of ground states of the possibly massless 
renormalized Nelson operator under an infrared regularity condition and for Kato decomposable 
electrostatic potentials fulfilling a binding condition. If the infrared regularity condition is violated, 
then we show non-existence of ground states of the massless renormalized Nelson operator with an 
arbitrary Kato decomposable potential. Furthermore, we prove the existence, uniqueness, and strict
positivity of ground states of the massless renormalized Nelson operator in a non-Fock representation
where the infrared condition is unnecessary.
Exponential and superexponential estimates on the pointwise spatial decay and the decay
with respect to the boson number for elements of spectral subspaces below localization
thresholds are provided. Moreover, some continuity properties of ground state eigenvectors are
discussed. Byproducts of our analysis are a hypercontractivity bound for the semigroup and
a new remark on Nelson's operator theoretic renormalization procedure.
Finally, we construct path measures associated with ground states of the renormalized Nelson
operator. Their analysis entails improved boson number decay estimates for ground state 
eigenvectors, as well as upper and lower bounds on the Gaussian localization with respect to the field 
variables in the ground state. 
As our results on uniqueness, positivity, and path measures exploit the ergodicity of the 
semigroup, we restrict our attention to one matter particle. All results are non-perturbative.

\medskip


\medskip

\end{abstract}

\maketitle

%%%%%%%%%%%%%%%%%%%%%%%%%%%%%%%%%%%%%%%%%%%%%%%%%
%%%%%%%%%%%%%%%%%%%%%%%%%%%%%%%%%%%%%%%%%%%%%%%%%
%%%%%%%%%%%%%%%%%%%%%%%%%%%%%%%%%%%%%%%%%%%%%%%%%

\section{Introduction and Main Results}\label{secintro}

\subsection{General introduction}

\noindent
The Nelson model describes a conserved number of non-relativistic quantum 
mechanical matter particles linearly coupled to a quantized radiation field comprising  
relativistic bosons. Its crucial feature is its comparatively simple renormalizability as demonstrated by 
Nelson back in 1964 \cite{Nelson1964proc,Nelson1964}. After introducing an ultraviolet
cutoff in the {\em a priori} ill-defined interaction terms of the Hamiltonian,
the addition of a diverging energy counter term suffices in fact to achieve norm resolvent
convergence of the so obtained operators as the cutoff goes to infinity.
Although the model has been studied extensively ever since, there still exist a lot of open 
mathematical questions on its spectral theory. In this article we shall address some problems left open 
in studying ground state eigenvectors and in particular in establishing their existence and uniqueness 
or proving their absence, depending on the respective assumptions. Since we are also interested in 
utilizing certain path measures associated with ground states, whose construction exploits the 
ergodicity of the semigroup generated by the Hamiltonian, we shall consider only one matter particle.

The key aspect of our work is that all theorems are {\em non-perturbative}, apply to the 
{\em renormalized} model, and cover {\em massless bosons} at the same time. Besides existence,
uniqueness, and absence of ground states, we shall discuss exponential and super-exponential 
decays with respect to the boson number and the spatial variable of the matter particle for ground 
state eigenvectors and more general elements of spectral subspaces. In addition we provide
upper and lower bounds describing the Gaussian localization with respect to the field variables
in ground states. We also present a new
hypercontractivity bound and address the continuity of ground state eigenvectors with respect to 
spatial coordinates and parameters. Furthermore, we treat the 
renormalized Nelson model in a non-Fock representation as well, and we give a new remark on
Nelson's operator theoretic renormalization procedure \cite{Nelson1964}. 
If an ultraviolet regularization is introduced in the model, then all our results are well-known, apart from 
the hypercontractivity bound and a few other technical improvements we shall mention later on. 
We comment on the earlier 
literature in the next subsection during a more detailed presentation of our main results. 

Before this however, we shall introduce our standing hypotheses on the model which are
{\em always} assumed to be satisfied throughout the whole article. In fact, the Nelson operators 
studied here depend on the following quantities:
\begin{enumerate}
\item[$\bullet$] 
a Kato decomposable exterior potential $V:\RR^3\to\RR$, which is a function of the spatial
coordinate $\V{x}$ of the matter particle. Hence, $V=V_+-V_-$ for some $V_-\ge0$ in the 
Kato class $K_3$ and $V_+\ge0$ in the local Kato class $K^{\loc}_3$. We recall that
$V_-\in K_3$ means by definition of $K_3$ that $V_-$ is measurable and
\begin{align*}
\lim_{r\downarrow 0}\sup_{\V{x}\in\RR^3}\int_{\{\V{y}\in\RR^3:|\V{x}-\V{y}|<r\}}\frac{V_-(\V{y})}{|\V{x}-\V{y}|}\Id\V{y}&=0.
\end{align*} 
Furthermore, $V_+\in K^{\loc}_3$ means explicitly that $V_+$ is measurable and satisfies
\begin{align*}
\lim_{r\downarrow 0}\sup_{\V{x}\in\RR^3}\int_{\{\V{y}\in\cK:|\V{x}-\V{y}|<r\}}
\frac{V_+(\V{y})}{|\V{x}-\V{y}|}\Id\V{y}&=0,
\end{align*} 
for every compact $\cK\subset\RR^3$. 
These fairly general assumptions on $V$ permit to define all Hamiltonians we are interested in by 
means of quadratic forms as well as to derive bounds on the associated semigroups with the help
of Feynman-Kac formulas. See \cite{AizenmanSimon1982} for a detailed discussion of Kato classes.
\item[$\bullet$] 
a non-negative boson mass $\mu\ge0$. The corresponding dispersion relation for a single
boson is given by
\begin{align*}
\omega(\V{k})&:=(\V{k}^2+\mu^2)^\eh,\quad\V{k}\in\RR^3.
\end{align*}
\item[$\bullet$] 
a measurable function $\eta:\RR^3\to\RR$ with $0\le\eta\le1$, which is even in the sense
that $\eta(-\V{k})=\eta(\V{k})$, for all $\V{k}\in\RR^3$. It is employed to impose a mild infrared 
regularization when we construct ground states of the massless Nelson operator.
Otherwise, the most relevant choice is $\eta=1$.
\item[$\bullet$] 
a coupling constant $\ee\in\RR$. In fact, its value does not affect the validity of any result 
obtained in this article.
\end{enumerate}

The heuristic interaction kernel supposed to appear in the field operator coupling the 
matter particle and the radiation field is thus given by the function 
$$
\RR^3\times\RR^3\ni(\V{x},\V{k})\longmapsto
\ee\omega(\V{k})^\mh\eta(\V{k})e^{-i\V{k}\cdot\V{x}}.
$$
In the most interesting case where $\eta=1$ away from zero, it is obviously not square-integrable with 
respect to $\V{k}$, whence Nelson developed his renormalization procedure.
\begin{enumerate}
\item[$\bullet$] 
Let $\NV$ stand for the renormalized Nelson operator corresponding to the above data. It is acting
in the Hilbert space $L^2(\RR^3,\sF)$ with $\sF$ denoting the bosonic Fock space over $L^2(\RR^3)$.
As alluded to above, $\NV$ is the norm resolvent limit as an ultraviolet cutoff parameter $\UV>0$
goes to infinity of the operators formally given by
\begin{align*}
&-\frac{1}{2}\Delta_{\V{x}}+V+\int_{\RR^3}\omega(\V{k})\ad(\V{k})a(\V{k})\Id\V{k}
\\
&\quad+\ee\int_{\{\V{k}\in\RR^3:|\V{k}|\le\UV\}}\frac{\eta(\V{k})}{\omega(\V{k})^{\eh}}(e^{-i\V{k}\cdot\V{x}}\ad(\V{k})
+e^{i\V{k}\cdot\V{x}}a(\V{k}))\Id\V{k}+E_{\UV}^{\ren}.
\end{align*}
Here $E_{\UV}^{\ren}>0$ is the logarithmically divergent energy counter term; as usual
$\ad(\V{k})$ and $a(\V{k})$ denote the ``pointwise" bosonic creation and annihilation operators
associated with $\sF$. Our article 
%-- see in particular Section~\ref{secdefs} -- 
contains in fact a self-contained construction of $H$ taylor-made on our needs and complementing 
the existing, rich literature on the subject.
\end{enumerate}
When describing our results obtained through path measures associated with ground states in the
next subsection, we also employ some notation related to a Feynman-Kac formula for $\NV$ derived 
in \cite{MatteMoeller2017}. This formula will be used crucially at several places in our article.
It has the form
\begin{align}\label{FKintro}
(e^{-t\NV}\Psi)(\V{x})&=\EE[W^V_{t}(\V{x})^*\Psi(\V{B}_t^{\V{x}})],\quad\text{a.e. $\V{x}\in\RR^3$},
\end{align}
with a three-dimensional standard Brownian motion $\V{B}$ and $\V{B}^{\V{x}}:=\V{x}+\V{B}$.
Unlike earlier Feynman-Kac representations
for the renormalized Nelson model \cite{GHL2014,Nelson1964proc}, the formula \eqref{FKintro}
contains the {\em $\LO(\sF)$-valued} stochastic process $(W^V_{t}(\V{x}))_{t\ge0}$ 
with finite moments of any order and it can be applied to {\em every} $\Psi\in L^2(\RR^3,\sF)$;
see Subsection~\ref{ssecFK} for more details. 
\begin{enumerate}
\item[$\bullet$] $\LO(\sV)$ denotes the space of bounded linear operators on some
normed vector space $\sV$.
\end{enumerate}

%%%%%%%%%%%%%%%%%%%%%%%%%%%%%%%%%%%%%%%%%%%%%%%%%
%%%%%%%%%%%%%%%%%%%%%%%%%%%%%%%%%%%%%%%%%%%%%%%%%

\subsection{Description of Results}

\noindent
As  usual, the first step towards proving the existence of ground states are localization
estimates. The next two theorems, which are proved at the end of Subsection~\ref{ssecexploc},
actually provide more information than necessary.
Both of them are new for the renormalized model; see 
\cite{BFS1998b,BHLMS2002,Griesemer2004,GLL2001,LHB2011,Matte2016,Panati2009} 
for similar estimates in the ultraviolet regularized case, sometimes
restricted to ground state eigenvectors and to the case where the parameter $r$ appearing
in \eqref{intro1} and \eqref{intro2} is equal to $0$. 
Theorem~\ref{introthmsupexp} also provides improved quantitative
information on superexponential decay compared to the earlier literature \cite{BHLMS2002,LHB2011}, 
where the quantity $(1-\ve)a/(p+1)$ appearing in \eqref{intro2} is replaced by some unspecified, 
sufficiently small positive constant.

Before stating the theorems we introduce more notation and some conventions:
\begin{enumerate}
\item[$\bullet$]
We write $a\wedge b:=\min\{a,b\}$ and $a\vee b:=\max\{a,b\}$ for real numbers~$a$ and~$b$.
\item[$\bullet$]
$\Id\Gamma(\vk)$ denotes the differential second 
quantization of some self-adjoint operator $\vk$ in the Hilbert space for one boson; see 
Subsection~\ref{ssecWeyl} where some elements of Fock space calculus are recalled. 
\item[$\bullet$]  The localization threshold of $\NV$ will be denoted by 
$\Thh\in(-\infty,\infty]$; see Subsection~\ref{ssecexploc} for its precise definition.
\item[$\bullet$]
If we talk about continuous representatives of elements of $L^2(\RR^3,\sF)$, then continuity
is understood in the obvious sense of maps from $\RR^3$ to $\sF$. 
\end{enumerate}
According to \cite{MatteMoeller2017}, elements of the range of $e^{-t\NV}$ with $t>0$ have
continuous representatives; see also Remark~\ref{remcontrange}. 
Hence, the same holds for elements of the range of the spectral
projection $1_{(-\infty,\lambda]}(\NV)$, for any $\lambda\in\RR$. 
 
\begin{thm}[Spatial Exponential Decay of Spectral Subspaces]\label{introthmexpdec}
For all $\lambda,\sigma\in\RR$ with $\lambda<\sigma\le\Thh$ and $\ve,r>0$, we find 
$c>0$ such that the unique continuous representative 
$\Psi(\cdot)$ of any normalized element $\Psi$ in the range of $1_{(-\infty,\lambda]}(\NV)$ satisfies
\begin{align}\label{intro1}
\|e^{r\Id\Gamma(\omega\wedge1)}\Psi(\V{x})\|&\le 
ce^{-(1-\ve)|\V{x}|\sqrt{2\sigma-2\lambda}},\quad\V{x}\in\RR^3.
\end{align}
\end{thm}

\begin{thm}[Spatial Superexponential Decay of Spectral Subspaces]\label{introthmsupexp}
Assume $V$ obeys the lower bound
\begin{align}\label{condVsupexp}
V(\V{x})&\ge\frac{a^2}{2}|\V{x}|^{2p}-b,\quad|\V{x}|\ge \rho,
\end{align}
for some $a,b,p,\rho>0$. Then, for all $\lambda\in\RR$ and $\ve,r>0$, there exists
$c>0$ such that the unique continuous representative 
$\Psi(\cdot)$ of any normalized element $\Psi$ in the range of $1_{(-\infty,\lambda]}(\NV)$ satisfies
\begin{align}\label{intro2}
\|e^{r\Id\Gamma(\omega\wedge1)}\Psi(\V{x})\|
&\le ce^{-(1-\ve)a|\V{x}|^{p+1}/(p+1)},\quad\V{x}\in\RR^3.
\end{align}
\end{thm}

The previous theorems are just two examples for the applicability of our $L^2$-localization
estimates in Proposition~\ref{propexploc}. This proposition also yields non-isotropic localization
estimates in situations where the potential behaves differently in various spatial directions.
We prove our $L^2$-localization estimates by further elaborating on the methods used 
in \cite{BFS1998b,Griesemer2004}. After that we turn them into $L^\infty$-decay estimates 
by means of the Feynman-Kac formula \eqref{FKintro} and the following bound proven
in Proposition~\ref{propGSmu},
\begin{align}\label{intro3}
\sup_{\V{x}\in\RR^3}\sup_{\|\Psi\|\le1}e^{F(\V{x})}\big\|e^{t\Id\Gamma(\omega\wedge1)/6}
\EE[W^V_{t}(\V{x})^*(e^{-F}\Psi)(\V{B}_t^{\V{x}})]\big\|_{\sF}
&\le c_{t}\frac{e^{6tL^2}}{t^{\nf{3}{4}}},
\end{align}
for all $t>0$.
Here $F:\RR^3\to\RR$ is Lipschitz continuous and $L$ is a Lipschitz constant for $F$. The
constant in \eqref{intro3} satisfies $\sup_{t\in(0,n]}c_{t}<\infty$, for all $n\in\NN$; 
besides $t$, it only depends on $\ee$ and $V$. 

As we shall see in Theorem~\ref{introthmsupexpNum} below, the function $\omega\wedge1$ can be
replaced by $1$ in the inequalities \eqref{intro1} and \eqref{intro2}, if we restrict our attention to
ground state eigenvectors (if any).
The question of whether ground state eigenvectors of Hamiltonians in non- or semi-relativistic
quantum field theory are in the domain of inverse Gaussians of field operators has also been
raised in the literature and treated with the help of associated path measures in 
\cite{Hiroshima2014}. The following remark shows that, if the boson wave function in the field operator 
does not contain too many soft boson modes, 
then this question can be answered in greater generality.
(As for ground state eigenvectors, see Remark~\ref{introremGauss2} below.)
\begin{enumerate}
\item[$\bullet$]
Henceforth, $\dom(T)$ denotes the domain of definition of a linear operator $T$. If $T$ is a 
non-negative self-adjoint operator in a Hilbert space, then $\fdom(T)$ is its form domain. 
\item[$\bullet$] The symbol $\vp(h)$ stands for the field operator associated with a
boson wave function $h$; see Subsection~\ref{ssecWeyl} for its precise definition.
\end{enumerate}

\begin{rem}[Gaussian Domination for Spectral Projections]\label{introremGauss1}
Let $\lambda\in\RR$, $r>0$, and $h\in\fdom(\omega^{-1})$ with $\|(\omega^\mh\vee1)h\|<1/2$. 
By virtue of \eqref{FKintro} and \eqref{intro3}, we then find a $(r,\lambda,\ee,V)$-dependent
constant $c>0$ such that the unique continuous representative $\Psi(\cdot)$ of any normalized 
element $\Psi$ in the range of the spectral projection $1_{(-\infty,\lambda]}(\NV)$ satisfies 
$\|e^{r\Id\Gamma(\omega\wedge1)}\Psi(\V{x})\|\le c$, for all $\V{x}\in\RR^3$.
The latter bound can further be combined with the following inequality proven in 
Appendix~\ref{appinvGauss},
\begin{align}\label{intro4}
\|e^{\vp(h)^2}\psi\|\le\frac{1}{\sqrt{1-4\alpha\|(\omega^\mh\vee1)h\|^2}}
\|e^{\Id\Gamma(\omega\wedge1)/(\alpha-1)}\psi\|,
\end{align}
for all $\alpha>1$ such that $4\alpha\|(\omega^\mh\vee1)h\|^2<1$ and
every $\psi\in\dom(e^{\Id\Gamma(\omega\wedge1)/(\alpha-1)})$.
Of course, \eqref{intro4} can also be combined with \eqref{intro1} or \eqref{intro2}.
\end{rem}

Whether the infimum of the spectrum of $\NV$,
\begin{align*}
\EV&:=\inf\sigma(\NV),
\end{align*}
is an eigenvalue or not, i.e., whether there is a ground state eigenvector or not, is clarified in the next 
two theorems, modulo a binding condition in the existence result. The latter condition obviously holds
for confining potentials. In many other relevant cases the binding condition can be verified with the
help of an argument from \cite{GLL2001}; see Subsection~\ref{ssecbinding}.

For sufficiently small $|\ee|$, the existence part of the next theorem could also be inferred from the
arguments given in \cite{HHS2005}. For a class of confining potentials $V$ and massive bosons
($\mu>0$), the existence of ground state eigenvectors was shown in \cite{Ammari2000} and with an
additional ultraviolet regularization in \cite{DerezinskiGerard1999}.
Several articles contain existence proofs for ground states in the case where both a mild infrared and 
an ultraviolet regularization are introduced in the massless Nelson model: Confining potentials 
are treated non-perturbatively in \cite{DerezinskiGerard2004,Gerard2000,Spohn1998}, 
weak coupling assumptions are used in \cite{BFS1998b,Spohn1998}. 
As far as ground states of {\em fiber} Hamiltonians in the renormalized Nelson model are concerned,
perturbative (resp. non-perturbative) results for massive bosons can be found in \cite{Cannon1971}
(resp. \cite{Froehlich1974}), while \cite{BDP2011} deals with the massless model at weak coupling.
Non-perturbative constructions of zero-modes for the ultraviolet regularized model appear
in \cite{BetzSpohn2005,Gross1972}; see also the textbook \cite{LHB2011} and the reference lists
in \cite{BDP2011,LHB2011}.

On the Fock space $\sF$ there exists a canonical notion of positivity defined by means of 
the Schr\"{o}dinger, or, ``$\cQ$-space" representation of $\sF$. (For the unfamiliar reader
we shall sketch a construction of $\cQ$-space in Subsection~\ref{ssecQspace}.)
This notion of positivity also induces a notion of positivity for maps from $\RR^3$ to $\sF$, 
which is employed in the next theorem and henceforth,
in particular when we are dealing with positivity improving or ergodic operators on $L^2(\RR^3,\sF)$.

\begin{thm}[Ground States for the IR Regular Nelson Operator]\label{introthmgsN}
Assume that the binding condition $\Thh>\EV$ and the infrared regularity condition
$\omega^{-3}\eta^2\in L^1_\loc(\RR^3)$ are fulfilled. Then $\EV$ is a non-degenerate eigenvalue
of $\NV$ and there exists a corresponding eigenvector whose unique continuous representative
is strictly positive.  
\end{thm}

\begin{proof}
The assertions on uniqueness and positivity follow directly from the ergodicity
of the semigroup generated by $\NV$ proven in \cite{MatteMoeller2017} and Faris'
Perron-Frobenius theorem \cite{Faris1972}; see also Theorem~\ref{thmPerronFrobenius} below. 
The existence part is contained in Theorem~\ref{thmGSN}. 
\end{proof}

Let us mention that the ergodicity of the semigroup generated by the {\em fiber} Hamiltonians in the
renormalized Nelson model -- with respect to a different notion of positivity on $\sF$ -- 
has been established in \cite{Lampart2021,Miyao2019}.

The infrared regularity condition in the previous theorem cannot be dropped as the next result shows.
In fact, the following theorem establishes non-existence of ground states in the infrared singular 
Nelson model for the first time without any ultraviolet regularization. Furthermore, it does not
require any binding condition or related technical restrictions as they were imposed on the 
potential $V$ in the previous literature 
\cite{DerezinskiGerard2004,GHPS2012,Hirokawa2006,LorincziMinlosSpohn2002,Panati2009}.
(A non-perturbative proof for the non-existence of ground states of infrared singular but
ultraviolet regularized {\em fiber} Hamiltonians in the Nelson model can be found in \cite{Dam2018};
see the references given there for corresponding perturbative results. Absence of ground states for 
renormalized massless fiber Hamiltonians in the Nelson model is proven in \cite{DamHinrichs2021}.)

The following theorem is proved by further elaborating on the argument 
given in \cite{DerezinskiGerard2004}.

\begin{thm}[Absence of Ground States for the IR Singular Nelson Operator]\label{introthmabs}
Consider the massless renormalized one-particle Nelson operator $\NV$ with an arbitrary 
Kato decomposable potential $V$ in the infrared singular case where
\begin{align}\label{introIRsing}
\ee^2\int_{\{|\V{k}|\le 1\}}\frac{\eta(\V{k})^2}{|\V{k}|^3}\Id\V{k}=\infty.
\end{align}
Then $\EV$ is not an eigenvalue of $\NV$.
\end{thm}

\begin{proof}
The assertion is contained in the one of Theorem~\ref{thmabsence}.
\end{proof}

Similarly as in most other related works, we shall first prove the existence of ground states under
several simplifying assumptions, such as strict positivity of the boson mass, which thereupon are
removed in a chain of compactness arguments.
A crucial technical step is to derive a suitable formula for $a(\V{k})\Phi_\iota:=(a\Phi_\iota)(\V{k})$,
with $a$ denoting the ``pointwise annihilation operator" and 
$\{\Phi_\iota\}_{\iota\in I}$ some family of approximate ground
state eigenvectors, which together with the spatial localization estimates reveals compactness
of the family $\{\Phi_\iota\}_{\iota\in I}$. A suchlike formula comprising manifestly square-integrable
functions of $\V{k}$ can be obtained by working with {\em Gross transformed versions} of the involved 
Nelson operators. Here the Gross transformation is a
unitary operator that ceases to exist in the infrared singular situation \eqref{introIRsing}. In the
latter case one thus has to introduce an infrared cutoff in the transformation. After transforming the
Nelson operator one can, however, send the infrared cutoff to zero again and observe 
norm resolvent convergence of the transformed Nelson operators to a self-adjoint operator that we
denote by $\HV$. If \eqref{introIRsing} is satisfied, then we refer to $\HV$ as the 
{\em renormalized Nelson operator in the non-Fock representation}, because it is not unitarily
equivalent to $\NV$ anymore. This nomenclature is reminiscent of the fact that the above limiting
procedure also gives rise to representations of the canonical commutation relations inequivalent
to the Fock representation; see \cite{Arai2001,DerezinskiGerard2004} where related transformations 
are discussed. If the left hand side of \eqref{introIRsing} is finite, then the limiting Gross transformation
exists and intertwines $\NV$ and $\HV$.

As it is given by an explicit quadratic form comprising readily tractable terms, 
we found it convenient to work with $\HV$ as much as possible. 
Here we should mention that recent techniques employing a concept of {interior boundary conditions}
provide formulas also for $\dom(H)$ as well as for the action of $H$ on it, at least when $V$ is
relatively bounded with respect to the Laplace operator with relative bound $<1/2$;
see \cite{LampartSchmidt2019} for massive bosons and \cite{Schmidt2020} for massless ones.
These results can be related to the construction of singular perturbations of selfadjoint operators 
by (here two-fold) applications of Kre\u{\i}n's resolvent formula \cite{Posilicano2020}; 
besides explicit characterizations of $\dom(H)$ and
expressions for $H$, a formula for the difference of the resolvents of $H$ and the non-interacting
Hamiltonian can be found in \cite{Posilicano2020}.

The analysis of the operator $\HV$ used here is, however, physically relevant in its own right since, in special 
cases, $\HV$ describes the interaction via a Bose field of two quantum mechanical matter 
particles, one of them 
having an ``infinite mass" and pinned down at the origin. This situation has been investigated 
before in \cite{HHS2005} where a ground state of the renormalized Nelson operator in the non-Fock 
representation has been shown to exist under a weak coupling condition. We have been able to 
remove this restriction and prove the following two theorems. Before reading them the reader should
note that $\NV$ and $\HV$ have the same localization threshold $\Thh$ and the same 
spectrum, so that in particular $\EV=\inf\sigma(\HV)$. Furthermore, the reader should note
that elements in the range of $e^{-t\HV}$ with $t>0$ have continuous representatives, as we shall
show in Remark~\ref{remcontrange}.

\begin{thm}[Spatial Decay in the Non-Fock Case]\label{introthmdecnonFock}
The statements of Theorem~\ref{introthmexpdec} and Theorem~\ref{introthmsupexp} hold true
without further changes when $\HV$ is put in place of $\NV$.
\end{thm}

\begin{proof}
This theorem is proved at the end of Subsection~\ref{ssecexploc}. (Given Proposition~\ref{propexploc} 
its proof is virtually identical to the ones of Theorem~\ref{introthmexpdec} and Theorem~\ref{introthmsupexp}.) 
\end{proof}

Notice that the next theorem does not require the infrared regularity condition imposed
in Theorem~\ref{introthmgsN}.

\begin{thm}[Existence of Ground States in the Non-Fock Case]\label{introthmgsnonFock}
Assume that the binding condition $\Thh>\EV$ is fulfilled. Then $\EV$ is a non-degenerate 
eigenvalue of $\HV$ and there exists a corresponding eigenvector whose unique continuous 
representative is strictly positive.
\end{thm}

\begin{proof}
The assertion follows from Theorem~\ref{thmPerronFrobenius} and Theorem~\ref{thmGSH}.
\end{proof}

The existence of ground states for ultraviolet regularized Nelson operators in non-Fock 
representations was proven for confining potentials in \cite{Arai2001} and 
under binding conditions, employing exponential $L^2$-localization estimates in
\cite{Panati2009,Sasaki2005}. There is an infrared problem arising in the proof of
Theorem~\ref{introthmgsnonFock} when we look for a suitable 
representation of the expression $(a\Phi_\iota)(\V{k})$ already discussed earlier.
In \cite{BFS1999} the analogous problem is solved for the Pauli-Fierz model by introducing
counter terms related to the Pauli-Fierz transformation in the computations. We shall employ
similar counter terms in the infrared region, exploiting the appearance of minimally coupled field 
operators in ultraviolet regularized versions of $\HV$ after the Gross transformation.

In fact, the construction of $\HV$ via quadratic forms in \cite{HHS2005} already
required a weak coupling condition which was traded for Nelson's assumption 
\cite{Nelson1964} that the infrared cutoff in the Gross transformation be sufficiently large.
The first non-perturbative construction of $\HV$ was achieved in \cite{MatteMoeller2017}.
Since $\HV$ was defined as the generator of a Feynman-Kac semigroup in
\cite{MatteMoeller2017}, we still have to verify that its quadratic form is given by the usual
formulas before we prove Theorem~\ref{introthmdecnonFock} and Theorem~\ref{introthmgsnonFock}. To 
avoid technical explanations at this point we allow ourselves to state the corresponding
theorem in a somewhat vague wording:

\begin{thm}[On Nelson's Renormalization Procedure]\label{introthmrenproc}
With only minor modifications, Nelson's operator theoretic renormalization procedure and its later
improvements for massless bosons \cite{GriesemerWuensch2017,HHS2005} can be carried 
through for every arbitrarily small infrared cutoff in the Gross transformation without any smallness 
assumptions on the matter-radiation coupling. This permits to verify that the quadratic 
form of the renormalized Nelson operator in the non-Fock representation $\HV$ is still given by 
Nelson's formulas.
\end{thm}

\begin{proof}
The statement is formulated precisely and proved in Theorem~\ref{thmrbUV}.
\end{proof}

The ``minor modifications" alluded to in the previous theorem merely consist in putting a 
sufficiently large part of the interaction terms into a comparison operator that plays the role of the free 
Hamiltonian in Nelson's estimates and that can be dealt with by other, non-perturbative means 
\cite{HaslerHerbst2008,Hiroshima2000esa,Hiroshima2002,Matte2017}.

As can be seen from the estimations in Appendix~\ref{apprbUV}, an analogue of 
Theorem~\ref{introthmrenproc} for fiber Hamiltonians holds in the translation invariant 
Nelson model as well; see \cite{Cannon1971,DamHinrichs2021,Lampart2021,MatteMoeller2017} for 
various ways to construct renormalized fiber Hamiltonians in Nelson's model.

We mentioned above that the existence of ground states will be proven for massive bosons first. In 
fact, we do this under the further simplifying assumption that the matter particle be confined to a
bounded open subset $\Geb\subset\RR^3$. The reason for this is that, for bounded $\Geb$ and
massive bosons, the existence of ground states follows immediately from an abstract
sufficient condition due to Gross \cite{Gross1972} and the hypercontractivity estimate in the following 
theorem. Gross applied his criterion to ultraviolet regularized fiber Hamiltonians at zero total
momentum. We felt that it might be worthwhile to demonstrate the usefulness of
his criterion in our situation. In fact, its application is much easier than the discretization 
arguments \cite{Froehlich1974} or the related Fock space localization techniques 
\cite{DerezinskiGerard1999} used earlier. In our setting, Gross' 
criterion works, however, only for one matter particle (no Pauli principle) as it deals with positivity 
preserving semigroups.
  
In the statement of the next theorem, $\NV_{\Geb}$ and $\HV_{\Geb}$ are Dirichlet realizations of
the renormalized Nelson operator and its non-Fock version, respectively, for a matter particle confined 
to $\Geb$. Furthermore, 
$$
\cU_{\Geb}:L^2(\Geb,\sF)\longrightarrow L^2(\Geb,L^2(\cQ,\nu))
=L^2(\Geb\times\cQ,\Id\V{x}\otimes\nu)
$$ 
is the direct $\Id\V{x}$-integral of a unitary transformation onto a $\cQ$-space representation 
$L^2(\cQ,\nu)$ of the Fock space $\sF$. 

\begin{thm}[Hypercontractivity]\label{introthmhypcontr}
Let $\Geb\subset\RR^3$ be an arbitrary open subset and suppose that the boson mass $\mu$ is 
strictly positive. Let $t>0$ and put $p_{\mu,t}:=e^{t\mu/3}+1$. Then there exists
$c>0$, depending only on ${\ee,\mu,t,V,\Geb}$, such that 
\begin{align*}
\|\cU_{\Geb}e^{-t\NV_{\Geb}}\cU_{\Geb}^*\Psi
\|_{L^{p_{\mu,t}}(\Geb\times\cQ,\Id\V{x}\otimes\nu)}
&\le c\|\Psi\|_{L^2(\Geb\times\cQ,\Id\V{x}\otimes\nu)},
\end{align*}
for all $\Psi\in L^2(\Geb\times\cQ,\Id\V{x}\otimes\nu)$.
The same holds with $\HV_{\Geb}$ put in place of $\NV_{\Geb}$.
\end{thm}

\begin{proof}
This theorem is proved together with Theorem~\ref{thmGSmass} at the end of Subsection~\ref{ssecGSmass}.
\end{proof}

Besides Nelson's hypercontractivity bound for differential second quantizations of strictly positive
operators, the proof of the previous theorem relies on the estimates on 
Feynman-Kac integrands in \cite{MatteMoeller2017}.

Before we discuss results on path measures associated with ground states, we remark that the 
compactness arguments employed to remove simplifying assumptions in the construction of ground 
states can also be used to study the $L^2$-continuity of ground states with respect to parameters. 
Exploiting properties of Feynman-Kac semigroups, one can further pass from $L^2$-continuity
in parameters to joint continuity in the spatial variable $\V{x}$ and the parameters. This has already 
been observed in \cite{Matte2016} where ultraviolet regularized models (comprising linearly 
and minimally coupled fields) have been treated. For simplicity, we only choose the coupling constant 
as variable external parameter; see \cite{Matte2016} on how to include variations of the potential $V$.
The following two theorems are non-trivial since, for massless bosons, our models are both infrared 
and ultraviolet singular and ground state eigenvalues are imbedded in the continuous spectrum so
that well-known perturbation theoretic arguments do not apply.

\begin{thm}[Continuity]\label{introthmcont}
Let $I\subset\RR$ be some open interval.
For every $\ee\in I$, let $\NV_{\ee}$ denote
the renormalized Nelson operator with coupling constant $\ee$ and
let $\Thh_\ee$ and $\EV_{\ee}$ stand for its localization threshold and ground state
energy, respectively. Assume the binding condition $\Thh_\ee>\EV_\ee$ holds,
for all $\ee\in I$. Finally, assume the infrared condition
$\omega^{-3}\eta^2\in L^1_\loc(\RR^3)$ is fullfilled.
Let $\Phi_{\ee}$ denote the positive, normalized ground state eigenvector of $\NV_{\ee}$.
Then the following holds:
\begin{enumerate}
\item[{\rm(i)}]\label{conti}
The map $I\ni\ee\mapsto\Phi_{\ee}\in L^2(\RR^3,\sF)$ is continuous. 
\item[{\rm(ii)}]
If $\Phi_{\ee}(\cdot)$ denotes the unique continuous representative of $\Phi_{\ee}$, then the map
$\RR^3\times I\ni(\V{x},\ee)\mapsto\Phi_{\ee}(\V{x})\in\sF$ is continuous.
\end{enumerate}
\end{thm}

\begin{proof}
All assertions follow from Theorem~\ref{thmcontN}.
\end{proof}

\begin{thm}[Continuity, Non-Fock Case]\label{introthmcontnonFock}
For every $\ee\in I$ in the open interval $I\subset\RR$, 
let $\HV_{\ee}$ denote the operator $\HV$ for the special choice $\ee$ of the 
coupling constant. Assume the binding condition $\Thh_\ee>\EV_\ee$ holds, for all $\ee\in I$.
Let $\Phi_{\ee}$ denote the positive, normalized ground state eigenvector of 
$\HV_{\ee}$. Then Statements {\rm(i)} and {\rm(ii)} above hold true.
\end{thm}

\begin{proof}
The statements (i) and (ii) follow from Theorem~\ref{thmL2cont} and Theorem~\ref{thmcont}, respectively.
\end{proof}

Our last results are derived by means of certain path measures associated with the ground states
found in Theorem~\ref{introthmgsN}. Similar measures have been introduced for the ultraviolet 
regularized Nelson model in \cite{BHLMS2002}. They have been further explored in the ultraviolet 
regularized Nelson model and the related non- and semi-relativistic Pauli-Fierz models in
\cite{BetzHiroshima2009,BetzSpohn2005,Hiroshima2004,Hiroshima2014,HiroshimaLorincziTakaesu2012,LorincziMinlosSpohn2002,LorincziMinlosSpohn2002b}. Thanks to the Feynman-Kac formulas
of \cite{MatteMoeller2017} and the ergodicity of the semigroup proven there we are able to construct 
them for the renormalized Nelson model as well. We shall actually introduce a {\em family} of path 
measures associated with the ground state of $\NV$, attaching a separate path measure to each 
position $\V{x}\in\RR^3$, while the aforementioned articles deal, roughly speaking, with integrals over 
such a family.

In fact, let $\V{B}_-$ be a three-dimensional standard Brownian motion independent from $\V{B}$ and 
let $(W^V_{-t}(\V{x}))_{t\ge0}$ denote the operator-valued process appearing in the Feynman-Kac
formula \eqref{FKintro} corresponding to $\V{B}_-$. If $\Phi$ is the continuous representative
of the positive normalized ground state eigenvector of $\NV$, then a Markov property proven in 
\cite{MatteMoeller2017} reveals that the expressions
\begin{align}
m_{t}(\V{x})&:=e^{2t\EV}\SPn{W^V_{-t}(\V{x})^*\Phi(\V{B}_{-t}^{\V{x}})}{W^V_{t}(\V{x})^*
\Phi(\V{B}_{t}^{\V{x}})},\quad t\ge0,
\end{align}
define a strictly positive martingale. Now choose the (completed) Wiener space of continuous paths in 
$\RR^6$ as underlying probability space and let $\EE_{\Wi}$ denote the corresponding
expectation. Furthermore, let $\V{B}$ and $\V{B}_{-}$ stand for the first three
and the last three components, respectively, of the canonical evaluation process on $\Omega_{\Wi}$. 
Finally, let $\mr{\fF}_t$ be the $\sigma$-algebra generated by all $\V{B}_s$ 
and $\V{B}_{-s}$ with $s\in[0,t]$ ({\em not} augmented by the null sets of the completed
Wiener measure). Then, via a standard construction, the relations
\begin{align*}
\mu_{\V{x},*}(A)&:=\frac{1}{\|\Phi(\V{x})\|^2}\EE_{\Wi}[1_Am_{t}(\V{x})],
\quad A\in\mr{\fF}_t,\,t\ge0,
\end{align*}
uniquely determine a probability measure on the Borel-$\sigma$-algebra of $C([0,\infty),\RR^6)$ 
extending the well-defined map $\mu_{\V{x},*}$ . This probability measure, call it 
$\mu_{\V{x},\infty}$, is the said path measure associated with $\Phi$ and $\V{x}$. 

It turns out that expectation values like $\SPn{\Phi(\V{x})}{e^{-z\Id\Gamma(\chi)}\Phi(\V{x})}$
with a compactly supported $\chi:\RR^3\to[0,1]$, which are well-defined {\em a priori} at least for all
$z\in\{\zeta\in\CC|\,\Re[\zeta]\ge0\}$, can be 
represented as integrals with respect to $\mu_{\V{x},\infty}$ having obvious 
analytic extensions to all $z\in\CC$. Combining this observation with Theorems~\ref{introthmexpdec} 
and~\ref{introthmsupexp}, we eventually arrive at the following result:

\begin{thm}[Strong Boson Number Localization in Ground States]\label{introthmsupexpNum}
Assume that the binding condition $\Thh>\EV$ and the infrared regularity condition
$$
d:=\int_{\{|\V{k}|\le1\}}\frac{\eta(\V{k})^2}{\omega(\V{k})^3}\Id\V{k}<\infty
$$
are fulfilled. Let $\Phi$ be the continuous representative of the normalized, strictly positive 
ground state eigenvector of $\NV$. Let $\ve,r>0$ be arbitrary and pick some $\sigma\in\RR$ with 
$\sigma\le\Thh$. Then $\Phi(\V{x})\in\dom(e^{r\Id\Gamma(1)})$, for all $\V{x}\in\RR^3$, and 
there exists $c>0$ such that 
\begin{align}\label{introsupexpNum}
\|e^{r\Id\Gamma(1)}\Phi(\V{x})\|&\le ce^{(e^{4r}-1)\ee^2d/4}
e^{-(1-\ve)|\V{x}|\sqrt{2\sigma-2\EV}},\quad\V{x}\in\RR^3.
\end{align}
If the potential $V$ satisfies \eqref{condVsupexp} (which entails $\Thh=\infty$), then
the second exponential on the right hand side of \eqref{introsupexpNum} can be replaced by
$e^{-(1-\ve)a|\V{x}|^{p+1}/(p+1)}$.
\end{thm}

\begin{proof}
This theorem is proved at the end of Subsection~\ref{ssecsuperexpN}.
\end{proof}

\begin{rem}[Gaussian Domination for Ground States]\label{introremGauss2}
Theorem~\ref{introthmsupexpNum} permits to relax the assumption on $h$ in Remark~\ref{introremGauss1},
if we restrict our attention to ground state eigenvectors instead of considering more general elements
of spectral subspaces. In fact, Lemma~\ref{leminvGauss} also implies the bound
\begin{align*}
\|e^{\vp(h)^2}\psi\|\le\frac{1}{\sqrt{1-4\alpha\|h\|^2}}\|e^{\Id\Gamma(1)/(\alpha-1)}\psi\|,
\quad\psi\in\dom(e^{\Id\Gamma(1)/(\alpha-1)}),
\end{align*}
for all $h\in L^2(\RR^3)$ with $\|h\|<1/2$ and $\alpha>1$ with $4\alpha\|h\|^2<1$, 
which can be combined with \eqref{introsupexpNum}.
\end{rem}

Suppose that $h\in L^2(\RR^3)$ is an element of the completely real subspace
\begin{align}\label{realsubspace}
\mathfrak{r}&:=\big\{f\in L^2(\RR^3)\big|\,\ol{f(\V{k})}=f(-\V{k}),\,\text{a.e. $\V{k}\in\RR^3$}\big\},
\end{align}
in which case $\vp(h)$ is interpreted as a position observable for the radiation field.
Then the condition $\|h\|<1/2$ in the estimation of $\|e^{\vp(h)^2}\Phi(\V{x})\|$ described in the
previous remark cannot be improved, no matter what special properties $h$ might have otherwise. 
This is illustrated by our last main theorem.
The lower bound \eqref{lbGauss} asserted in it is again derived by means of the path measure 
$\mu_{\V{x},\infty}$. Notice that $\Phi(\V{x})\not=0$, for all $\V{x}\in\RR^3$, in \eqref{lbGauss}
and \eqref{notinGauss}.

\begin{thm}[Lower Bound on the Gaussian Domination for Ground States]\label{introthmGaussdomlb}
Under the assumptions of Theorem~\ref{introthmsupexpNum}, consider some $h\in\mathfrak{r}$
and let $\V{x}\in\RR^3$. Then
\begin{align}\label{lbGauss}
\|e^{\vp(h)^2}\Phi(\V{x})\|^2&\ge\frac{\|\Phi(\V{x})\|^2}{\sqrt{1-4\|h\|^2}},\quad
\text{if}\;\,\|h\|<\frac{1}{2},
\\\label{notinGauss}
\Phi(\V{x})&\notin\dom(e^{\vp(h)^2}),\qquad\;\:\text{if}\;\,\|h\|\ge\frac{1}{2}.
\end{align}
\end{thm}

\begin{proof}
This theorem is proved in Subsection~\ref{ssecGaussdom}.
\end{proof}

%%%%%%%%%%%%%%%%%%%%%%%%%%%%%%%%%%%%%%%%%%%%%%%%%

\subsection*{Organization of the Article}

\begin{enumerate}
\item[$\bullet$] 
In the succeeding Section~\ref{secdefs} we shall introduce the Nelson model and explain
the Feynman-Kac formulas found in \cite{MatteMoeller2017} in detail. In particular, we clarify and 
prove the assertions made in Theorem~\ref{introthmrenproc}, deferring the most technical steps to the
appendix. 
\item[$\bullet$] 
The existence of ground states is addressed in Section~\ref{secGS}. Along the way we 
derive our hypercontrativity bound as well as the decay estimates stated in
Theorems~\ref{introthmexpdec},~\ref{introthmsupexp}, and~\ref{introthmdecnonFock},
again deferring most of the technical work to the appendix. Section~\ref{secGS} concludes with a
discussion of continuity properties of ground states. 
\item[$\bullet$] 
We prove our results on absence
of ground states in Section~\ref{secabsence}. 
\item[$\bullet$] 
Finally, we study path measures associated with
ground states in Section~\ref{secGibbs}. 
\item[$\bullet$] 
As indicated above,
the main text is followed by several appendices containing more technical material.
\end{enumerate}

%%%%%%%%%%%%%%%%%%%%%%%%%%%%%%%%%%%%%%%%%%%%%%%%%
%%%%%%%%%%%%%%%%%%%%%%%%%%%%%%%%%%%%%%%%%%%%%%%%%
%%%%%%%%%%%%%%%%%%%%%%%%%%%%%%%%%%%%%%%%%%%%%%%%%

\section{Definitions and Preliminary Results}\label{secdefs}

\noindent
In this section we bring together all prerequisites necessary to derive the results summarized in
the introduction. In Subsections~\ref{ssecWeyl} through~\ref{ssecGross} we discuss, respectively,
some basic bosonic Fock space calculus, $\cQ$-space representations of the Fock space, the
Schr\"{o}dinger operator corresponding to $V$, ultraviolet regularized Nelson operators, and 
Gross transformations. The renormalized Nelson operator and its non-Fock version are introduced
in Subsection~\ref{ssecUVren}. In the final Subsection~\ref{ssecFK} we explain the Feynman-Kac formulas
derived in \cite{MatteMoeller2017}.

%%%%%%%%%%%%%%%%%%%%%%%%%%%%%%%%%%%%%%%%%%%%%%%%%

\subsection{Fock Space and Weyl Representation}\label{ssecWeyl}

\noindent
In what follows we present some well-known material on bosonic Fock spaces and 
corresponding Weyl representations one should have in mind while reading this article; see, e.g., 
\cite{Parthasarathy1992} for a textbook exposition of these matters.

\subsubsection*{Fock Space}
The bosonic Fock space modeled over the Hilbert space for a single boson, $L^2(\RR^3)$, is the 
countable direct sum of ``$n$-particle subspaces" given by
\begin{align}\label{defFock}
\sF:=\CC\oplus\bigoplus_{n\in\NN}L_{\mathrm{sym}}^2(\RR^{3n}).
\end{align}
Here $L_{\mathrm{sym}}^2(\RR^{3n})$ denotes the closed subspace in $L^2(\RR^{3n})$ of all its
elements $\psi^{(n)}$ satisfying
$$
\psi^{(n)}(\V{k}_{\pi(1)},\ldots,\V{k}_{\pi(n)})=\psi^{(n)}(\V{k}_1,\ldots,\V{k}_{n}),
$$
for a.e. $(\V{k}_1,\ldots,\V{k}_{n})\in(\RR^3)^n$ an every permutation $\pi$ of
$\{1,\ldots,n\}$. 

\subsubsection*{Exponential Vectors}
The Fock space contains the exponential vectors
\begin{align*}
\epsilon(h):=(1,h,\ldots,(n!)^\mh h^{\otimes_n},\ldots\:)\in\sF,\quad h\in L^2(\RR^3),
\end{align*}
which are convenient for introducing the Weyl representation and computations.
In the previous formula $h^{\otimes_n}(\V{k}_1,\ldots,\V{k}_{n}):=\prod_{j=1}^nh(\V{k}_j)$,
for a.e. $(\V{k}_1,\ldots,\V{k}_{n})\in(\RR^3)^n$.
The set of all exponential vectors is total in $\sF$ and the map
$L^2(\RR^3)\ni h\mapsto\epsilon(h)\in\sF$ is analytic.

\subsubsection*{Weyl Representation}
Let $\sU(\sK)$ denote the set of unitary operators on a complex Hilbert space $\sK$
equipped with the topology associated with strong convergence of operators.
For all $f,h\in\ L^2(\RR^3)$ and $U\in\sU(L^2(\RR^3))$, we now set
\begin{align}\label{defWeyl}
\sW(f,U)\epsilon(h)&:=e^{-\|f\|^2/2-\SPn{f}{Uh}}\epsilon(f+Uh).
\end{align}
Together with a linear and isometric extension this prescription defines a unitary operator
on $\sF$, again denoted by the symbol $\sW(f,U)$. The so-obtained map
$L^2(\RR^3)\times\sU(L^2(\RR^3))\ni(f,U)\mapsto\sW(f,U)\in\sU(\sF)$, called the
Weyl representation, is strongly continuous. If $U=\id$, then we typically write
$\sW(f):=\sW(f,\id)$ for short. Furthermore, $\Gamma(U):=\sW(0,U)$.

The map $\sW$ is indeed a projective representation of the semi-direct product of
$L^2(\RR^3)$ with $\sU(L^2(\RR^3))$. More precisely, we have the following Weyl relations,
\begin{align}\label{Weylrel}
\sW(f_1,U_1)\sW(f_2,U_2)=e^{-i\Im\SPn{f_1}{U_1f_2}}
\sW(f_1+U_1f_2,U_1U_2),
\end{align}
for all $f_1,f_2\in L^2(\RR^3)$ and $U_1,U_2\in\sU(L^2(\RR^3))$.

\subsubsection*{Field Operators}
Let $f\in L^2(\RR^3)$. Then the above remarks imply that $\RR\ni t\mapsto\sW(-itf)$
is a strongly continuous unitary group. Its self-adjoint generator is called the field operator
associated with $f$ and denoted by $\vp(f)$.

\subsubsection*{Differential Second Quantizations}
Likewise, if $\vk$ is a self-adjoint multiplication operator in $L^2(\RR^3)$, then the generator of
the strongly continuous unitary group $\RR\ni t\mapsto\Gamma(e^{-it\vk})$ 
is called the (differential) second quantization of $\vk$ and denoted by $\Id\Gamma(\vk)$. For instance,
$h\in\dom(\vk)$ implies $\epsilon(h)\in\dom(\Id\Gamma(\vk))$ and
\begin{align}\label{dGexpv}
\Id\Gamma(\vk)\epsilon(h)=\epsilon'(h)\vk h.
\end{align}
In fact, every $n$-particle subspace in the direct sum \eqref{defFock} is a reducing subspace
of $\Id\Gamma(\vk)$. The restriction of $\Id\Gamma(\vk)$ to $L_{\mathrm{sym}}^2(\RR^{3n})$ is
equal to the maximal operator of multiplication with the function
$(\V{k}_1,\ldots,\V{k}_n)\mapsto\vk(\V{k}_1)+\dots+\vk(\V{k}_n)$. 
Thus, if $\vk$ is semi-bounded from
above, then $\epsilon(h)\in\dom(e^{\Id\Gamma(\vk)})$ and
\begin{align}\label{SGdGamma}
e^{\Id\Gamma(\vk)}\epsilon(h)&=\epsilon(e^\vk h),\quad h\in L^2(\RR^3).
\end{align}

\subsubsection*{Creation and Annihilation Operators, Differentiation Formulas}
The creation and annihilation operators associated with $f\in L^2(\RR^3)$ are, respectively, defined by
\begin{align*}
a(f)\epsilon(h)&:=\SPn{f}{h}\epsilon(h)\quad\text{and}\quad\ad(f)\epsilon(h):=\epsilon'(h)f,
\quad h\in L^2(\RR^3),
\end{align*}
followed by linear and closed extensions. It turns out that $a(f)^*=\ad(f)$, $\ad(f)^*=a(f)$, and
$\vp(f)$ is the closure of $\ad(f)+a(f)$. Hence, if $h\in L^2(\RR^3)$ and ${\mathfrak{r}}$ is the
completely real subspace given by \eqref{realsubspace}, then the map 
${\mathfrak{r}}\ni f\mapsto\sW(f)\epsilon(h)\in\sF$ has a 
Fr\'{e}chet derivative, whose action on the tangent vector $v\in{\mathfrak{r}}$ is given by
\begin{align*}
\Id_f\sW(f)\epsilon(h)v&=(\ad(v)-\SPn{v}{f+h})\sW(f)\epsilon(h)=-i\vp(iv)\sW(f)\epsilon(h).
\end{align*}
In particular, the following chain rule holds, for every differentiable $\beta:\RR^3\to{\mathfrak{r}}$,
\begin{align}\label{karl}
\partial_{x_j}\sW(\beta(\V{x}))\epsilon(h)&=-i\vp(i\partial_{x_j}\beta(\V{x}))\sW(\beta(\V{x}))\epsilon(h),
\quad\V{x}\in\RR^3.
\end{align}

\subsubsection*{Relative Bounds and Commutation Relations}
Assume that $\vk$ is a maximal, non-negative, and invertible multiplication operator
in $L^2(\RR^3)$ and let $\psi\in\fdom(\Id\Gamma(\vk))$. Then the maps $f\mapsto a(f)\psi$, 
$f\mapsto\ad(f)\psi$, and $f\mapsto\vp(f)\psi$ are well-defined and real linear from $\dom(\vk^\mh)$ 
into $\sF$, and we have the relative bounds
\begin{align}\label{rba}
\|a(f)\psi\|&\le\|\vk^\mh f\|\|\Id\Gamma(\vk)^\eh\psi\|,
\\\label{rbad}
\|\ad(f)\psi\|&\le\|(\vk^\mh\vee1)f\|\|(1+\Id\Gamma(\vk))^\eh\psi\|,
\\\label{rbvp}
\|\vp(f)\psi\|&\le2\|(\vk^\mh\vee1)f\|\|(1+\Id\Gamma(\vk))^\eh\psi\|.
\end{align}
Finally, if $f,g,h\in\dom(\vk^\mh)$, then $a(f)$, $\ad(f)$, and $\vp(f)$ map $\dom(\Id\Gamma(\vk))$
into $\fdom(\Id\Gamma(\vk))$ and the Weyl relations entail the following commutation relations,
\begin{align}\label{CCR1}
[a(f),\ad(g)]\psi&=[a(f),\vp(g)]\psi=[\vp(f),\ad(g)]\psi=\SPn{f}{g}\psi,
\\
[a(f),a(g)]\psi&=[\ad(f),\ad(g)]\psi=0,\quad\psi\in\dom(\Id\Gamma(\vk)).
\end{align}

\subsubsection*{Vector Notation}
Again let $\vk$ be a maximal, non-negative, and invertible multiplication operator
in $L^2(\RR^3)$ and $\psi\in\fdom(\Id\Gamma(\vk))$.
For vectors $\V{f}=(f_1,\ldots,f_n)$ and $\V{\phi}=(\phi_1,\ldots,\phi_n)$ with 
$f_1,\ldots,f_n\in\dom(\vk^\mh)$ and $\phi_1,\ldots,\phi_n\in\fdom(\Id\Gamma(\vk))$, 
we shall use the shorthands
\begin{align}\label{vecnot0}
\vp(\V{f})\psi:=\big(\vp(f_1)\psi,\ldots,\vp(f_n)\psi\big),\quad
\vp(\V{f})\cdot\V{\phi}:=\sum_{j=1}^n\vp(f_j)\phi_j,
\end{align}
and their analogues for $a$ and $\ad$. We combine this with the notation
\begin{align}\label{vecnot1}
\|\V{v}\|^2:=\|v_1\|^2+\dots+\|v_n\|^2,\quad
\SPn{\V{v}}{\V{w}}:=\SPn{v_1}{w_1}+\dots+\SPn{v_n}{w_n},
\end{align}
where $\V{v}=(v_1,\ldots,v_n)$ and $\V{w}=(w_1,\ldots,w_n)$ are tuples comprised of 
elements of a fixed Hilbert space. For instance,
\begin{align}\label{vecnot2}
\|a(\V{f})\psi\|^2:=\sum_{j=1}^n\|a(f_j)\psi\|^2,\quad
\SPn{\ad(\V{f})\psi}{a(\V{f})\psi}=\sum_{j=1}^n\SPn{\ad(f_j)\psi}{a(f_j)\psi}.
\end{align}

%%%%%%%%%%%%%%%%%%%%%%%%%%%%%%%%%%%%%%%%%%%%%%%%%

\subsection{Schr\"{o}dinger Representation of Fock Space ($\cQ$-Space)}\label{ssecQspace}

\noindent
We shall quite substantially make use of the fact that one can unitarily map the bosonic
Fock space $\sF$ onto an $L^2$-space defined by a probability measure, such that
the field operators corresponding to a certain completely real subspace $\mathfrak{r}$
of the one-boson space turn into maximal operators of multiplication with elements of a 
Gaussian process indexed by $\mathfrak{r}$. To introduce the corresponding notation and to 
shed some light on this unitary transformation for the non-expert reader, we briefly explain 
one canonical possibility to construct it. More details on the construction sketched below can be 
found, e.g., in \cite{Berezanskii1986}. For discussions of $\cQ$-space in the context of quantum
field theory we refer to \cite{LHB2011,Simon1974}.

Recall the definition of the completely real subspace ${\mathfrak{r}}$ in \eqref{realsubspace}. 
In view of the Weyl relations \eqref{Weylrel} the corresponding set of Weyl operators
$\{\sW(f):f\in{\mathfrak{r}}\}$ generates a commutative unital sub-$C^*$-algebra of $\LO(\sF)$ 
that we call $\sA$.
Let ${\sf G}:\sA\to C(\cQ)$ denote the corresponding Gelfand $*$-isomorphism onto the continuous
functions on the maximal ideal space $\cQ$ of $\sA$, which is a compact Hausdorff space.
Let $\fQ$ denote the Borel-$\sigma$-algebra of $\cQ$ and let $\GM:\fQ\to[0,1]$ be the unique Borel 
probability measure representing the positive and normalized linear functional
\begin{align*}
\ell:C(\cQ)\longrightarrow\CC,\quad F\longmapsto\ell(F)
:=\SPn{\epsilon(0)}{{\sf G}^{-1}(F)\epsilon(0)},
\end{align*}
according to Riesz' representation theorem. Since the complex linear span of all
exponential vectors $\{\epsilon(h):h\in{\mathfrak{r}}\}$ is dense in $\sF$ and since
$\sW(h)\epsilon(0)$ is colinear to $\epsilon(h)$, for all $h\in{\mathfrak{r}}$, one can verify that
a complex linear and isometric extension of the prescription
\begin{align*}
\cU\sW(h)\epsilon(0):={\sf G}(\sW(h)),\quad h\in{\mathfrak{r}},
\end{align*}
yields a unitary operator $\cU:\sF\to L^2(\cQ,\GM)$. In particular, $\cU\epsilon(0)=1$.
If we set $\phi(f):=\cU\vp(f)\cU^*$, then $\{\phi(f):f\in{\mathfrak{r}}\}$ is indeed a Gaussian process,
\begin{align*}
\int_\cQ e^{-it\phi(f)}\Id\GM=\SPn{\epsilon(0)}{e^{-it\vp(f)}\epsilon(0)}=
\SPn{\epsilon(0)}{\sW(-itf)\epsilon(0)}=e^{-t^2\|f\|^2/2},\quad f\in{\mathfrak{r}}.
\end{align*}

%%%%%%%%%%%%%%%%%%%%%%%%%%%%%%%%%%%%%%%%%%%%%%%%%

\subsection{Scalar- and Vector-Valued Schr\"{o}dinger Operators}

\noindent
Next, we introduce the Hamilton operator for the matter particle alone, in absence of the
quantized radiation field. It is given by a Schr\"{o}dinger operator with the possibly quite
singular Kato decomposable potential $V$. For technical reasons we shall actually define 
Dirichlet realizations of Schr\"{o}dinger operators on open subsets of $\RR^3$.

In the whole article 
\begin{align*}
\Geb\subset\RR^3\;\text{is open and non-empty.}
\end{align*} 
We define a corresponding {\em minimal} quadratic form $\mathfrak{s}_{\Geb}^V$ 
in the Hilbert space $L^2(\Geb)$ as follows: First, we introduce the {\em maximal} form by
\begin{align*}
\mathfrak{s}_{\Geb,\max}^+[h]&:=\frac{1}{2}\int_{\Geb}|\nabla h(\V{x})|^2\Id\V{x}
+\int_{\Geb}V_+(\V{x})|h(\V{x})|^2\Id\V{x},
\end{align*}
for all $h\in\dom(\mathfrak{s}_{\Geb,\max}^+):= W^{1,2}(\Geb)\cap\fdom(V_+\restr_{\Geb})$.
As a sum of two non-negative closed forms this maximal form is non-negative and closed as well.
Then we let $\mathfrak{s}_{\Geb}^+$ denote the closure of 
$\mathfrak{s}_{\Geb,\max}^+\restr_{C_0^\infty(\Geb)}$ and set
\begin{align*}
\mathfrak{s}_{\Geb}[h]&:=\mathfrak{s}_{\Geb}^+[h]-\int_{\Geb}V_-(\V{x})
|h(\V{x})|^2\Id\V{x},\quad h\in\dom(\mathfrak{s}_{\Geb}^+).
\end{align*}
Since $V_-$ is infinitesimally form bounded with respect to the 
negative Laplacian on $\RR^3$, \cite{AizenmanSimon1982}, it follows that the restriction of
$V_-$ to $\Geb$ is infinitesimally form bounded with respect to $\mathfrak{s}_{\Geb}^+$.
Hence, $\mathfrak{s}_{\Geb}$ is semi-bounded and closed. The self-adjoint operator representing 
$\mathfrak{s}_\Geb$ will be denoted by $\SGV$. It is the Dirichlet realization of
the Schr\"{o}dinger operator on $\Geb$ with potential $V$.

We also define a vector-valued version of the Dirichlet-Schr\"{o}dinger operator acting in 
$L^2(\Geb,\sF)$. To this end we first
recall that $\Psi\in L^1_\loc(\Geb,\sF)$ is said to have weak partial derivatives if, for every
$j\in\{1,2,3\}$, we can find a (necessarily unique) $\Upsilon_j\in L^1_\loc(\Geb,\sF)$ such that
\begin{align*}
\int_\Geb\partial_{x_j}\bar{g}(\V{x})\SPn{\phi}{\Psi(\V{x})}\Id\V{x}&=
-\int_{\Geb}\bar{g}(\V{x})\SPn{\phi}{\Upsilon_j(\V{x})}\Id\V{x},
\end{align*}
for all $g\in C_0^\infty(\Geb)$ and all $\phi$ in some total subset of $\sF$. In the affirmative case
we write $\partial_{x_j}\Psi:=\Upsilon_j$. In complete analogy to the scalar case,
$W^{1,2}(\Geb,\sF)$ is the space of all $\Psi\in L^2(\Geb,\sF)$ having weak
partial derivatives $\partial_{x_1}\Psi$, $\partial_{x_2}\Psi$, and $\partial_{x_3}\Psi$ that
belong to $L^2(\Geb,\sF)$ as well. 

Mimicking the construction in the scalar case we now define a maximal form by
\begin{align*}
\mathfrak{t}_{\Geb,\max}^+[\Psi]&:=\frac{1}{2}\int_{\Geb}\|\nabla\Psi(\V{x})\|^2\Id\V{x}
+\int_{\Geb}V_+(\V{x})\|\Psi(\V{x})\|^2\Id\V{x},
\end{align*}
for all $\Psi\in W^{1,2}(\Geb,\sF)$ such that the second integral in the above formula is finite;
here we use the notation \eqref{vecnot1}. Writing
\begin{align}\label{defsDsE}
\sD(\Geb,\sE)&:=\mathrm{span}\big\{g\phi\big|\,g\in C_0^\infty(\Geb),\,\phi\in\sE\big\},
\quad\sE\subset\sF,
\end{align}
and denoting 
\begin{align}\label{deftGplus}
\mathfrak{t}_\Geb^+&:=\ol{\mathfrak{t}_{\Geb,\max}^+\restr_{\sD(\Geb,\sF)}},
\end{align}
we finally set
\begin{align}\label{defvvSchroedinger}
\mathfrak{t}_\Geb[\Psi]:=\mathfrak{t}_\Geb^+[\Psi]-\int_{\Geb}V_-(\V{x})\|\Psi(\V{x})\|^2\Id\V{x},
\quad\Psi\in\dom(\mathfrak{t}_\Geb^+).
\end{align}
As explained in \cite[\textsection4]{Matte2017}, the form $\mathfrak{t}_\Geb$ is again semi-bounded
and closed. The vector-valued Schr\"{o}dinger operator representing its closure 
will be denoted by $\TV_\Geb$.

%%%%%%%%%%%%%%%%%%%%%%%%%%%%%%%%%%%%%%%%%%%%%%%%

\subsection{The Nelson Operator with Ultraviolet Cutoff}\label{ssecNelson}

\noindent
Next, we introduce the Nelson operator with an ultraviolet cutoff matter-radiation interaction.
It shall be convenient to denote the identity map on $\RR^3$ by
\begin{align*}
\V{m}(\V{k})&=(m_1(\V{k}),m_2(\V{k}),m_3(\V{k})):=\V{k},\quad\V{k}\in\RR^3,
\end{align*}
when it is interpreted as the momentum operator of the bosons in Fourier space. Furthermore,
the following notation for free waves will be convenient,
\begin{align*}
e_{\V{x}}:=e^{-i\V{m}\cdot\V{x}},\quad\V{x}\in\RR^3.
\end{align*}
For  infrared and ultraviolet cutoff parameters $0\le K\le\Lambda\le\infty$, we now define
\begin{align}\label{deffbeta}
f_{\Lambda}:=\ee\omega^\mh\eta1_{\{|\V{m}|\le\Lambda\}},
\quad\beta_{K,\Lambda}:=(\omega+\V{m}^2/2)^{-1}f_{\Lambda}1_{\{|\V{m}|>K\}}.
\end{align}
Here and henceforth $1_A$ denotes the characteristic function of a set $A$.

Notice that $f_{\infty}$ is locally square-integrable but not in $L^2(\RR^3)$, if
$\ee\eta$ is constant and non-zero near infinity. For strictly positive $K$, the function $\beta_{K,\infty}$
is in $L^2(\RR^3)$, while $\beta_{0,\Lambda}$ might have a non-square-integrable singularity at 
$0$, if the boson mass $\mu$ is zero. 

In our definition of the Nelson Hamiltonian we shall add the energy renormalization right away
which, for all $0\le\Lambda<\infty$, is given by
\begin{align*}
E_{\Lambda}^{\ren}&:=\int_{\RR^3}f_{\Lambda}(\V{k})\beta_{0,\Lambda}(\V{k})\Id\V{k}.
\end{align*}
Notice that $E_{\Lambda}^{\ren}$ diverges logarithmically as $\Lambda\to\infty$, if $\ee\eta$ is
constant and non-zero near infinity.

The Nelson Hamiltonian on $\Geb$ with ultraviolet cutoff
$\Lambda\in[0,\infty)$, denoted $\NV_{\Geb,\Lambda}$, is the unique self-adjoint operator 
representing the semi-bounded, closed form $\nv_{\Geb,\Lambda}$ defined on the domain
\begin{align}\label{defQGV}
\QGV&:=\dom(\mathfrak{t}_\Geb^+)\cap L^2(\Geb,\fdom(\Id\Gamma(\omega)))
\end{align}
by the formula
\begin{align}\nonumber
\nv_{\Geb,\UV}[\Psi]&:=\mathfrak{t}_\Geb[\Psi]+
\int_{\Geb}\|\Id\Gamma(\omega)^\eh\Psi(\V{x})\|^2\Id\V{x}
\\\label{Nform}
&\quad+\int_{\Geb}\SPn{\Psi(\V{x})}{\vp(e_{\V{x}}f_{\UV})\Psi(\V{x})}\Id\V{x}
+E_{\UV}^{\ren}\|\Psi\|^2,\quad\Psi\in{\QGV}.
\end{align}
The above form is indeed semi-bounded and closed, as the first line of the right hand side
of \eqref{Nform} contains a sum of two semi-bounded closed forms and, by \eqref{rbvp}, the second 
line is infinitesimally bounded with respect to the first. In fact, 
\begin{align}\label{OpformelNV}
(\NV_{\Geb,\UV}\Psi)(\V{x})&=(\TV_{\Geb}\Psi)(\V{x})+\Id\Gamma(\omega)\Psi(\V{x})
+\vp(e_{\V{x}}f_{\UV})\Psi(\V{x})+E_{\UV}^{\ren}\Psi(\V{x}),
\end{align}
for a.e. $\V{x}\in\Geb$ and all $\Psi\in\dom(\NV_{\Geb,\Lambda})=
\dom(\TV_\Geb)\cap L^2(\Geb,\dom(\Id\Gamma(\omega)))$; see, e.g.,
\cite[Lem.~5.1(1) and Rem.~5.8]{Matte2017}.

Defining the comparison form
\begin{align}\label{compform}
\cfG[\Psi]&:=\mathfrak{t}_{\Geb}^+[\Psi]+
\int_{\Geb}\|\Id\Gamma(\omega)^\eh\Psi(\V{x})\|^2\Id\V{x},
\end{align}
for all $\Psi\in{\QGV}$, we also have the formula
\begin{align*}
\nv_{\Geb,\UV}[\Psi]&=\cfG[\Psi]-\int_{\Geb}V_-(\V{x})\|\Psi(\V{x})\|^2\Id\V{x}
+\int_{\Geb}\SPn{\Psi(\V{x})}{\vp(e_{\V{x}}f_{\UV})\Psi(\V{x})}\Id\V{x}+E_{\UV}^{\ren}\|\Psi\|^2,
\end{align*}
as well as the relative bounds
\begin{align}\label{rbcfnv}
\frac{1}{c}\cfG[\Psi]-c\|\Psi\|^2&\le\nv_{\Geb,\UV}[\Psi]\le c\cfG[\Psi]+c\|\Psi\|^2,
\end{align}
where $c>0$ depends only on $\UV\in[0,\infty)$, $\ee$, and $V$. The latter bounds follow
again from \eqref{rbvp} and the fact that $V_-$ is infinitesimally form-bounded with respect
to the Laplacian.

%%%%%%%%%%%%%%%%%%%%%%%%%%%%%%%%%%%%%%%%%%%%%%%%%

\subsection{The Gross Transformation}\label{ssecGross}

\noindent
To define ultraviolet renormalized operators we follow Nelson \cite{Nelson1964} and introduce
a Gross transformation $G_{K,\UV}\in\LO(L^2(\Geb,\sF))$, for all $0\le K\le\UV\le\infty$ with 
$\beta_{K,\UV}\in L^2(\RR^3)$. This is the unitary operator given by
\begin{align}\label{defGrosstrafo}
(G_{K,\UV}\Psi)(\V{x})&:=\sW(e_{\V{x}}\beta_{K,\UV})\Psi(\V{x}),
\quad\text{a.e. $\V{x}\in\Geb$},
\end{align}
for every $\Psi\in L^2(\Geb,\sF)$. Note that $G_{K,\UV}$ depends on the open subset $\Geb$
of $\RR^3$, which is not displayed in the notation since $G_{K,\UV}$ is always 
given by the same formula. According to the discussion following \eqref{deffbeta},
$G_{K,\UV}$ is defined only for strictly positive $K$, if $\mu=0$ and 
$\ee^2|\V{m}|^{-3}\eta^2$ is not in $L^1_\loc(\RR^3)$.

For all finite $\UV$ but without any restrictions on $K\in[0,\UV]$, we shall now construct 
another quadratic form $\hv_{\Geb,K,\Lambda}$. Later on we shall verify that it is indeed 
the Gross transformed Nelson form provided that, in addition, $\beta_{K,\UV}\in L^2(\RR^3)$ holds.
Let $0\le K\le\UV<\infty$ and $\Psi\in{\QGV}\subset W^{1,2}(\RR^3,\sF)$. Then we define
\begin{align}\label{deftLambdax}
\mathfrak{t}_{K,\UV,\V{x}}[\Psi]&:=\frac{1}{2}\|\nabla\Psi(\V{x})
+i\vp(e_{\V{x}}\V{m}\beta_{K,\UV})\Psi(\V{x})\|^2+V(\V{x})\|\Psi(\V{x})\|^2,
\end{align}
and
\begin{align}\nonumber
\hv_{K,\UV,\V{x}}[\Psi]
&:=\mathfrak{t}_{K,\UV,\V{x}}[\Psi]+\|\Id\Gamma(\omega)^\eh\Psi(\V{x})\|^2
+\Big(E_K^{\ren}-\frac{1}{2}\|\V{m}\beta_{K,\UV}\|^2\Big)\|\Psi(\V{x})\|^2
\\\label{defhGx}
&\quad+\frac{1}{2}\SPn{\Psi(\V{x})}{\vp(e_{\V{x}}\V{m}^2\beta_{K,\UV})\Psi(\V{x})}
+\SPn{\Psi(\V{x})}{\vp(e_{\V{x}}f_{K})\Psi(\V{x})},
\end{align}
for a.e. $\V{x}\in\Geb$, as well as
\begin{align*}
\hv_{\Geb,K,\UV}[\Psi]&:=\int_{\Geb}\hv_{K,\UV,\V{x}}[\Psi]\Id\V{x}.
\end{align*}

\begin{rem}\label{remcompare}
 Let $0\le K\le\UV<\infty$. Then well-known arguments 
(see, e.g., \cite[Rem.~3.1 and \textsection4, in particular Lem.~4.5]{Matte2017}) imply the
following relative bounds with respect to the comparison form $\cfG$ defined in \eqref{compform},
\begin{align}\label{fbUV}
\frac{1}{c_\UV}\cfG[\Psi]-c_\UV\|\Psi\|^2&\le\hv_{\Geb,K,\UV}[\Psi]
\le{c_{\UV}}\cfG[\Psi]+c_{\UV}\|\Psi\|^2,
\end{align}
for all $\Psi\in{\QGV}$, where the constant $c_\UV>0$ depends only on $\UV$, $\ee$, and $V$.
\end{rem}

In view of \eqref{fbUV} the form $\hv_{\Geb,K,\UV}$ is semi-bounded and closed on its domain 
${\QGV}$. Denoting the unique self-adjoint operator representing $\hv_{\Geb, K,\UV}$
by $\HV_{\Geb,K,\UV}$, we have the following standard result. For the reader's convenience, 
we provide a proof of the next proposition in Appendix~\ref{appGrosstrafo}.

\begin{prop}\label{propGrosstrafo}
Let $0\le K\le\UV<\infty$ be such that $\beta_{K,\UV}\in L^2(\RR^3)$. Then
$G_{K,\UV}$ and $G_{K,\UV}^*$ map ${\QGV}$ into itself and
\begin{align}\label{trafoNelsonform}
\nv_{\Geb,\UV}[G_{K,\UV}^*\Psi]=\hv_{\Geb,K,\UV}[\Psi],
\quad\Psi\in{\QGV}.
\end{align}
In particular, 
\begin{align}\label{trafoNelsonHamUV}
G_{K,\Lambda}\NV_{\Geb,\UV}G_{K,\UV}^*=\HV_{\Geb,K,\UV}.
\end{align}
\end{prop}

\begin{prop}\label{propdomHKLV}
Let $0\le K\le\UV<\infty$. Then the Hamiltonian $\HV_{\Geb,K,\UV}$ has the domain
\begin{align*}
\dom(\HV_{\Geb,K,\UV})&=\dom(\TV_\Geb)\cap L^2(\Geb,\dom(\Id\Gamma(\omega))),
\end{align*}
and its action on $\Psi\in\dom(\HV_{\Geb,K,\UV})$ is given by
\begin{align}\nonumber
(\HV_{\Geb,K,\UV}\Psi)(\V{x})&=
(\TV_\Geb\Psi)(\V{x})+\Id\Gamma(\omega)\Psi(\V{x})
-i\vp(e_{\V{x}}\V{m}\beta_{K,\UV})\cdot\nabla\Psi(\V{x})
\\\nonumber
&\quad+\frac{1}{2}\vp(e_{\V{x}}\V{m}\beta_{K,\UV})^2\Psi(\V{x})+\vp(e_{\V{x}}f_{K})\Psi(\V{x})
\\\label{forHKLV}
&\quad+a(e_{\V{x}}\V{m}^2\beta_{K,\UV})\Psi(\V{x})
-\frac{1}{2}\|\V{m}\beta_{K,\UV}\|^2\Psi(\V{x}),\quad\text{a.e. $\V{x}\in\Geb$.}
\end{align}
If $\sC$ is a core for the Schr\"{o}dinger operator on scalar functions $\SGV$ and
$\sD$ is a core for $\Id\Gamma(\omega)$, then an operator core for $\HV_{\Geb,K,\UV}$ is given by
$\mathrm{span}\{g\phi|\,g\in\sC,\phi\in\sD\}$.
Furthermore, $\TV_\Geb(g\phi)=(\SGV g)\phi$, for all $g\in\sC$ and $\phi\in\sD$.
\end{prop}

\begin{proof}
The assertions on the domain and operator cores of $\HV_{\Geb,K,\UV}$ 
are special cases of \cite[Thm.~5.7 and Rem.~5.8]{Matte2017}.
(For a smaller class of potentials $V$ and $\Geb=\RR^3$, these results also follow from 
\cite{HaslerHerbst2008,Hiroshima2000esa,Hiroshima2002}.) Of course, the last assertion
is standard; see, e.g., \cite[Lem.~4.2]{Matte2017}.
The formula \eqref{forHKLV} is a direct consequence of \cite[Prop.~5.2(1)]{Matte2017} and the relation
\begin{align}\label{fiberXX}
\frac{i}{2}\vp(ie_{\V{x}}\V{m}^2\beta_{K,\Lambda})\psi
+\frac{1}{2}\vp(e_{\V{x}}\V{m}^2\beta_{K,\Lambda})\psi
&=a(e_{\V{x}}\V{m}^2\beta_{K,\Lambda})\psi,
\end{align}
valid for all $\V{x}\in\RR^3$ and $\psi\in\fdom(\Id\Gamma(\omega))$.
\end{proof}

%%%%%%%%%%%%%%%%%%%%%%%%%%%%%%%%%%%%%%%%%%%%%%%%%

\subsection{Ultraviolet Renormalization}\label{ssecUVren}

\noindent
To discuss the limiting behavior as $\UV$ goes to infinity we split up
$$
\beta_{K,\UV}=\beta_{K,L}+\beta_{L,\UV},\quad \text{for $0\le K\le L\le \UV.$}
$$
This gives rise to a representation of the form $\hv_{\Geb,K,\UV}$ as the sum
of a ``simple" part with fixed ultraviolet cutoff at $L$, whose properties are
well-known, and a $\UV$-dependent perturbation. A computation explained in 
Lemma~\ref{lemLsplit} shows indeed that
\begin{align}\label{defqKLV}
\hv_{\Geb,K,\UV}[\Psi]&=\hv_{\Geb,K,L}[\Psi]+\mathfrak{v}_{\Geb,K,L,\UV}[\Psi],
\end{align}
for all $\Psi\in{\QGV}$ and $0\le K\le L\le\UV<\infty$, with
\begin{align*}
\mathfrak{v}_{\Geb,K,L,\UV}[\Psi]:=\int_{\Geb}\mathfrak{v}_{K,L,\UV,\V{x}}[\Psi]\Id\V{x},
\end{align*}
where, for a.e. $\V{x}\in\Geb$, we abbreviate
\begin{align}\nonumber
&\mathfrak{v}_{K,L,\UV,\V{x}}[\Psi]
\\\nonumber
&:=
\Re\Big[2\SPB{a(e_{\V{x}}\V{m}\beta_{L,\UV})\Psi(\V{x})}{-i\nabla\Psi(\V{x})
+\vp(e_{\V{x}}\V{m}\beta_{K,L})\Psi(\V{x})}
\\\label{defvKL}
&\quad+\|a(e_{\V{x}}\V{m}\beta_{L,\UV})\Psi(\V{x})\|^2+\SPn{\ad(e_{\V{x}}\V{m}\beta_{L,\UV})
\Psi(\V{x})}{a(e_{\V{x}}\V{m}\beta_{L,\UV})\Psi(\V{x})}\Big].
\end{align}
By a simple modification of Nelson's ideas \cite{Nelson1964} and subsequent
extensions to massless bosons \cite{GriesemerWuensch2017,HHS2005} we
obtain the relative bounds on $\mathfrak{v}_{\Geb,K,L,\UV}$ in the next two
propositions, whose proofs are deferred to Appendix~\ref{apprbUV}. The main difference to the
earlier work is the introduction of the extra parameter $L$.
In \cite{GriesemerWuensch2017,HHS2005,Nelson1964} the case $K=L$ is
treated which forces one to either choose $K$ large enough or $|\ee|$ small enough.

\begin{prop}\label{proprbUV1}
Let $\ve>0$. Then there exists $L_{\ve}>0$, otherwise only depending on $\ee$, such that, for all 
$L\ge L_\ve$, there exists $c_{\ve,L}>0$, otherwise only depending on $\ee$ and $V$, such that
\begin{align*}
\int_{\Geb}|\mathfrak{v}_{K,L,\UV,\V{x}}[\Psi]|\Id\V{x}&\le\ve\hv_{\Geb,K,L}[\Psi]+c_{\ve,L}\|\Psi\|^2,
\end{align*}
for all $\Psi\in{\QGV}$ and $K,\UV$ satisfying $0\le K\le L\le\UV<\infty$.
\end{prop}

In the next proposition and henceforth we abbreviate
\begin{align*}
b_{K,\UV}:=\Big(\int_{K<|\V{k}|\le\UV}
\frac{1\vee|\V{k}|^{\eh}}{(|\V{k}|+\V{k}^2/2)^2}\Id\V{k}\Big)^\eh,\quad0\le K\le\UV\le\infty.
\end{align*}

\begin{prop}\label{proprbUV2}
There exists $c>0$, depending only on $\ee$, and, for all $L\ge0$, there exists $c_L>0$, 
otherwise only depending on $\ee$ and $V$, such that
\begin{align*}
\int_{\Geb}|\mathfrak{v}_{K,L,\UV,\V{x}}[\Psi]-\mathfrak{v}_{K,L,\UV',\V{x}}[\Psi]|\Id\V{x}
&\le cb_{\UV,\UV'}(\hv_{\Geb,K,L}[\Psi]+c_L\|\Psi\|^2),
\end{align*}
for all $\Psi\in{\QGV}$, $K\in[0,L]$, and $L\le\UV\le\UV'<\infty$.
\end{prop}

As alluded to above, in the earlier literature the next theorem has been proved only for sufficiently 
small $|\ee|$ \cite{HHS2005} or for sufficiently large $K$ \cite{GriesemerWuensch2017,Nelson1964}.

\begin{thm}\label{thmrbUV}
Let $K_0\in(0,\infty)$. Then the following holds, for all $K\in[0,K_0]$:
\begin{enumerate}
\item[{\rm(1)}] The following limits exist and define a closed semi-bounded form in $L^2(\Geb,\sF)$,
\begin{align}\label{defhKinftyV}
\hv_{\Geb,K,\infty}[\Psi]&:=\lim_{\UV\to\infty}\hv_{\Geb,K,\UV}[\Psi],\quad\Psi\in{\QGV}.
\end{align}
\item[{\rm(2)}] There exists $c>0$, depending only on $\ee$, $V$, and $K_0$, such that
\begin{align}\label{pernille0}
\frac{1}{c}\cfG[\Psi]-c\|\Psi\|^2&\le\hv_{\Geb,K,\UV}[\Psi]\le{c}\cfG[\Psi]+c\|\Psi\|^2,
\end{align}
for all $\Psi\in\QGV$ and $\UV\in[K,\infty]$. (Here $\cfG$ is defined in \eqref{compform}.) 
\item[{\rm(3)}] Let $\HV_{\Geb,K,\infty}$ denote the 
self-adjoint operator representing $\hv_{\Geb,K,\infty}$. Then
$$
\HV_{\Geb,K,\UV}\xrightarrow{\;\;\UV\to\infty\;\;}\HV_{\Geb,K,\infty}\quad
\text{in norm resolvent sense.}
$$
\item[{\rm(4)}] 
Let $\zeta_\UV>-\inf\sigma(\HV_{\Geb,K,\UV})+\epsilon$, for every $\UV\in[K,\infty]$
and some $\epsilon>0$, and suppose that $\zeta_\UV\to\zeta_\infty$, as $\UV\to\infty$. Abbreviate
\begin{align*}
\wt{D}_{\Geb,K,\UV}&:=(\HV_{\Geb,K,\UV}+\zeta_\UV)^{-1}
-(\HV_{\Geb,K,\infty}+\zeta_\infty)^{-1}.
\end{align*}
Then
\begin{align*}
\sup_{{\Psi\in{\QGV}:\atop\|\Psi\|=1}}\big\|(\HV_{\Geb,K,\infty}+\zeta_\infty)^{\nf{1}{2}}
\wt{D}_{\Geb,K,\UV}(\HV_{\Geb,K,\infty}+\zeta_\infty)^{\nf{1}{2}}\Psi\big\|&
\xrightarrow{\;\;\;\UV\to\infty\;\;\;}0,
\\
\sup_{{\Phi\in L^2(\Geb,\fdom(\Id\Gamma(\omega))):\atop\|\Phi\|=1}}
\big\|(1+\Id\Gamma(\omega))^{\nf{1}{2}}\wt{D}_{\Geb,K,\UV}
(1+\Id\Gamma(\omega))^{\nf{1}{2}}\Phi\big\|&\xrightarrow{\;\;\;\UV\to\infty\;\;\;}0.
\end{align*}
\end{enumerate}
\end{thm}

\begin{proof}
We choose $\ve=1/2$ in Proposition~\ref{proprbUV1}, let $L_{\eh}$ denote the 
corresponding parameter appearing in its statement, and put $L:=\max\{K_0,L_{\eh}\}$ so that
$L$ depends on $\ee$ and $K_0$ only. Then
\begin{align}\label{knud}
|\mathfrak{v}_{\Geb,K,L,\UV}[\Psi]|\le\frac{1}{2}\hv_{\Geb,K,L}[\Psi]
+c\|\Psi\|^2,\quad\Psi\in{\QGV},
\end{align}
for all $\UV\in[L,\infty)$ and some $c>0$ depending only on $\ee$, $V$, and $K_0$. 
Proposition~\ref{proprbUV2} shows that the following limits exist,
\begin{align}\label{knud2}
\mathfrak{v}_{\Geb,K,L,\infty}[\Psi]:=\lim_{\UV\to\infty}\mathfrak{v}_{\Geb,K,L,\UV}[\Psi],
\quad\Psi\in{\QGV}.
\end{align}
In view of \eqref{defqKLV} the limits \eqref{defhKinftyV} exist as well.
Since \eqref{knud} extends to $\UV=\infty$,
the symmetric form $\mathfrak{v}_{\Geb,K,L,\infty}$ is a small perturbation of the 
semi-bounded, closed form $\hv_{\Geb,K,L}$. This shows that
$\hv_{\Geb,K,\infty}=\hv_{\Geb,K,L}+\mathfrak{v}_{\Geb,K,L,\infty}$
is semi-bounded and closed, too. Altogether this proves (1).

To prove (2) we first consider $\UV\in[L,\infty]$. Since $L$ depends only on $\ee$ and $K_0$,
the bound \eqref{pernille0} is then a consequence of \eqref{fbUV} and \eqref{knud} 
with $\UV\in[L,\infty]$. Since the constant in \eqref{pernille0} is $\eta$-independent,  
the case $\UV\in[K,L)$ can be accommodated for by choosing an appropriate $\eta$
in the bound derived for $\UV\in[L,\infty]$.

Another consequence of Proposition~\ref{proprbUV2}, \eqref{knud} with $\UV=\infty$, and
\eqref{knud2} is the bound
\begin{align}\label{pernille1}
|\hv_{\Geb,K,\UV}[\Psi]-\hv_{\Geb,K,\infty}[\Psi]|&\le
c'b_{\UV,\infty}(\hv_{\Geb,K,\infty}[\Psi]+\zeta\|\Psi\|^2),\quad\Psi\in{\QGV},
\end{align}
valid for all $\UV\in[L,\infty)$. Here $c'>0$ depends only on $\ee$ and $\zeta>0$
only on $\ee$, $V$, and $K_0$. The bounds \eqref{pernille0} and \eqref{pernille1} together with 
Lemma~\ref{lemeasyresolvent} now imply all statements of (4), which in turn implies (3).
\end{proof}

\begin{rem}\label{remdefrenNelson}
Pick $K\ge0$ such that $\beta_{K,\infty}$ is square-integrable. Then the strong continuity of
the Weyl representation implies that $G_{K,\UV}\to G_{K,\infty}$, $\UV\to\infty$, strongly.
Therefore, the following limit exists in strong resolvent sense,
\begin{align}\label{defNelsonHam}
\NV_{\Geb,\infty}:=\lim_{\UV\to\infty}\NV_{\Geb,\UV}
=\lim_{\UV\to\infty}G_{K,\UV}^*\HV_{\Geb,K,\UV}G_{K,\UV}
=G_{K,\infty}^*\HV_{\Geb,K,\infty}G_{K,\infty}.
\end{align}
This is Nelson's \cite{Nelson1964} definition of the 
{\em renormalized Nelson Hamiltonian $\NV_{\Geb,\infty}$}.
Later on it was observed that the convergence \eqref{defNelsonHam}
actually holds in norm resolvent sense as well
\cite{Ammari2000,GriesemerWuensch2017,MatteMoeller2017}.
\end{rem}

We denote the quadratic form associated with $\NV_{\Geb,\infty}$ by $\nv_{\Geb,\infty}$.

\begin{rem}\label{remdefnonFockNelson}
In the case where $\beta_{0,\infty}\notin L^2(\RR^3)$, we refer to the operator $\HV_{\Geb,0,\infty}$ 
as the {\em renormalized Nelson operator in the non-Fock representation}. 

In fact, the operators $\HV_{\RR^3,K,\infty}$, $K\ge0$, have already been
constructed non-perturbatively in \cite{MatteMoeller2017} as generators of the Feynman-Kac
semigroups introduced further below. The convergence 
$\HV_{\RR^3,K,\UV}\to\HV_{\RR^3,K,\infty}$, $\UV\to\infty$, in norm resolvent sense
has also been observed in \cite{MatteMoeller2017}. In the latter book it is, however, not
verified that the form corresponding to $\HV_{\RR^3,K,\infty}$ is given by the limit 
\eqref{defhKinftyV}. Notice that the latter result does not follow from the norm resolvent convergence
and general principles. For example, in the case where $\mu=0$ and $\ee\eta$ is constant and
non-zero, it is known \cite{GriesemerWuensch2017} that
$\fdom(\NV_{\RR^3,\infty})\cap{\sQ_{\RR^3}}=\{0\}$. In particular, 
the forms $\nv_{\RR^3,\UV}$, $\UV\in[0,\infty)$, which are defined on ${\sQ_{\RR^3}}$,
do not converge pointwise to $\nv_{\RR^3,\infty}$, although the corresponding operators converge
in norm resolvent sense.
\end{rem}

\begin{rem}
The two operators we are really interested in are $\NV_{\Geb,\infty}$ and $\HV_{\Geb,0,\infty}$.
Assume that $\beta_{0,\infty}\notin L^2(\RR^3)$. Then they are not unitarily equivalent, but they
still have the same spectrum. This holds because, by construction, $\NV_{\Geb,\infty}$ is unitarily
equivalent to every $\HV_{\Geb,K,\infty}$ with $K>0$ and
$\HV_{\Geb,K,\infty}\to\HV_{\Geb,0,\infty}$, $K\downarrow0$, in the norm resolvent sense;
for the latter result confer \cite{HHS2005,MatteMoeller2017} or apply Lemma~\ref{lempert} below
with $\hat{\ee}=\ee$ and $\hat{\eta}=1_{\{|\V{m}|>K\}}\eta$.
\end{rem}

For later reference we note two simple consequences of Theorem~\ref{thmrbUV}:

\begin{ex}\label{exIMS}
Let $0\le K<\UV\le\infty$ and let
$\chi^2_{0}+\chi^2_{1}=1$ be a smooth IMS type partition of unity on $\RR^3$, 
where $\chi_{0}\ge0$ has compact support in $\{|\V{x}|\le2\}$ and $\chi_1\ge0$ is
supported in $\{|\V{x}|\ge1\}$. 
We may further assume that $|\nabla\chi_k|\le2$, for $k\in\{0,1\}$.
Put $\chi_{k,R}(\V{x}):=\chi_k(\V{x}/R)$, $R\ge1$, $k\in\{0,1\}$. Then multiplication with 
$\chi_{0,R}$ or $\chi_{1,R}$ leaves $\dom(\hv_{\Geb,\UV})={\QGV}$ invariant, 
and the following IMS localization formula is valid for all $\Psi\in{\QGV}$,
\begin{equation}\label{IMSq}
\hv_{\Geb,K,\UV}[\Psi]=
\sum_{k=0}^1\bigg\{\hv_{\Geb,K,\UV}[\chi_{k,R}\Psi]
-\frac{1}{2R^2}\int_{\Geb}|\nabla\chi_{k}(\V{x}/R)|^2\|\Psi(\V{x})\|^2\Id\V{x}\bigg\}.
\end{equation}
The formula is well-known at least for finite $\UV$ and extends to $\UV=\infty$ by 
Theorem~\ref{thmrbUV}. If $f\in C^\infty(\RR^3,\RR)$ is bounded and has a bounded derivative, 
then multiplication with $f$ leaves ${\QGV}$ invariant as well, and we readily verify the relation
\begin{align}\label{exp5}
\Re\,\hv_{\Geb,K,\UV}[f^2\Psi,\Psi]&=
\hv_{\Geb,K,\UV}[f\Psi]-\frac{1}{2}\|(\nabla f)\Psi\|^2,\quad\Psi\in{\QGV},
\end{align}
again starting with finite $\UV$ and passing
to the limit $\UV\to\infty$ with the help of Theorem~\ref{thmrbUV}.
\end{ex}

After constructing $\NV_{\Geb,\infty}$ and clarifying its relation to $\HV_{\Geb,0,\UV}$
the parameter $K$ has served its purpose. We shall set it to zero in the remaining part of the main
text and simplify our notation by setting
\begin{align*}
\beta_\UV&:=\beta_{0,\UV},\quad
\hv_{\Geb,\UV}:=\hv_{\Geb,0,\UV},\quad \HV_{\Geb,\UV}:=\HV_{\Geb,0,\UV},\quad\UV\in[0,\infty].
\end{align*}
We further abbreviate
\begin{align*}
\mathfrak{v}_{\Geb,\UV}:=\mathfrak{v}_{\Geb,0,0,\UV},\quad\UV\in[0,\infty),
\end{align*}
and notice that, by \eqref{defqKLV} and \eqref{defhKinftyV}, the limit
\begin{align*}
\mathfrak{v}_{\Geb,\infty}&:=\lim_{\UV\to\infty}\mathfrak{v}_{\Geb,\UV}
\end{align*}
exists on $\QGV$. We then have
\begin{align}\label{forhtv}
\hv_{\Geb,\UV}=\mathfrak{t}_\Geb+\mathfrak{v}_{\Geb,\UV},\quad\UV\in[0,\infty],
\end{align}
where $\mathfrak{t}_\Geb$ is defined in \eqref{defvvSchroedinger}.

\begin{lem}\label{lempert}
Let $\UV\in[0,\infty]$. Pick a second coupling constant $\hat{\ee}$ and another measurable
even function $\hat{\eta}:\RR^3\to\RR$ with $0\le\hat{\eta}\le1$. Keep the boson mass
$\mu\ge0$ fixed and define
\begin{align*}
\hat{\beta}_\UV&:=\hat{\ee}\hat{\eta}\omega^\mh(\omega+\V{m}^2/2)^{-1}1_{\{|\V{m}|\le\UV\}},
\\
d_\UV&:=\|(|\V{m}|^\eh\vee|\V{m}|^{\nf{3}{4}})(\beta_\UV-\hat{\beta}_\UV)\|.
\end{align*}
Let $\hat{\mathfrak{h}}_{\Geb,\UV}$ be the quadratic form obtained upon putting 
$\hat{\ee}\hat{\eta}$ in place of $\ee\eta$ in the construction of $\hv_{\Geb,\UV}$. Then
\begin{align}\label{rbhtildeh}
|\hv_{\Geb,\UV}[\Psi]-\hat{\mathfrak{h}}_{\Geb,\UV}[\Psi]|&\le 
cd_\UV(\hv_{\Geb,\UV}[\Psi]+\zeta\|\Psi\|^2),\quad\Psi\in{\QGV},
\end{align}
for some $c>0$ depending only on $V$ and a common upper bound on $|\ee|$ and $|\hat{\ee}|$, 
and for some $\zeta>0$ depending only on $\ee$ and $V$.
\end{lem}

\begin{proof}
First, suppose that $0\le\UV<\infty$. Combining \eqref{forhtv} and its analogue for 
$\hat{\mathfrak{h}}_{\Geb,\UV}$ with \eqref{bdvmu} we then deduce that
\begin{align*}
|\hv_{\Geb,\UV}[\Psi]-\hat{\mathfrak{h}}_{\Geb,\UV}[\Psi]|
\le c'd_\UV(1\vee|\ee|\vee|\hat{\ee}|)(\cfG[\Psi]+\|\Psi\|^2),\quad\Psi\in{\QGV},
\end{align*}
with a universal constant $c'>0$. Together with \eqref{pernille0} this implies \eqref{rbhtildeh} for finite 
$\UV$, which then extends to the case $\UV=\infty$ by virtue of Theorem~\ref{thmrbUV}.
\end{proof}

%%%%%%%%%%%%%%%%%%%%%%%%%%%%%%%%%%%%%%%%%%%%%%%%%

\subsection{Feynman-Kac Formulas}\label{ssecFK} 

\noindent
Our constructions of path measures associated with ground states are based on 
Feynman-Kac formulas for the semigroups generated by
$\NV_{\Geb,\UV}$ and $\HV_{\Geb,\UV}$. For $\UV=\infty$ and $\Geb=\RR^3$, these
formulas were proven in \cite{MatteMoeller2017}. 
For finite $\UV$ and $\Geb=\RR^3$, Feynman-Kac formulas with a representation of
the integrand different from the one given below have been known since a long time;
see, e.g., the textbook \cite{LHB2011} and the references given there. 
The latter well-known formulas have, however, the disadvantage of applying only to vectors $\Psi$ 
in suitable dense subspaces of $L^2(\RR^3,\sF)$ an they do not seem to imply $L^2$-to-$L^p$-norm or
hypercontractivity bounds on the semigroup.

To explain our Feynman-Kac formulas we first have to introduce more notation.
In the whole article $(\Omega,\fF,(\fF_{t})_{t\ge0},\PP)$ is a filtered probability space 
satisfying the ``usual assumptions" of completeness and right continuity. 
The bold letter $\V{B}$ denotes a three-dimensional $(\fF_{t})_{t\ge0}$-Brownian motion. 
For every $\V{x}\in\RR^3$, we put $\V{B}^{\V{x}}:=\V{x}+\V{B}$. The first entry time of 
$\V{B}^{\V{x}}$ into $\Geb^c$ will be denoted by
\begin{align}\label{firstentryBGebc}
\tau_{\Geb}(\V{x})&:=\inf\big\{t\ge0\,\big|\:\V{B}^{\V{x}}_t\in\Geb^c\big\}.
\end{align}
We call a stochastic process continuous if all its path are continuous, and not just almost all of them.
The Brownian motion $\V{B}$ is assumed to be continuous in this sense.

In what follows we further abbreviate
\begin{align}\label{deffrk}
\mathfrak{k}&:=L^2(\RR^3,[\omega^{-1}\vee1]\Id\V{k}).
\end{align}
Our Feynman-Kac formulas involve the series 
\begin{align*}
F_t(h)&:=\sum_{n=0}^\infty\frac{1}{n!}\ad(h)^ne^{-t\Id\Gamma(\omega)},
\quad h\in\mathfrak{k},\,t>0.
\end{align*}
Their partial sums are indeed well-defined on $\sF$ and they
converge absolutely with respect to the operator norm on $\LO(\sF)$. The resulting maps
$F_t:\mathfrak{k}\to\LO(\sF)$ are analytic. Furthermore,
\begin{align}\label{Fexpv}
F_t(h)\epsilon(g)&=\epsilon(h+e^{-t\omega}g),
\\\label{Fstarexpv}
F_t(h)^*\epsilon(g)&=e^{\SPn{h}{g}}\epsilon(e^{-t\omega}g),
\\\label{bdFt}
\|F_t(h)\|&\le c\exp\Big(4\|[(t\omega)^\mh\vee1]h\|_{L^2(\RR^3)}^2\Big),
\end{align}
for all $h\in\mathfrak{k}$, $g\in L^2(\RR^3)$, $t>0$, and some $c>0$.
The maps $F_t$ have been introduced and discussed in \cite[App.~6]{GMM2017}
and \eqref{bdFt} follows easily from Lemma~17.4 in that paper.

Recall the definition \eqref{realsubspace} of the completely real subspace 
${\mathfrak{r}}\subset L^2(\RR^3)$ and let $\UV\in[0,\infty]$ and $\V{x}\in\RR^3$. 
In \cite{MatteMoeller2017} we constructed
\begin{enumerate}
\item[$\bullet$] continuous adapted real-valued processes 
\begin{align*}
u_{\UV}=(u_{\UV,t})_{t\ge0},\quad\tilde{u}_{\UV}=(\tilde{u}_{\UV,t})_{t\ge0},
\end{align*}
satisfying $u_{\UV,0}=\tilde{u}_{\UV,0}=0$;
\item[$\bullet$] continuous adapted ${\mathfrak{r}}$-valued processes 
\begin{align*}
U_{\UV}^+=(U_{\UV,t}^+)_{t\ge0},\quad U_{\UV}^-=(U_{\UV,t}^-)_{t\ge0},\quad
\wt{U}_{\UV}^+=(\wt{U}_{\UV,t}^+)_{t\ge0},\quad\wt{U}_{\UV}^-=(\wt{U}_{\UV,t}^-)_{t\ge0},
\end{align*}
satisfying $U_{\UV,0}^\pm=\wt{U}_{\UV,0}^\pm=0$;
\end{enumerate}
such that the contributions to the Feynman-Kac integrands coming from the radiation field are given by
\begin{align}\label{defW}
W_{\UV,t}(\V{x})&:=e^{u_{\UV,t}}F_{\nf{t}{2}}(-e_{\V{x}}U_{\UV,t}^+)
F_{\nf{t}{2}}(-e_{\V{x}}U_{\UV,t}^-)^*,
\\\label{defwtW}
\wt{W}_{\UV,t}(\V{x})&:=e^{\tilde{u}_{\UV,t}}
F_{\nf{t}{2}}(e_{\V{x}}\wt{U}_{\UV,t}^+)F_{\nf{t}{2}}(e_{\V{x}}\wt{U}_{\UV,t}^-)^*,
\end{align}
for all $t>0$, and $W_{\UV,0}(\V{x}):=\wt{W}_{\UV,0}(\V{x}):=\id_{\sF}$.
(The notation $W_{t}^V(\V{x})$ used in the introduction is defined in \eqref{WVshorthand}.)

In the whole article it will never be necessary to employ explicit formulas for 
$u_\UV$, $\tilde{u}_\UV$, $U^+_\UV$, or $\wt{U}_\UV^\pm$, 
whence we refer the interested reader 
to \cite{MatteMoeller2017} for detailed information. We shall merely introduce and employ some 
formulas for $U^-_\infty$  in Sects.~\ref{ssecsuperexpN} and~\ref{ssecGaussdom}.

The Feynman-Kac semigroups associated with the above processes and $\Geb$ are defined by
\begin{align}\label{defTGL}
(T_{\Geb,\UV,t}\Psi)(\V{x})&:=\EE\Big[1_{\{\tau_{\Geb}(\V{x})>t\}}
e^{-\int_0^tV(\V{B}^{\V{x}}_s)\Id s}W_{\UV,t}(\V{x})^*\Psi(\V{B}^{\V{x}}_t)\Big],
\\\label{defwtTGL}
(\wt{T}_{\Geb,\UV,t}\Psi)(\V{x})&:=\EE\Big[1_{\{\tau_{\Geb}(\V{x})>t\}}
e^{-\int_0^tV(\V{B}^{\V{x}}_s)\Id s}\wt{W}_{\UV,t}(\V{x})^*\Psi(\V{B}^{\V{x}}_t)\Big],
\end{align}
for all $t\ge0$ and $\Psi\in L^2(\RR^3,\sF)$. Notice that, by their definition, 
$T_{\Geb,\UV,t}$ and $\wt{T}_{\Geb,\UV,t}$ act on equivalence classes of
functions defined on the whole $\RR^3$. Since $1_{\{\tau_{\Geb}(\V{x})>t\}}=0$, for all
$\V{x}\in\Geb^c$ and $t\ge0$, their action on $\Psi$ depends, however, only on the restriction
of $\Psi$ to $\Geb$. If $\Psi\in L^2(\Geb',\sF)$, with an open $\Geb'\subset\RR^3$ 
not necessarily equal to $\Geb$,
then we extend it by $0$ to the whole $\RR^3$ and
denote the action of $T_{\Geb,\UV,t}$ and $\wt{T}_{\Geb,\UV,t}$ on this extension
again by the symbols on the left hand sides of \eqref{defTGL} and \eqref{defwtTGL}, respectively.

For every $\V{x}\in\RR^3$, the expectations in \eqref{defTGL} and \eqref{defwtTGL} are well-defined
$\sF$-valued Bochner-Lebesgue integrals. In fact, since $V$ is Kato decomposable, we know 
\cite{AizenmanSimon1982} that
\begin{align}\label{kashmir}
\sup_{\V{x}\in\RR^3}\EE\Big[e^{-p\int_0^tV(\V{B}_s^{\V{x}})\Id s}\Big]<\infty,\quad p,t\ge 0,
\end{align}
and in \cite{MatteMoeller2017} it is shown that
\begin{align}\label{bduU}
\EE[e^{pu_{\UV,t}}]\le c_{t},\quad
\EE\Big[e^{p\|[(t\omega)^\mh\vee1]U^\pm_{\UV,t}\|^2}\Big]\le c_{t},
\end{align}
for all $p,t>0$. Here and in \eqref{WinLp} below the constants $c_t,c_t'>0$ depend only on 
$p$ and $\ee$ besides $t$. The bounds \eqref{bduU} still hold true, when the symbols 
$\tilde{u}$ and $\wt{U}$ are put in place of $u$ and $U$, respectively. In view of \eqref{bdFt}
we thus have
\begin{align}\label{WinLp}
\sup_{\V{x}\in\RR^3}\EE\big[\|W_{\UV,t}(\V{x})\|^p\big]\le c_{t}',\quad
\sup_{\V{x}\in\RR^3}\EE\big[\|\wt{W}_{\UV,t}(\V{x})\|^p\big]\le c_{t}',\quad p,t\ge0.
\end{align}

\begin{thm}\label{thmFKD}
For all $\UV\in[0,\infty]$, $\Psi\in L^2(\Geb,\sF)$, and $t\ge0$, the following Feynman-Kac
formulas are satisfied,
\begin{align}\label{FKDfor}
e^{-t\NV_{\Geb,\UV}}\Psi&=T_{\Geb,\UV,t}\Psi,\quad
e^{-t\HV_{\Geb,\UV}}\Psi=\wt{T}_{\Geb,\UV,t}\Psi.
\end{align}
\end{thm}

\begin{proof}
For $\Geb=\RR^3$, the theorem is proven in \cite{MatteMoeller2017}. The extension to
proper open subsets $\Geb$ of $\RR^3$ proceeds along the lines of the appendix to
\cite{Simon1978Adv}. The details are explained in Appendix~\ref{appFKD} where we use some
technical results of \cite{Matte2019}.
\end{proof}

We shall crucially use the following result on the Feynman-Kac integrands, whose 
proof can be found in \cite[\textsection8.1]{MatteMoeller2017}. Here we employ
the unitary map $\cU:\sF\to L^2(\cQ,\GM)$ constructed in Subsection~\ref{ssecQspace}. 

\begin{thm}\label{thmWpos}
For all $\UV\in[0,\infty]$, $\V{x}\in\RR^3$, $t>0$, and pointwise on $\Omega$, the operators
$\cU W_{\UV,t}(\V{x})^*\cU^*$ and $\cU\wt{W}_{\UV,t}(\V{x})^*\cU^*$
are positivity improving.
\end{thm}

With the help of $\cU$ and Fubini's theorem we can construct a natural isomorphism
\begin{align*}
\cU_\Geb:L^2(\Geb,\sF)\longrightarrow L^2(\Geb\times\cQ,\Id\V{x}\otimes\GM),
\end{align*} 
by setting 
\begin{align}\label{defcUG}
(\cU_\Geb\Psi)(\V{x},q):=(\cU\Psi(\V{x}))(q),
\end{align} 
for a.e. $(\V{x},q)\in\Geb\times\cQ$ and all $\Psi\in L^2(\Geb,\sF)$.

\begin{cor}\label{corTpos}
Let $\UV\in[0,\infty]$ and $t>0$. Then $\cU_\Geb T_{\Geb,\UV,t}\cU_\Geb^*$ and 
$\cU_\Geb\wt{T}_{\Geb,\UV,t}\cU_\Geb^*$ are positivity preserving.
If $\Geb$ is connected, then they are positivity improving.
\end{cor}

\begin{proof}
The first claim is evident from Theorem~\ref{thmWpos}. So, assume right away that $\Geb$ is
connected. We shall only consider $T_{\Geb,\UV,t}$, as the proof for 
$\wt{T}_{\Geb,\UV,t}$ is identical. Let $\Phi,\Upsilon\in L^2(\Geb\times\cQ,\Id\V{x}\otimes\GM)$ 
be non-negative and non-zero. We have to show that 
$\SPn{\Phi}{\cU_\Geb T_{\Geb,\UV,t}\cU_\Geb^*\Upsilon}>0$.
Let $\Phi(\cdot)$ and $\Upsilon(\cdot)$ be representatives of $\Phi$ and $\Upsilon$,
respectively. Let $\cP_\Phi$ be the set of all $\V{x}\in\Geb$ for which
$\int_\cQ|\Phi(\V{x},q)|^2\Id\GM(q)>0$ and define $\cP_\Upsilon$ analogously. 
Pick some elementary event $\gamma\in\Omega$.
If $\V{x}\in\cP_\Phi$ and $\V{B}_t^{\V{x}}(\gamma)\in\cP_{\Upsilon}$,
then it follows from Theorem~\ref{thmWpos} that
\begin{align*}
\SPb{\cU^*\Phi(\V{x},\cdot)}{W_{\UV,t}(\V{x},\gamma)^*\cU^*
\Psi(\V{B}_t^{\V{x}}(\gamma),\cdot)}_{L^2(\cQ,\nu)}>0.
\end{align*}
Therefore, it remains to show that, for every $\V{x}\in\Geb$, 
\begin{align*}
\EE\big[1_{\{\tau_\Geb(\V{x})>t\}}1_{\cP_\Upsilon}(\V{B}_t^{\V{x}})\big]&>0.
\end{align*}
Since $\cP_\Upsilon$ has strictly positive measure,
this follows, however, from the Feynman-Kac formula for the Dirichlet-Laplacian on $\Geb$
and the fact that the semigroup generated by the latter operator is positivity improving because 
$\Geb$ is connected; see \cite[Lem.~1]{FarisSimon1975}.
\end{proof}

In the next theorem and henceforth we call a vector $\Psi\in L^2(\Geb,\sF)$ strictly positive, 
if $\cU_\Geb\Psi$ is strictly positive. Furthermore, we set
\begin{align*}
E_{\Geb,\UV}&:=\inf\sigma(\NV_{\Geb,\UV})=\inf\sigma(\HV_{\Geb,\UV}).
\end{align*}

\begin{thm}\label{thmPerronFrobenius}
Let $\UV\in[0,\infty]$, suppose that $\Geb$ is connected, and assume that $E_{\Geb,\UV}$
is an eigenvalue of $\HV_{\Geb,\UV}$. Then $E_{\Geb,\UV}$ has multiplicity one and there
exists a corresponding eigenvector that is strictly positive. The same statement holds with
$\NV_{\Geb,\UV}$ put in place of $\HV_{\Geb,\UV}$.
\end{thm}

\begin{proof}
The assertion follows from the Feynman-Kac formulas of Theorem~\ref{thmFKD}, from
Corollary~\ref{corTpos}, and from Faris' Perron-Frobenius type theorem \cite{Faris1972}.
\end{proof}

%%%%%%%%%%%%%%%%%%%%%%%%%%%%%%%%%%%%%%%%%%%%%%%%%
%%%%%%%%%%%%%%%%%%%%%%%%%%%%%%%%%%%%%%%%%%%%%%%%%
%%%%%%%%%%%%%%%%%%%%%%%%%%%%%%%%%%%%%%%%%%%%%%%%%

\section{Existence of Ground States}\label{secGS}

\noindent
The objective of this section is to show that, under a binding condition discussed in
Subsection~\ref{ssecbinding}, the minimal energy $E_{\Geb,\UV}$
is always an eigenvalue of $\HV_{\Geb,\UV}$ and, under the additional infrared regularity condition
$\omega^{-3}\eta^2\in L^1_\loc(\RR^3)$, it is also an eigenvalue of $\NV_{\Geb,\UV}$. 
The existence proofs proceed in two main steps:
\begin{enumerate}
\item[(i)] We consider strictly positive boson masses $\mu$ and bounded $\Geb$ and apply
a criterion due to Gross \cite{Gross1972}; the required hypercontractivity of the semigroups
can be inferred from the results of \cite{MatteMoeller2017}.
\item[(ii)] In a chain of approximation arguments we successively trade the restriction
$\mu>0$ for a sharp infrared cutoff, remove that infrared cutoff afterwards, and pass to
possibly unbounded $\Geb$. For technical reasons we perform these three steps at a finite
ultraviolet cutoff, which is removed in a last approximation step. 
In each of these four steps we apply a recent variant \cite{Matte2016} of a compactness 
argument from \cite{GLL2001} to some approximating sequence of eigenvectors.
\end{enumerate}
The main steps~(i) and~(ii) are presented in Subsection~\ref{ssecGSmass} and 
Subsection~\ref{ssecGSgen}, respectively. The compactness argument mentioned in (ii) requires two 
crucial technical ingredients. The first one, which is only needed when $\Geb$ is unbounded, 
is a uniform bound on the spatial localization of the considered eigenvector sequence. It is presented
in Subsection~\ref{ssecexploc} and most parts of its proof are deferred to Appendix~\ref{appexploc}.
The second ingredient is a formula for the action of a 
``pointwise" annihilation operator on the ground state eigenvectors revealing information 
about their dependence on the boson momenta; see Subsection~\ref{ssecIR}.
The compactness argument itself is explained in Subsection~\ref{sseccomparg}. It is based on a Fock 
space adaption of the well-known characterization of compact sets in $L^2(\RR^d)$, whose detailed 
proof is provided in Appendix~\ref{appcpt} for the convenience of the reader. In the final
Subsection~\ref{sseccontGS} we discuss continuity properties of ground states.

%%%%%%%%%%%%%%%%%%%%%%%%%%%%%%%%%%%%%%%%%%%%%%%%%
%%%%%%%%%%%%%%%%%%%%%%%%%%%%%%%%%%%%%%%%%%%%%%%%%

\subsection{Ground States for Massive Bosons and Bounded Domains}\label{ssecGSmass}

\noindent
The next proposition provides some key estimates permitting to prove the existence of ground
states for massive bosons and bounded $\Geb$ in the subsequent theorem. 
The proposition itself holds, however, also for massless bosons and unbounded $\Geb$.
Notice that everything done in this subsection applies to the renormalized operators 
($\UV=\infty$) right away. Later on we shall, however, apply Theorem~\ref{thmGSmass} only for 
$\UV<\infty$, since finiteness of $\UV$ is required in Proposition~\ref{prop-IR}.

\begin{prop}\label{propGSmu} 
Let $\vk$ be a non-negative, bounded multiplication operator on $L^2(\RR^3)$ and suppose
that $\vk\le\omega$. Then the following holds, for all $\UV\in[0,\infty]$ and $t>0$:
\begin{enumerate}
\item[{\rm(1)}] For all $\V{x}\in\RR^3$ and pointwise on $\Omega$, the operator 
${W}_{\UV,t}(\V{x})^*$ maps $\sF$ into $\dom(e^{t\Id\Gamma(\vk)/6})$ and, 
for every $p>0$, there exists a constant $c_{p,t}>0$, depending only on
$\ee$ and $\|\vk\|$ in addition, such that
\begin{align}\label{massiv1}
\sup_{\V{x}\in\RR^3}\EE\Big[\|e^{t\Id\Gamma(\vk)/6}{W}_{\UV,t}(\V{x})^*\|^p\Big]&\le c_{p,t}.
\end{align}
\item[{\rm(2)}] Let $p\in[2,\infty]$. Suppose that $F:\RR^3\to\RR$ is Lipschitz continuous
with Lipschitz constant $L\ge0$ and bounded from below. Then ${T}_{\Geb,\UV,t}$ maps the range
$e^{-F}L^2(\Geb,\sF)$ into  $L^p(\Geb,\dom(e^{t\Id\Gamma(\vk)/6});e^F\Id\V{x})$ and, 
for every $\Psi\in L^2(\Geb,\sF)$,
\begin{align}\label{hcbd}
\|e^{t\Id\Gamma(\vk)/6}e^F T_{\Geb,\UV,t}e^{-F}\Psi\|_{L^{p}(\Geb,\sF)}&\le 
c_{p,t}\frac{e^{6tL^2}}{t^{3(1/2-1/p)/2}}\|\Psi\|_{L^2(\Geb,\sF)}.
\end{align}
Here $c_{p,t}>0$ depends only on $\ee$, $\|\vk\|$, and $V$ in addition and it is
monotonically increasing in $t$.
Instead of the $L^\infty$-norm, i.e., essential supremum, one can also take the pointwise
supremum in \eqref{hcbd}. Furthermore, we can replace $L^q(\Geb,\sF)$ by
$L^q(\RR^3,\sF)$, for $q\in\{2,p\}$, in \eqref{hcbd}.
\item[{\rm(3)}] If $\Phi_{\Geb,\UV}$ is an eigenvector of $\NV_{\Geb,\UV}$, then 
$\Phi_{\Geb,\UV}\in L^p(\Geb,\dom(e^{r\Id\Gamma(\vk)}))$, for all $r>0$ and $p\in[2,\infty]$.
\end{enumerate}
The same assertions hold for $\wt{W}_{\UV,t}(\V{x})$, $\wt{T}_{\Geb,\UV,t}$, and eigenvectors
of $\HV_{\Geb,\UV}$ as well.
\end{prop}

\begin{proof}
To prove the first part, let $s>r>0$ and $g\in\mathfrak{k}$, where $\mathfrak{k}$ is the space defined
in \eqref{deffrk}. Then \eqref{SGdGamma} and \eqref{Fexpv} imply
\begin{align*}
e^{r\Id\Gamma(\vk)}F_{0,s}(g)\epsilon(h)&=\epsilon(e^{-s\omega+r\vk}h+e^{r\vk}g)
=F_{0,s-r}(e^{r\vk}g)e^{-r\Id\Gamma(\omega-\vk)}\epsilon(h),
\end{align*}
for all $h\in L^2(\RR^3)$. From this, the totality of the exponential vectors, and \eqref{bdFt} 
we infer that $F_{0,s}(g)$ maps $\sF$ into $\dom(e^{r\Id\Gamma(\vk)})$ and
\begin{align*}
e^{r\Id\Gamma(\vk)}F_{0,s}(g)=F_{0,s-r}(e^{r\vk}g)e^{-r\Id\Gamma(\omega-\vk)}.
\end{align*}
Here $\|e^{-r\Id\Gamma(\omega-\vk)})\|\le1$. Together with \eqref{bdFt} 
and \eqref{defW} this further implies that, at every fixed elementary event, 
${W}_{\UV,t}(\V{x})^*$ maps $\sF$ into $\dom(e^{t\Id\Gamma(\vk)/6})$ with
\begin{align*}
\|e^{t\Id\Gamma(\vk)/6}{W}_{\UV,t}(\V{x})^*\|&
\le c'e^{u_{\UV,t}+ce^{t\|\vk\|/3}\|[(t\omega)^\mh\vee1]U^-_{\UV,t}(\V{x})\|^2
+c\|[(t\omega)^\mh\vee1]U^+_{\UV,t}(\V{x})\|^2},
\end{align*}
for some universal constants $c,c'>0$. The bound \eqref{massiv1} now follows from
\eqref{bduU} and H\"{o}lder's inequality.

Now let $p\in[2,\infty)$ and $\Psi\in L^2(\RR^3,\sF)$. Then we further obtain
\begin{align*}
\int_\Geb&\EE\Big[e^{F(\V{x})-F(\V{B}_t^{\V{x}})-\int_0^tV(\V{B}_s^{\V{x}})\Id s}
\|e^{t\Id\Gamma(\vk)/6}W_{\Lambda,t}(\V{x})^*\|\|\Psi(\V{B}_t^{\V{x}})\|\Big]^p\Id\V{x}
\\
&\le\EE\Big[e^{6L|\V{B}_t|}\Big]^{\nf{p}{6}}
\sup_{\V{y}\in\RR^3}\EE\Big[e^{-6\int_0^tV(\V{B}_s^{\V{y}})\Id s}\Big]^{\nf{p}{6}}
\\
&\quad\cdot\sup_{\V{z}\in\RR^3}
\EE\Big[\|e^{t\Id\Gamma(\vk)/6}{W}_{\UV,t}(\V{z})^*\|^6\Big]^{\nf{p}{6}}
\int_{\RR^3}\EE\big[\|\Psi(\V{B}_t^{\V{x}})\|^2\big]^{\nf{p}{2}}\Id\V{x}.
\end{align*}
Since $e^{t\Delta/2}$ maps $L^{1}(\RR^3)$ continuously into $L^{\nf{p}{2}}(\RR^3)$, the integral
in the last line is less than or equal to
$c_{p,t}\|\|\Psi(\cdot)\|^2\|_{L^1(\RR^3)}^{\nf{p}{2}}=c_{p,t}\|\Psi\|_{L^2(\RR^3,\sF)}^p$.
According to wellknown bounds on the semigroup of the free Laplacian we may take
$c_{p,t}=c_pt^{-3(p/2-1)/2}$. In view of \eqref{kashmir} and the bound
\begin{align*}
\EE\Big[e^{r|\V{B}_t|}\Big]&\le2^{\nf{3}{2}}e^{r^2t},\quad r\ge0,
\end{align*}
we arrive at \eqref{hcbd} for finite $p$ and $L^q(\RR^3,\sF)$-norms. Its version for 
$L^q(\Geb,\sF)$-norms follows trivially by considering $\Psi$ that vanish a.e. on $\Geb^c$.
Obvious modifications of these arguments take 
care of the case $p=\infty$. Altogether this implies the second part of the proposition.

The third part follows from the second one and 
$\Phi_{\Geb,\UV}=e^{tE_{\Geb,\UV}}T_{\Geb,\UV,t}\Phi_{\Geb,\UV}$, $t\ge0$. 

The same proof works also for the Gross transformed objects, whence the last assertion is clear.
\end{proof}

\begin{thm}\label{thmGSmass}
Assume that $\mu>0$ and that $\Geb\subset\RR^3$ is open and bounded. Let $\UV\in[0,\infty]$.
Then $E_{\Geb,\UV}$ is an eigenvalue of both $\NV_{\Geb,\UV}$ and $\HV_{\Geb,\UV}$ and
all corresponding eigenvectors are contained in the domain of $e^{r\Id\Gamma(1)}$, for every $r>0$. 
\end{thm}

\begin{proof}[Proof of Theorem~\ref{introthmhypcontr} and Theorem~\ref{thmGSmass}.]
We only treat $\NV_{\Geb,\UV}$ explicitly and make use of the first 
Feynman-Kac formula in \eqref{FKDfor} without further mention. To deal with 
$\HV_{\Geb,\Lambda}$ we simply have to employ the second formula in \eqref{FKDfor} instead. 

Recall the definitions of the Wiener-It\^{o}-Segal type isomorphisms $\cU$ and $\cU_{\Geb}$
in Subsection~\ref{ssecQspace} and \eqref{defcUG}, respectively.

Since $(\Geb\times\cQ,\fB(\Geb)\otimes\fQ,\Id\V{x}\otimes\GM)$ (with the Borel-$\sigma$-algebra 
$\fB(\cG)$ of $\Geb$) is a finite measure
space and $\cU_\Geb{T}_{\Geb,\UV,t}\cU_\Geb^*$ is bounded, self-adjoint, and positivity preserving, 
we may employ an abstract result of Gross \cite[Thm.~1]{Gross1972} to prove the existence of ground
state eigenvectors. According to Gross' theorem, it suffices to prove that 
$\cU_\Geb{T}_{\Geb,\UV,t}\cU_\Geb^*$ maps 
$L^2(\Geb\times\cQ,\Id\V{x}\otimes\GM)$ continuously into
$L^p(\Geb\times\cQ,\Id\V{x}\otimes\GM)$, for some 
$t>0$ and $p>2$. We thus fix some arbitrary $t>0$ and define $p_{\mu,t}:=e^{t\mu/3}+1$.
Then a well-known hypercontractivity bound of Nelson 
(see, e.g., \cite[Thm.~I.17]{Simon1974}) implies that 
$$
\|\cU e^{-t\Id\Gamma(\mu)/6}\cU^*\phi\|_{L^{p_{\mu,t}}(\GM)}
\le\|\phi\|_{L^2(\GM)}=\|\cU^*\phi\|_{\sF},\quad\phi\in L^2(\cQ,\GM).
$$ 
By virtue of Proposition~\ref{propGSmu} we thus get
\begin{align*}
\int_{\Geb}\|\cU (T_{\Geb,\UV,t}\Psi)(\V{x})\|^{p_{\mu,t}}_{L^{p_{\mu,t}}(\GM)}\Id\V{x}
&\le \int_{\Geb}\|e^{t\Id\Gamma(\mu)/6}(T_{\Geb,\UV,t}\Psi)(\V{x})\|_{\sF}^{p_{\mu,t}}\Id\V{x}
\le c\|\Psi\|_{L^2(\Geb,\sF)}^{p_{\mu,t}},
\end{align*}
for all $\Psi\in L^2(\Geb,\sF)$, where the constant $c$ depends on
$\ee$, $\mu$, $t$, $V$, and $\Geb$. This shows that
$$
\|\cU_\Geb T_{\Geb,\UV,t}\cU_\Geb^*\Phi
\|_{L^{p_{\mu,t}}(\Geb\times\cQ,\Id\V{x}\otimes\GM)}
\le c\|\Phi\|_{L^{2}(\Geb\times\cQ,\Id\V{x}\otimes\GM)},
$$
for all $\Phi\in L^{2}(\Geb\times\cQ,\Id\V{x}\otimes\GM)$.
The aforementioned result of Gross now implies that $\|\cU_\Geb {T}_{\Geb,\UV,t}\cU_\Geb^*\|$
is an eigenvalue of $\cU_\Geb {T}_{\UV,t}\cU_\Geb^*$. Together with the spectral calculus this 
shows that $E_{\Geb,\UV}$ is an eigenvalue of $\NV_{\Geb,\UV}$. 
\end{proof}

%%%%%%%%%%%%%%%%%%%%%%%%%%%%%%%%%%%%%%%%%%%%%%%%%
%%%%%%%%%%%%%%%%%%%%%%%%%%%%%%%%%%%%%%%%%%%%%%%%%

\subsection{Exponential Localization}\label{ssecexploc}

\noindent
In this subsection we derive the first technical prerequisite for the compactness
arguments mentioned in the beginning of this section, namely estimates on the spatial localization
of elements in spectral subspaces below a localization threshold. (Later on we shall see examples 
of such subspaces other than $\{0\}$.) The general method applied to prove exponential localization of 
these spectral subspaces, which might belong to the continuous subspace of the considered operator, 
originates from \cite{BFS1998b}, has been further developed in \cite{Griesemer2004}, and was
used in several other articles (e.g. \cite{KMS2011,Panati2009}). In our case the implementation 
of the method requires a few extra arguments to get the strengthened bound \eqref{exot2} and to
cover exponential weight functions $F$ that are not linearly bounded. Notice also that everything done
in this subsection applies to the case $\UV=\infty$ right away.

Let $\UV\in[0,\infty]$ and $R>0$. Then we abbreviate
\begin{align*}
\Geb_R&:=\Geb\cap\{|\V{x}|>R\},
\end{align*}
and recall that, according to our earlier notation,
\begin{align*}
\Th{R}{\UV}&=\inf\sigma(\NV_{\Geb_R,\UV})=
\inf\sigma(\HV_{\Geb_R,\UV}),\quad\text{if $\Geb_R\not=\emptyset$.}
\end{align*}
We further set
\begin{align*}
\Th{R}{\UV}&:=\infty,\quad\text{if $\Geb_R=\emptyset$.}
\end{align*}
By ``localization threshold" we mean the generalized limit
\begin{align*}
\Sigma_{\Geb,\UV}&:=\lim_{R\to\infty}\Th{R}{\UV}\in(-\infty,\infty].
\end{align*}
It exists because $\Th{R}{\UV}$ is monotonically increasing in $R>0$.
The monotonicity in turn is obvious since $\sD(\Geb_R,\fdom(\Id\Gamma(\omega)))$,
defined as in \eqref{defsDsE}, is a form core for $\HV_{\Geb_R,\UV}$, if $\Geb_R$ is non-empty.
In the next proposition and sometimes later on we shall also use the notation
\begin{align}\label{defEnull}
E_{\tilde{\Geb},\UV}^0&:=E_{\tilde{\Geb},\UV}\;\,\text{in case that $V=0$, for any
open $\tilde{\Geb}\subset\RR^3$.}
\end{align}

Of course, the spectral subspaces of the Nelson operators corresponding to bounded open subsets
of $\RR^3$ are trivially localized. The crucial point about the bounds asserted in the next proposition
is that their right hand sides depend on $\UV$ only through the quantities $\Th{R}{\UV}$ 
and $E_{\Geb,\UV}$ and are uniform in the possibly bounded open subsets $\Geb'\subset\Geb$.
In the statement of the proposition the symbol $\sQ_{\Geb'}$ is defined by \eqref{defQGV} with
$\Geb'$ put in place of $\Geb$.

\begin{prop}\label{propexploc}
Assume that $\Geb\subset\RR^3$ is open and unbounded. 
Let $\Lambda\in[0,\infty]$, $\Geb'\subset\Geb$ be open, and suppose that
$\lambda\in\RR$ and $R>0$ satisfy
\begin{align}\label{defDelta}
\Delta&:=\Th{R}{\UV}-\lambda-\frac{4}{R^2}>0.
\end{align}
Let $F:\RR^3\to\RR$ be locally Lipschitz continuous and bounded from below and assume that
{\em one} of the following two bounds holds,
\begin{align}\label{bedFa}
\Th{R}{\UV}&\ge\frac{1}{2}|\nabla F|^2+\frac{4}{R^2}+\lambda,\quad\text{a.e. on $B_R^{c}$,}
\\\label{bedFb}
\Th{R}{\UV}^0+V&\ge\frac{1}{2}|\nabla F|^2+\frac{4}{R^2}+\lambda,\quad\text{a.e. on $B_R^{c}$,}
\end{align}
where $B_{r}$ is the open ball of radius $r>0$ about the origin in $\RR^3$.
Then, for every $\ve\in(0,1)$, the range of 
$1_{(-\infty,\lambda]}(\HV_{\Geb',\UV})$ is contained in the domain of $e^{(1-\ve)F}$ and
there exists a universal constant $c>0$ such that
\begin{align}\label{exot1}
\|e^{(1-\ve)F}1_{(-\infty,\lambda]}(\HV_{\Geb',\UV})\|
\le\frac{c}{\ve^2}\frac{(\Th{R}{\UV}-E_{\Geb,\UV}+1)^2}{(1\wedge\Delta)^2}\cdot C_F(\ve,R)&,
\end{align}
with
\begin{align*}
C_F(\ve,R)&:=\exp\Big({(1-\ve)\max_{\ol{B}_{2R}}F}\Big)
\big(1+\underset{B_{2R}}{\mathrm{ess\,sup}}|\nabla F|^2\big).
\end{align*}
Furthermore, $e^{(1-\ve)F}$ maps the range of $1_{(-\infty,\lambda]}(\HV_{\Geb',\UV})$ 
into $\sQ_{\Geb'}$ and
\begin{align}\nonumber
\|(\HV_{\Geb',\UV}-E_{\Geb',\UV})^\eh e^{(1-\ve)F}1_{(-\infty,\lambda]}(\HV_{\Geb',\UV})\|&
\\\label{exot2}
\le\frac{c}{\ve^{\nf{5}{2}}}\frac{(\Th{R}{\UV}-E_{\Geb,\UV}+1)^{\nf{5}{2}}}{(1\wedge\Delta)^2}
\cdot C_F(\ve,R)&.
\end{align}
The same assertions hold when the symbol $\HV_{\Geb',\UV}$ is replaced by $\NV_{\Geb',\UV}$
everywhere.
\end{prop}

\begin{proof}
A detailed proof of this proposition is given in Appendix~\ref{appexploc}. Let us mention that
the {\em finiteness} of the left hand side of \eqref{exot1} is (almost) a direct consequence
of \eqref{IMSq}, \eqref{exp5}, and \cite[Thm.~1]{Griesemer2004}, at least if $F$ is not increasing
faster than linearly. The precise form of the upper bound in \eqref{exot1} essentially 
follows from analyzing the proofs in \cite{Griesemer2004} (see also \cite{Panati2009}), while the 
proof of \eqref{exot2} requires additional arguments.
\end{proof}

We finish this subsection by extracting Theorems~\ref{introthmexpdec},~\ref{introthmsupexp},
and~\ref{introthmdecnonFock} from the previous proposition.

\begin{proof}[Proof of Theorem~\ref{introthmexpdec}.]
We drop the subscripts $\Geb=\RR^3$ and $\UV=\infty$ so that 
$\NV=\NV_{\RR^3,\infty}$ and $\Sigma=\Sigma_{\RR^3,\infty}$. 
(We only have to consider $\Geb=\RR^3$, 
but the proof obviously works for general open $\Geb$.) 

Let $\lambda\in\RR$, suppose that $\Psi$ is a normalized element of the range of
$1_{(-\infty,\lambda]}(\NV)$, and let $\Psi(\cdot)$ be its continuous representative. We write
\begin{align}\label{piet}
\Psi=e^{-t\NV}1_{(-\infty,\lambda]}(\NV)\big\{e^{t\NV}1_{(-\infty,\lambda]}(\NV)\big\}\Psi.
\end{align}
H\"{o}lder's inequality, Proposition~\ref{propGSmu}(2), and \eqref{piet} with $t=6r/\delta$ imply
\begin{align*}
\|e^{r\Id\Gamma(\omega\wedge1)}\Psi(\V{x})\|
&\le\|e^{r\Id\Gamma(\omega\wedge1)/\delta}\Psi(\V{x})\|^\delta\|\Psi(\V{x})\|^{1-\delta}
\le ce^{6r\lambda}\|\Psi(\V{x})\|^{1-\delta},
\end{align*}
for all $\V{x}\in\RR^3$, $\delta\in(0,1)$, $r>0$, and some $c>0$ depending only on $\delta$, $r$, $\ee$, and $V$. 
Therefore, it remains to treat the case $r=0$.

Given $\lambda<\sigma\le\Sigma$ and $\ve>0$, we put
$F(\V{x}):=(1-\ve/2)\sqrt{2\sigma-2\lambda}|\V{x}|$, $\V{x}\in\RR^3$. Then \eqref{defDelta} and 
\eqref{bedFa} are fulfilled for sufficiently large $R>0$. Define $\ve'>0$ such that
$(1-\ve')(1-\ve/2)=(1-\ve)$. Choosing $t=1$ in \eqref{piet} we find, for all $\V{x}\in\RR^3$, 
\begin{align}\label{piet2}
\|e^{(1-\ve')F(\V{x})}\Psi(\V{x})\|
&\le  \|e^{(1-\ve')F}e^{-\NV}e^{(\ve'-1)F}\|_{2,\infty}
\|e^{(1-\ve')F}1_{(-\infty,\lambda]}(\NV)\|e^\lambda.
\end{align}
Applying \eqref{exot1} with $\ve'$ and $\NV$ put in place 
of $\ve$ and $\HV_{\Geb',\UV}$, respectively, and using \eqref{hcbd}, we see that the
right hand side of \eqref{piet2} is indeed well-defined and finite.
\end{proof}

\begin{proof}[Proof of Theorem~\ref{introthmsupexp}.]
Again we drop the subscripts for $\Geb=\RR^3$ and $\UV=\infty$. 
Furthermore, we write $E_R^0:=E_{\{|\V{x}|>R\},\infty}^0$ for short.
By the same argument as in the preceding proof of 
Theorem~\ref{introthmexpdec} it suffices to treat the case $r=0$.

Given $\ve>0$, we set $F(\V{x}):=(1-\ve/3)a|\V{x}|^{p+1}/(p+1)$, $\V{x}\in\RR^3$.
For a sufficiently large $R\ge\rho$, that we keep fixed in what follows, the condition \eqref{defDelta} is satisfied and
\begin{align*}
V(\V{x})+E_R^0\ge\frac{a^2}{2}|\V{x}|^{2p}-b+E_R^0
\ge\frac{1}{2}|\nabla F(\V{x})|^2+\frac{4}{R^2}+\lambda,\quad|\V{x}|\ge R.
\end{align*}
Let $\Psi$ be the continuous representative of a normalized element of the range of
$1_{(-\infty,\lambda]}(\NV)$ and let $\V{y}\in\RR^3\setminus\{0\}$.
We shall apply \eqref{hcbd} with the Lipschitz continuous weight function given by
\begin{align*}
F_{\V{y}}(\V{x}):=\left\{\begin{array}{ll}
F(\V{x}),&\text{if $|\V{x}|\le|\V{y}|$,}
\\
F(\V{y}),&\text{if $|\V{x}|>|\V{y}|$.}
\end{array}\right.
\end{align*}
A Lipschitz constant for $F_{\V{y}}$ obviously is $L:=1\vee[(1-\ve/3)a|\V{y}|^p]$.
We further choose $\ve'>0$ such that $(1-\ve')(1-\ve/3)=(1-2\ve/3)$.
Writing $\Psi$ as in \eqref{piet} with $t=1/L^2$, applying an analogue of \eqref{piet2},
and using  \eqref{hcbd} and \eqref{exot1} afterwards, we find some $c>0$ such that
\begin{align*}
\|e^{(1-\ve')F_{\V{y}}(\V{y})}\Psi(\V{y})\|&\le c(1\vee|\V{y}|)^{3p/2}.
\end{align*}
Notice that we chose $t$ such that the exponential on the right hand side of \eqref{hcbd} equals 
$e^{6tL^2}=e^6$. We also used that the constant appearing in \eqref{hcbd} is non-decreasing in the 
time parameter and that $t=1/L^2\le1$. Finally, we used that 
$\|e^{t\NV}1_{(-\infty,\lambda]}(\NV)\|\le e^{\lambda/L^2}\le e^{0\vee\lambda}$.
We thus arrive at
\begin{align*}
\|\Psi(\V{y})\|&\le c(1\vee|\V{y}|)^{3p/2}e^{-(1-2\ve/3)a|\V{y}|^{p+1}/(p+1)}
\le c'e^{-(1-\ve)a|\V{y}|^{p+1}/(p+1)},
\end{align*}
where $\V{y}$ was an arbitrary element of $\RR^3\setminus\{0\}$ and $c'>0$ is
$\V{y}$-independent.
\end{proof}

\begin{proof}[Proof of Theorem~\ref{introthmdecnonFock}.]
To infer Theorem~\ref{introthmdecnonFock} from Proposition~\ref{propexploc}
we merely have to replace the symbol $\NV$ by $\HV$ in the above proofs
of Theorem~\ref{introthmexpdec} and Theorem~\ref{introthmsupexp}.
Notice that \eqref{hcbd} can be applied with $\wt{T}_{\RR^3,\infty,t}$
put in place of $T_{\Geb,\UV,t}$.
\end{proof}

%%%%%%%%%%%%%%%%%%%%%%%%%%%%%%%%%%%%%%%%%%%%%%%
%%%%%%%%%%%%%%%%%%%%%%%%%%%%%%%%%%%%%%%%%%%%%%%

\subsection{The Binding Condition}\label{ssecbinding}

\noindent
Proposition~\ref{propexploc} is non-trivial only if the following {\em binding condition} is fulfilled,
\begin{align}\label{bindcond}
\Sigma_{\Geb,\UV}>E_{\Geb,\UV}.
\end{align}
We emphasize once more that both the ionization threshold on the left hand side of \eqref{bindcond}
and the infimum of the spectrum on the right hand side are the same for 
$\NV_{\Geb,\UV}$ and for $\HV_{\Geb,\UV}$ and this holds for all $\UV\in[0,\infty]$.
 
\begin{ex}\label{exbindcond} Let $\Lambda\in[0,\infty]$.
Then the binding condition 
\begin{align}\label{bindcondRR3} 
\Sigma_{\RR^3,\UV}>E_{\RR^3,\UV}
\end{align}
holds in the following two cases: 
\begin{enumerate}
\item[(1)] If the potential is confining, i.e., $V(\V{x})\to\infty$, $|\V{x}|\to\infty$, then we
obviously have $\Sigma_{\RR^3,\UV}=\infty$, so that \eqref{bindcondRR3} is trivially satisfied.
\item[(2)] For finite $\UV$, an argument in \cite[Thm.~3.1]{GLL2001} applied to 
$\NV_{\RR^3,\UV}$ yields
\begin{align}\label{bindbd}
E_{\RR^3,\UV}^0+\inf\sigma(S_{\RR^3})\ge E_{\RR^3,\UV}.
\end{align}
Here $E_{\RR^3,\UV}^0$ is defined in \eqref{defEnull} and $S_{\RR^3}$ is the
ordinary Schr\"{o}dinger operator with potential $V$.
The bound \eqref{bindbd} extends to the case $\UV=\infty$ by the norm resolvent convergence
$\NV_{\RR^3,\UV}\to\NV_{\RR^3,\infty}$, $\UV\to\infty$, which also holds for $V=0$, 
of course. Therefore, \eqref{bindcondRR3} is fulfilled, if
\begin{align}\label{mogens70}
\Sigma_{\RR^3,\UV}&\ge E_{\RR^3,\UV}^0+\inf\sigma_{\mathrm{ess}}(S_{\RR^3}),
\quad\UV\in[0,\infty],
\end{align}
where $\sigma_{\mathrm{ess}}$ denotes the essential spectrum, and
\begin{align*}
\inf\sigma_{\mathrm{ess}}(S_{\RR^3})&>\inf\sigma(S_{\RR^3}).
\end{align*}
Here \eqref{mogens70} is easily verified when $V(\V{x})\to\inf\sigma_{\mathrm{ess}}(S_{\RR^3})$,
as $|\V{x}|\to\infty$, by working with the quadratic forms $\hv_{\{|\V{x}|>R\},\UV}$.
\end{enumerate}
\end{ex}

In what follows it will be convenient to put a hat $\,\hat{{}}\,$ on top of $\V{x}$ when it should be 
regarded as a multiplication operator rather than a variable. Furthermore, we shall
employ the common notation 
$$
\langle\V{x}\rangle:=(1+|\V{x}|^2)^\eh,\quad\V{x}\in\RR^3.
$$

\begin{lem}\label{lemeva}
Let $0\le\UV<\infty$, $\ve\in(0,1]$, and $B\subset\Geb$ be some open ball. 
Assume that the binding condition \eqref{bindcond} holds. Choose $R\ge1$ such that
$$
\Delta:=\Th{R}{\UV}-\frac{4}{R^2}-\lambda=\frac{1}{2}(\Th{R}{\UV}-E_{\Geb,\UV})-\frac{4}{R^2}>0,
$$
with
\begin{align*}%\label{deflambdaexploc}
\lambda&:= E_{\Geb,\UV}+\frac{1}{2}(\Th{R}{\UV}-E_{\Geb,\UV}),
\end{align*}
 and abbreviate
\begin{align*}
\alpha:=(1-\ve)\sqrt{2\Delta},\quad M:=\frac{e^{\alpha(1+2R)}}{\ve^{\nf{5}{2}}}
\frac{(\Th{R}{\UV}-E_{\Geb,\UV}+1)^{\nf{7}{2}}}{(\Delta\wedge1)^2}.
\end{align*}
Pick any open subset $\Geb'\subset\Geb$ and assume that:
\begin{enumerate}
\item[{\rm(a)}] $E_{\Geb',\UV}\le\lambda$.
\item[{\rm(b)}] $\Geb'$ contains $B$.
\item[{\rm(c)}] $\HV_{\Geb',\UV}$ has a normalized ground state eigenvector $\Phi_{\Geb',\UV}$.
\end{enumerate}
 Then $\Phi_{\Geb',\UV}$, the components of the weak gradient $\nabla\Phi_{\Geb',\UV}$, and 
the components of $a(e_{\hat{\V{x}}}\V{m}\beta_{\UV})\Phi_{\Geb',\UV}$ are in the domain of 
$e^{\alpha\langle\hat{\V{x}}\rangle}$ and 
\begin{align}\label{eva1}
\|e^{\alpha\langle\hat{\V{x}}\rangle}\Phi_{\Geb',\UV}\|&\le cM,
\\\label{eva1alpha}
\|e^{\alpha\langle\hat{\V{x}}\rangle}\nabla\Phi_{\Geb',\UV}\|&\le c'(1+\alpha)M,
\\\label{eva1beta}
\|e^{\alpha\langle\hat{\V{x}}\rangle}a(e_{\hat{\V{x}}}\V{m}\beta_{\UV})\Phi_{\Geb',\UV}\|
&\le c'M,
\end{align}
where $c>0$ is a universal constant and $c'>0$ depends only on $B$, $\ee$, and $V$.
In fact, $c'$ is locally bounded in $\ee\in\RR$ when $B$ and $V$ are held fixed.
\end{lem}

Notice that $M$ depends on $\ve$, $R$, and on all model parameters $\mu$, $\ee$, $\eta$,
$V$, $\Geb$, $\UV$, but only through the quantities $E_{\Geb,\UV}$ and $\Th{R}{\UV}$, 
which can be controlled in many relevant situations.

\begin{proof}
We shall apply Proposition~\ref{propexploc} with $F(\V{x})=\sqrt{2\Delta}\langle{\V{x}}\rangle$, observing
that \eqref{defDelta} and \eqref{bedFa} are satisfied under the present assumptions.
The bound \eqref{eva1} is then a direct consequence of \eqref{exot1} and
$|\nabla F|^2\le2\Delta\le\Th{R}{\UV}-E_{\Geb,\UV}$.
Furthermore, since $e^{\alpha\langle\hat{\V{x}}\rangle}\Phi_{\Geb',\UV}\in\sQ_{\Geb'}$ 
by Proposition~\ref{propexploc}, we may insert it into \eqref{pernille0}, which yields
\begin{align}\nonumber
&\frac{1}{2}\int_{\Geb'}\|\nabla(e^{\alpha\langle\hat{\V{x}}\rangle}\Phi_{\Geb',\UV})(\V{x})
\|^2\Id\V{x}+\int_{\Geb'}e^{2\alpha\langle{\V{x}}\rangle}\|\Id\Gamma(\omega)^\eh
\Phi_{\Geb',\UV}(\V{x})\|^2\Id\V{x}
\\\label{eva2}
&\le c\|(\HV_{\Geb',\UV}-E_{\Geb',\UV})^\eh e^{\alpha\langle\hat{\V{x}}\rangle}
\Phi_{\Geb',\UV}\|^2
+c(E_{\Geb',\UV}+c')\|e^{\alpha\langle\hat{\V{x}}\rangle}\Phi_{\Geb',\UV}\|^2,
\end{align}
where $c,c'>0$ depend only on $\ee$ and $V$. Recall that $\epsilon(0)$ is the vacuum 
vector in $\sF$. In view of \eqref{defvKL} and \eqref{forhtv} we then have the upper bound
\begin{align}\label{ubEDir}
E_{\Geb',\UV}&\le\inf_{g\in C_0^\infty(B),\|g\|=1}\mathfrak{h}_{\Geb',\UV}[g\epsilon(0)]
=\inf_{g\in C_0^\infty(B),\|g\|=1}\mathfrak{s}_{B}[g],
\end{align}
where the rightmost quantity obviously depends on $B$ and $V$ only.
For measurable $\Psi:\Geb'\to\fdom(\Id\Gamma(\omega))$, we further have the $\UV$-uniform bound
\begin{align}\label{eva4}
\int_{\Geb'}e^{2\alpha\langle{\V{x}}\rangle}\|a(e_{\V{x}}\V{m}\beta_{\UV})\Psi(\V{x})\|^2\Id\V{x}
&\le c''\ee^2\int_{\Geb'}e^{2\alpha\langle{\V{x}}\rangle}
\|\Id\Gamma(\omega)^\eh\Psi(\V{x})\|^2\Id\V{x},
\end{align}
with a universal constant $c''$ satisfying $2\||\V{m}|^\eh\beta_{\infty}\|^2\le c''\ee^2$. Finally,
\begin{align}\label{daniel}
e^{\alpha\langle{\V{x}}\rangle}\partial_{x_j}\Psi&=\partial_{x_j}(e^{\alpha\langle{\V{x}}\rangle}
\Psi)-\alpha\frac{x_j}{\langle{\V{x}}\rangle}e^{\alpha\langle{\V{x}}\rangle}\Psi,
\quad\text{in $L^1_\loc(\Geb',\sF)$,}
\end{align}
for all $\Psi\in{W}^{1,2}(\Geb',\sF)$ and $ j\in\{1,2,3\}$. In conjunction with \eqref{exot2}, 
\eqref{eva1}, \eqref{eva2}, and \eqref{ubEDir} the relations \eqref{eva4} and \eqref{daniel} imply  
\eqref{eva1beta} and \eqref{eva1alpha}, respectively.
\end{proof}

\begin{rem}\label{remcstlokbd}
Independent of whether a binding condition holds or not, the bounds \eqref{eva2},
\eqref{ubEDir}, and \eqref{eva4} in the previous proof are still valid when $\alpha=0$.
For all finite $\UV$ and all open $\Geb'\subset\RR$ satisfying conditions (b) and (c)
of Lemma~\ref{lemeva}, we thus find
\begin{align}\label{eva1anull}
\|\nabla\Phi_{\Geb',\UV}\|&\le c,\quad
\|a(e_{\hat{\V{x}}}\V{m}\beta_{\Lambda})\Phi_{\Geb',\UV}\|\le c.
\end{align}
Here the constant $c$ depends only on $B$, $\ee$, and $V$ and it is
locally bounded in $\ee\in\RR$.
\end{rem}

%%%%%%%%%%%%%%%%%%%%%%%%%%%%%%%%%%%%%%%%%%%%%%%
%%%%%%%%%%%%%%%%%%%%%%%%%%%%%%%%%%%%%%%%%%%%%%%

\subsection{Infrared Behavior}\label{ssecIR}

\noindent
Next, we assume, for some finite $\UV$, that $\Phi_{\Geb,\UV}$ is a ground state eigenvector
of $\HV_{\Geb,\UV}$. We shall derive a formula for $a\Phi_{\Geb,\UV}$, where $a$ denotes
the ``pointwise" annihilation operator. It can be defined as follows. For every
$n\in\NN$, $n\ge2$, there is a canonical isomorphism 
$\tilde{I}_n:L^2(\RR^{3n})\to L^2(\RR^3,L^2(\RR^{3(n-1)}))$, whose restriction to
$L_{\mathrm{sym}}^2(\RR^{3n})$ we denote by $I_n$. The symbol $I_1$ will denote the
identity on $L^2(\RR^3)$. Defining $\dom(a):=L^2(\Geb,\fdom(\Id\Gamma(\omega)))$ and
\begin{align*}
\big((a\Psi)(\V{k})\big)(\V{x})
:=\big((n+1)^\eh[I_{n+1}( \Psi(\V{x})^{(n+1)})](\V{k})\big)_{n=0}^\infty\in\sF,
\end{align*}
for representatives of equivalence classes and $\Psi\in\dom(a)$, we obtain a well-defined map
$a:\dom(a)\to L^2(\RR^3,L^2(\Geb,\sF);\omega^\eh\Id\V{k})$. 
We have indeed the well-known relation
\begin{align}\label{fordGak}
\int_{\Geb}\|\Id\Gamma(\omega)^\eh\Psi(\V{x})\|^2\Id\V{x}
=\int_{\RR^3}\omega(\V{k})\|(a\Psi)(\V{k})\|^2\Id\V{k},\quad\Psi\in\dom(a).
\end{align}

The proof of the next proposition is a suitable version of a well-known commutator, or, 
``pull-through" argument; compare, e.g., \cite{Froehlich1974}. Instead of commuting the 
Hamiltonian with an annihilation operator we shall, however, commute it with the direct integral 
operators in \eqref{defdirintopIR}. Here the inclusion of the second expression under the direct integral 
is inspired by \cite{BFS1999}. It is used to control infrared singularities showing up in the 
computations: In fact, the fractions in \eqref{defB2} and \eqref{defB3} are bounded by $|\V{k}|$, for 
small $|\V{k}|$, thanks to the term $-\chi(\V{k})$ coming from the second expression under the direct 
integral.

\begin{prop}\label{prop-IR}
Let $0\le\UV<\infty$. Assume that the binding condition \eqref{bindcond} holds and that
$\Phi_{\Geb,\UV}\in\dom(\HV_{\Geb,\UV})$ is a ground state eigenvector of 
$\HV_{\Geb,\UV}$. We abbreviate
\begin{align}\label{defUpsilon1}
\V{\Upsilon}_{1}&:=i\hat{\V{x}}\Phi_{\Geb,\UV},
\\\label{defUpsilon2}
\V{\Upsilon}_{2}&:=\langle\hat{\V{x}}\rangle(-i\nabla+a(e_{\hat{\V{x}}}\V{m}\beta_{\UV}))
\Phi_{\Geb,\UV},
\\\label{defUpsilon3}
\Upsilon_{3}&:=\langle\hat{\V{x}}\rangle\Phi_{\Geb,\UV},
\end{align}
observing that these vectors are well-defined elements of $L^2(\Geb,\sF)$ by Lemma~\ref{lemeva}.
Let $\chi\in C_0^\infty(\RR^3)$ satisfy $\chi(\V{k})=1$, if $|\V{k}|\le1$, and
$\chi(\V{k})=0$, if $|\V{k}|\ge2$. For every $\V{k}\in\RR^3\setminus\{0\}$, we further write
$$
R_{\Geb,\UV}(\V{k}):=(\HV_{\Geb,\UV}-E_{\Geb,\UV}+\omega(\V{k}))^{-1},
$$ 
and we introduce the following bounded operators on $L^2(\Geb,\sF)$,
\begin{align}\label{defB1}
\Xi_{1}(\V{k})&:=(\HV_{\Geb,\UV}-E_{\Geb,\UV})R_{\Geb,\UV}(\V{k})\chi(\V{k}),
\\\label{defB2}
\Xi_{2}(\V{k})&:=R_{\Geb,\UV}(\V{k})
\frac{e^{-i\V{k}\cdot\hat{\V{x}}}-\chi(\V{k})}{\langle\hat{\V{x}}\rangle},
\\\nonumber
\V{\Xi}_{3}(\V{k})&:=(\HV_{\Geb,\UV}-E_{\Geb,\UV}+1)^\eh R_{\Geb,\UV}(\V{k})
\\\label{defB3}
&\quad\times\{a(e_{\hat{\V{x}}}\V{m}\beta_{\UV})
(\HV_{\Geb,\UV}-E_{\Geb,\UV}+1)^{\mh}\}^*
\frac{e^{-i\V{k}\cdot\hat{\V{x}}}-\chi(\V{k})}{\langle\hat{\V{x}}\rangle}.
\end{align}
Then the following identity holds, for a.e. $\V{k}\in\RR^3$,
\begin{align}\label{aki1}
(a\Phi_{\Geb,\UV})(\V{k})
&=-\sum_{j=1}^2\Xi_{j}(\V{k}){\beta_{\UV}(\V{k})}\V{k}\cdot\V{\Upsilon}_{j}
-\V{\Xi}_{3}(\V{k})\cdot\V{k}{\beta_{\UV}(\V{k})}\Upsilon_{3}.
\end{align}
\end{prop}

In the preceding statement and below we are using shorthands analogous to \eqref{vecnot0}.
For instance, the expression in \eqref{defB3} actually is a triplet of operators comprising one
bounded operator for each component $m_j$ of $\V{m}$.

\begin{proof}
Let $\V{p}\in\RR^3\setminus\{{0}\}$ and let $g\in L^2(\RR^3)$ have a compact support in 
$\RR^3\setminus\{{0}\}$, so that $\dom(\Id\Gamma(\omega)^\eh)\subset\dom(a(g))$.
Furthermore, let $\tilde{\vt}:\RR\to\RR$ be smooth and such that $0\le\tilde{\vt}\le1$, 
$\tilde{\vt}=1$ on $(-\infty,1]$, and $\tilde{\vt}=0$ on $[2,\infty)$. Set
$\vt_r(\V{x}):=\tilde{\vt}(r^{-1}\ln\langle\V{x}\rangle)$, for all $\V{x}\in\RR^3$ and $r\ge1$. Then 
$$
|\nabla\vt_r(\V{x})|\le\frac{1}{r}\frac{c}{\langle\V{x}\rangle},\quad 
|\partial_j^2\vt_r(\V{x})|\le\frac{1}{r}\frac{c}{\langle\V{x}\rangle^2}, \quad j\in\{1,2,3\},
\quad\V{x}\in\RR^3,
$$  
for some $r$-independent constant $c>0$. We finally introduce the direct integral operators
\begin{align}\label{defdirintopIR}
A_r(g):=\int_{\Geb}^\oplus(a(g)+\vt_r(\V{x})
\SPn{g}{i\chi\V{x}\cdot\V{m}\beta_{\UV}})\Id\V{x},\quad r\ge1,
\end{align}
and denote by $A(g)$ the operator obtained upon putting $1$ in place of $\vt_r$
in the preceding formula.

For every $j\in\{1,2,3\}$, the expression $x_j\vt_r(\V{x})$ defines an element of $C_0^\infty(\RR^3)$.
Hence, we know that multiplication with it leaves $\dom(S_\Geb)$ invariant and 
\begin{align}\label{adriane1}
\V{x}\vt_rS_\Geb v-S_\Geb\V{x}\vt_rv&=(\vt_r +\V{x}\vt_r')\nabla v
+\frac{1}{r}\V{\ve}_rv,\quad v\in\dom(S_\Geb)\subset\mr{W}^{1,2}(\Geb),
\end{align}
where $\V{\ve}_r\in C_0^\infty(\RR^3,\RR^3)$ satisfies $\sup_{r\ge1}\|\V{\ve}_r\|_\infty<\infty$ and
$\mr{W}^{1,2}(\Geb)$ denotes the closure of $C_0^\infty(\Geb)$ with respect to the norm on $W^{1,2}(\Geb)$.
Let $\mho\in\mathrm{span}\{v\epsilon(h):v\in\dom(S_\Geb),\,h\in\dom(\omega)\}$. Then, 
by virtue of Proposition~\ref{propdomHKLV} and \eqref{adriane1}, the following computation is justified,
\begin{align}\nonumber
&\big(A_r(g)(\HV_{\Geb,\UV}-E_{\Geb,\UV})\mho-
(\HV_{\Geb,\UV}-E_{\Geb,\UV})A_r(g)\mho\big)(\V{x})
\\\nonumber
&=a(\omega g)\mho(\V{x})-i\SPn{g}{(e_{\V{x}}-\vt_r(\V{x})\chi)\V{m}\beta_{\UV}}
\cdot(\nabla+i\vp(e_{\V{x}}\V{m}\beta_{\UV}))\mho(\V{x})
\\\nonumber
&\quad+i\SPn{g}{\chi\V{x}\cdot\V{m}\beta_{\UV}}(\nabla\vt_r)(\V{x})\cdot
(\nabla+i\vp(e_{\V{x}}\V{m}\beta_\UV))\mho(\V{x})
\\\label{adriane2}
&\quad+\frac{1}{r}\V{\ve}_r(\V{x})\cdot\SPn{g}{i\chi\V{m}\beta_\UV}\mho(\V{x}),
\quad\text{a.e. $\V{x}\in\Geb$.}
\end{align}
Next, we scalar multiply the above expression with
$R_{\Geb,\UV}(\V{p})\Psi$, for arbitrary $\Psi\in\sD(\Geb,\fdom(\Id\Gamma(\omega)))$, 
and re-write the left hand side containing the commutator as 
\begin{align}\label{adriane3}
\SPn{A_r(g)^*R_{\Geb,\UV}(\V{p})\Psi}{(\HV_{\Geb,\UV}-E_{\Geb,\UV})\mho}
-\SPn{A_r(g)^*(\HV_{\Geb,\UV}-E_{\Geb,\UV})R_{\Geb,\UV}(\V{p})\Psi}{\mho}.
\end{align}
After that we pass to the limit $r\to\infty$. Then the operator $A_r(g)$ gets replaced by $A(g)$
in \eqref{adriane3}. Here we take into account that
$R_{\Geb,\UV}(\V{p})\Psi\in\dom(\langle\hat{\V{x}}\rangle)\cap\QGV$ by \eqref{resexploc2} 
(where $V-E_{\Geb,\UV}+\omega(\V{p})$ should be put in place of $V$) and therefore also 
$$
(\HV_{\Geb,\UV}-E_{\Geb,\UV})R_{\Geb,\UV}(\V{p})\Psi
=\Psi-\omega(\V{p})R_{\Geb,\UV}(\V{p})\Psi\in\dom(\langle\hat{\V{x}}\rangle)\cap\QGV.
$$
The contribution of the third and fourth lines of \eqref{adriane2} vanish in the limit $r\to\infty$ due to 
the properties of $\vt_r$ and $\V{\ve}_r$. In the next step we replace $\mho$ by $\Phi_{\Geb,\UV}$, 
which is possible since, according to Proposition~\ref{propdomHKLV}, $\mho$ can be chosen in an operator
core for $\HV_{\Geb,\UV}$. Notice that, if 
$\mho_n\in\mathrm{span}\{v\epsilon(h):v\in\dom(S_\Geb),\,h\in\dom(\omega)\}$, $n\in\NN$,
converge to $\Phi_{\Geb,\UV}$ in the graph norm of 
$\HV_{\Geb,\UV}$, as $n\to\infty$, then also
\begin{align*}
a(\omega g)\mho_n\to a(\omega g)\Phi_{\Geb,\UV},\quad
(\nabla+i\vp(e_{\hat{\V{x}}}\V{m}\beta_{\UV}))\mho_n
\to(\nabla+i\vp(e_{\hat{\V{x}}}\V{m}\beta_{\UV}))\Phi_{\Geb,\UV},
\end{align*} 
by \eqref{rba} and the formulas for $\hv_{\Geb,\UV}$.
Since $\Psi$ was chosen in a dense subset of $L^2(\Geb,\sF)$, this procedure eventually results in
\begin{align}\label{andi1}
A(g){\Phi_{\Geb,\UV}}
&=R_{\Geb,\UV}(\V{p})\int_{\RR^3}\bar{g}(\V{p})\Theta_{\V{p}}(\V{k})\Id\V{k},
\end{align}
where the integrand $\Theta_{\V{p}}:\RR^3\to L^2(\Geb,\sF)$ of the 
$L^2(\Geb,\sF)$-valued Bochner-Lebesgue integral is given by
\begin{align*}
\Theta_{\V{p}}(\V{k})&:=i\V{k}(e^{-i\V{k}\cdot\hat{\V{x}}}-\chi(\V{k})){\beta}_{\UV}(\V{k})
\cdot(\nabla+i\vp(e_{\hat{\V{x}}}\V{m}\beta_{\UV}))\Phi_{\Geb,\UV}
\\
&\quad
+i\omega(\V{p})\chi(\V{k}){\hat{\V{x}}}\cdot\V{k}\beta_{\UV}(\V{k})\Phi_{\Geb,\UV}
+(\omega(\V{p})-\omega(\V{k}))(a\Phi_{\Geb,\UV})(\V{k}).
\end{align*}
Now we choose $g(\V{k}):=1_{\cC_n(\V{p})}(\V{k})$ with
$\cC_n(\V{p}):=\{\V{k}\in\RR^3:\|\V{k}-\V{p}\|_{\infty}<1/n\}$ and $n\in\NN$ sufficiently large
such that $1/n<|\V{p}|/\sqrt{3}$. Taking into account that $\langle\hat{\V{x}}\rangle\Phi_{\Geb,\UV}$
is in $L^2(\Geb,\sF)$ and using that $|\omega(\V{p})-\omega(\V{k})|\le\sqrt{3}/n$ 
for all $\V{k}\in \cC_n(\V{p})$,
it is then easy to see that the term on the right hand side of \eqref{andi1} converges to
$R_{\Geb,\UV}(\V{p})\Theta_{\V{p}}(\V{p})$, as $n\to\infty$, provided that $\V{p}$ is a 
Lebesgue point of $\|(a\Phi_{\Geb,\UV})(\cdot)\|\in L^1_\loc(\RR^3\setminus\{0\})$.
In what follows we will further suppose that $\V{p}$ is a Lebesgue point of
$a\Phi_{\Geb,\UV}\in L^1_\loc(\RR^3\setminus\{{0}\},L^2(\Geb,\sF))$, recalling that the
Lebesgue point theorem also holds for the Bochner-Lebesgue integral; see, e.g., 
\cite[Cor.~1 on p.~87]{HillePhillips1957}. Then the left hand side of \eqref{andi1}
converges to the left hand side of the following identity
$$
(a\Phi_{\Geb,\UV})(\V{p})+i\chi(\V{p})\hat{\V{x}}\cdot\V{p}\beta_{\UV}(\V{p})
\Phi_{\Geb,\UV}=R_{\Geb,\UV}(\V{p})\Theta_{\V{p}}(\V{p}).
$$
Altogether we see that the previous identity holds for a.e. $\V{p}$. Finally, we re-write the vector
$\vp(e_{\hat{\V{x}}}\V{m}\beta_{\UV})\Phi_{\Geb,\UV}$ 
appearing in $\Theta_{\V{p}}(\V{p})$ as the sum of
$\ad(e_{\hat{\V{x}}}\V{m}\beta_{\UV})\Phi_{\Geb,\UV}$ and
$a(e_{\hat{\V{x}}}\V{m}\beta_{\UV})\Phi_{\Geb,\UV}$, and employ the relations
\begin{align*}
&R_{\Geb,\UV}(\V{p})\ad(e_{\hat{\V{x}}}\V{m}\beta_{\UV})
\\
&=(\HV_{\Geb,\UV}-E_{\Geb,\UV}+1)^\eh R_{\Geb,\UV}(\V{p})
(\HV_{\Geb,\UV}-E_{\Geb,\UV}+1)^\mh \ad(e_{\hat{\V{x}}}\V{m}\beta_{\UV})
\\
&\subset(\HV_{\Geb,\UV}-E_{\Geb,\UV}+1)^\eh R_{\Geb,\UV}(\V{p})
\{a(e_{\hat{\V{x}}}\V{m}\beta_{\UV})(\HV_{\Geb,\UV}-E_{\Geb,\UV}+1)^\mh\}^*
\end{align*}
to conclude.
\end{proof}

%%%%%%%%%%%%%%%%%%%%%%%%%%%%%%%%%%%%%%%%%%%%%%%
%%%%%%%%%%%%%%%%%%%%%%%%%%%%%%%%%%%%%%%%%%%%%%%

\subsection{Compactness of Families of Ground State Eigenvectors}\label{sseccomparg}

\noindent
The next proposition is an adaption of the well-known characterization of compact sets in
$L^2(\RR^d)$ to the Hilbert space $L^2(\RR^3,\sF)$. 
The cutoff function appearing in its statement takes
care of possible infrared singularities in the boson momenta; by choosing more complicated 
cutoffs one could in principle allow for singularities along lower-dimensional
sets that are more complex than $\{0\}$.
The proposition has implicitly been used in \cite{Matte2016} as a substitute for
an argument based on the Rellich-Kondrachov theorem in \cite{GLL2001}. The latter requires
technically more cumbersome photon derivative bounds as an input.
For the convenience of the reader a proof of the next proposition is given in Appendix~\ref{appcpt}.

\begin{prop}\label{propcpt}
Let $\{\Phi_\iota\}_{\iota\in\sI}$ be a bounded family of vectors in $\sF$ or $L^2(\RR^3,\sF)$.
Pick a cutoff function $\tilde{\vr}\in C^\infty(\RR,\RR)$ with 
$0\le\tilde{\vr}\le1$, $\tilde{\vr}=0$ on $(-\infty,1]$ and $\tilde{\vr}=1$ on $[2,\infty)$. 
Set $\vr_\delta(\V{k}):=\tilde{\vr}(|\V{k}|/\delta)$, for all $\V{k}\in\RR^3$ and $\delta\in(0,1]$. Define
\begin{align*}
N(r_0,r_1)&:=\sup_{\iota\in\sI}\int_{\{r_0\le|\V{k}|\le r_1\}}\|(a\Phi_\iota)(\V{k})\|^2\Id\V{k},
\\
\triangle_\delta(\V{h})&:=\sup_{\iota\in\sI}\int_{\RR^3}
\|(\vr_\delta a\Phi_\iota)(\V{k})-(\vr_\delta a\Phi_\iota)(\V{k}+\V{h})\|^2\Id\V{k},
\end{align*}
for all $0\le r_0<r_1\le\infty$ and $(\delta,\V{h})\in(0,1]\times\RR^3$, respectively. Assume that
\begin{align}\label{cpt1a}
N(0,\infty)<\infty,\quad\lim_{r_0\to\infty}N(r_0,\infty)&=0,\quad\lim_{r_1\downarrow0}N(0,r_1)=0,
\\\label{cpt1b}
\forall\delta\in(0,1]:\quad\lim_{\V{h}\to{0}}\triangle_\delta(\V{h})&=0.
\end{align}
If the Hilbert space $L^2(\RR^3,\sF)$ is considered, assume in addition that
\begin{align}\label{cpt2a}
\lim_{R\to\infty}\sup_{\iota\in\sI}\int_{\{|\V{x}|\ge R\}}\|\Phi_\iota(\V{x})\|^2\Id\V{x}&=0,
\\\label{cpt2b}
%\forall \alpha\in\NN:\quad
\lim_{\V{y}\to0}\sup_{\iota\in\sI}\int_{\RR^3}\|\Phi_\iota(\V{x}+\V{y})-
\Phi_\iota(\V{x})\|^2\Id\V{x}&=0.
\end{align}
%for some $\vt_\alpha\in C^\infty(\Geb,\RR)$ such that $0\le\vt_\alpha\le1$ and $\vt_\alpha(\V{x})=0$, if $\dist(\V{x},\Geb^c)<1/\alpha$, for all $\alpha\in\NN$, and such that $\vt_\alpha\to1$, $\alpha\to\infty$, pointwise on $\Geb$.
If all these conditions are fulfilled, then $\{\Phi_\iota\}_{\iota\in\sI}$ is relatively compact.
\end{prop}

In the following three corollaries we apply the preceding proposition to sequences of ground state
eigenvectors. The first corollary deals with a fixed and bounded region $\Geb$, the second one with 
an unbounded $\Geb$, while the third corollary will be used to approximate unbounded regions by
bounded ones.

\begin{cor}\label{corIRbd}
Assume that $\Geb$ is bounded. Let $\{\mu_\iota\}_{\iota\in\NN}$ be a converging
sequence of non-negative boson masses, $\{\ee_\iota\}_{\iota\in\NN}$ a converging sequence of
coupling constants, and assume that the measurable even functions $\eta_\iota:\RR^3\to[0,1]$,
$\iota\in\NN$, converge pointwise on $\RR^3$. Finally, let $\{\UV_\iota\}_{\iota\in\NN}$ be
a sequence of finite ultraviolet cutoffs that either converges in $[0,\infty)$ or diverges to $\infty$. 
Denote by $\HV_\iota$ the operator obtained upon choosing $\mu=\mu_\iota$, $\ee=\ee_\iota$, 
and $\eta=\eta_\iota$ in the construction of $\HV_{\Geb,\UV_\iota}$ and
assume that $\Phi_\iota$ is a normalized ground state eigenvector of $\HV_\iota$.
Then $\{\Phi_\iota\}_{\iota\in\NN}$ contains a subsequence that converges in $L^2(\Geb,\sF)$.
\end{cor}

\begin{proof}
We extend every $\Phi_\iota$ to $\RR^3$ by setting it equal to $0$ on $\Geb^c$ and denote
this extension again by the same symbol. We shall apply Proposition~\ref{propcpt} to show that the set
$\{\Phi_\iota:\iota\in\NN\}$ is relatively compact in $L^2(\RR^3,\sF)$, which will prove the claim 
because $1_\Geb L^2(\RR^3,\sF)$ is a closed subspace of $L^2(\RR^3,\sF)$. 

Of course, \eqref{cpt2a} is satisfied trvially since $\Geb$ is bounded. 
To verify \eqref{cpt2b} it suffices to show that $\{\Phi_\iota:\iota\in\NN\}$ is bounded
in $W^{1,2}(\RR^3,\sF)$. In view of \eqref{deftGplus} and \eqref{defQGV}
we know, however, that every $\Phi_\iota$ is in the completion of
$\sD(\Geb,\sF)$ with respect to the norm on $W^{1,2}(\RR^3,\sF)$. In particular
$\nabla\Phi_\iota=0$ a.e. on $\Geb^c$ and the bound 
$\sup_{\iota\in\NN}\|\nabla\Phi_\iota\|<\infty$
on the weak gradients of the extended functions follows from Remark~\ref{remcstlokbd}.

To verify \eqref{cpt1a} and \eqref{cpt1b} we first discuss the operators defined in
\eqref{defB1}--\eqref{defB3}, whose $\iota$-dependence will be indicated by a
superscript ${}^{(\iota)}$. We further set
\begin{align*}
\omega_\iota&:=(\V{m}^2+\mu_\iota^2)^\eh,
\\
\beta_{\iota}&:=\ee_\iota\eta_\iota1_{\{|\V{m}|\le\UV_\iota\}}\omega_\iota^\mh
(\omega_\iota+\V{m}^2/2)^{-1},
\\
E_\iota&:=\inf\sigma(H_\iota),
\end{align*}
for all $\iota\in\NN$. We first note the elementary bound
\begin{align*}
\frac{|e^{-i\V{k}\cdot\V{x}}-\chi(\V{k})|}{\langle\V{x}\rangle}&\le
1_{\{\chi=1\}}(\V{k})\frac{|\V{k}||\V{x}|}{\langle\V{x}\rangle}+2\cdot1_{\{\chi<1\}}(\V{k})
\le2\wedge(2|\V{k}|),\quad\V{k},\V{x}\in\RR^3,
\end{align*}
as well as the following consequence of \eqref{rba}, \eqref{pernille0}, and \eqref{ubEDir},
\begin{align}%\nonumber
\|a(e_{\hat{\V{x}}}\V{m}\beta_\iota)(H_\iota-E_\iota+1)^{\mh}\|\label{cpt99}
&\le c\||\V{m}|^\eh\beta_\iota\|(E_\iota+c)^\eh\le c_1,\quad\iota\in\NN.
\end{align}
Here $c>0$ depends only on $\ee^\star:=\sup_{\iota}|\ee_\iota|$ and $V$. 
The constants $c_1,c_2,\ldots>0$ 
appearing here and later on in this proof depend only on $B,\ee^\star,V$,
where $B$ is some open ball contained in $\Geb$.
The norm of the operator in the first line of the right hand side of \eqref{defB3}
is $\le\sup_{t\ge0}\sqrt{t+1}/(t+\omega_\iota(\V{k}))\le1/\sqrt{|\V{k}|(2\wedge|\V{k}|)}$.
We thus find that, uniformly in $\iota\in\sI$,
\begin{align}\label{bdsB}
\|\Xi_1^{(\iota)}(\V{k})\|,\|\Xi_2^{(\iota)}(\V{k})\|&\le2\wedge\frac{2}{|\V{k}|},\qquad
\|\V{\Xi}_3^{(\iota)}(\V{k})\|\le c_2\Big({1\wedge{\frac{1}{|\V{k}|}}}\Big)^\eh,
\end{align}
for all $\V{k}\in\RR^3\setminus\{{0}\}$. Of course, the boundedness of $\Geb$ implies
\begin{align}\label{bdGeb}
\mathfrak{c}_\Geb:=\sup_{\V{x}\in\Geb}\langle\V{x}\rangle<\infty.
\end{align}
Also employing \eqref{eva1anull} we conclude that
\begin{align}\label{carlotta0}
N(r_0,r_1)&\le c_3\mathfrak{c}_\Geb^2\int_{\{r_0\le|\V{k}|\le r_1\}}
\frac{1\wedge|\V{k}|}{(|\V{k}|+\V{k}^2/2)^2}\Id\V{k},\quad 0\le r_0\le r_1\le\infty.
\end{align}
This verifies \eqref{cpt1a}.

The spectral calculus, \eqref{cpt99}, and elementary estimations further reveal that
\begin{align*}
\left.\begin{array}{r}
\phantom{\Big|}\|(\Xi_j^{(\iota)})(\V{k}+\V{h})-(\Xi_j^{(\iota)})(\V{k})\|
\\
\phantom{\Big|}
\|(\V{\Xi}_3^{(\iota)})(\V{k}+\V{h})-(\V{\Xi}_3^{(\iota)})(\V{k})\|\end{array}\right\}
&\le c_4\frac{|\V{h}|}{|\V{k}+\V{h}|},\quad j\in\{1,2\},
\end{align*}
for all $\V{h},\V{k}\in\RR^3$ with $\V{k}\not=0$ and $\V{k}+\V{h}\not=0$. This permits to get
\begin{align*}
&\|(\vr_\delta\V{m}\beta_\iota\Xi_j^{(\iota)})(\V{k}+\V{h})
-(\vr_\delta\V{m}\beta_\iota\Xi_j^{(\iota)})(\V{k})\|
\\
&\le c_5|\V{h}||(\vr_\delta\beta_\iota)(\V{k}+\V{h})|
+c_5\Big({1\wedge{\frac{1}{|\V{k}|}}}\Big)^\eh
|(\vr_\delta\V{m}\beta_\iota)(\V{k}+\V{h})-(\vr_\delta\V{m}\beta_\iota)(\V{k})|,
\end{align*}
for $j\in\{1,2\}$, as well as a completely analogous bound for 
$\vr_\delta\beta_\iota\V{m}\cdot\V{\Xi}^{(\iota)}_3$.
Here we further estimate, assuming $|\V{h}|\le1$ in addition,
\begin{align*}
&\Big({1\wedge{\frac{1}{|\V{k}|}}}\Big)^\eh
|(\vr_\delta\V{m}\beta_\iota)(\V{k}+\V{h})-(\vr_\delta\V{m}\beta_\iota)(\V{k})|
\\
&\le|\V{h}||(\vr_\delta\beta_\iota)(\V{k}+\V{h})|+\omega_\iota(\V{k})^\eh
|(\vr_\delta\beta_\iota)(\V{k}+\V{h})-(\vr_\delta\beta_\iota)(\V{k})|
\\
&\le2|\V{h}|^\eh|(\vr_\delta\beta_\iota)(\V{k}+\V{h})|+
|(\vr_\delta\omega_\iota\beta_\iota)(\V{k}+\V{h})-(\vr_\delta\omega_\iota\beta_\iota)(\V{k})|.
\end{align*}
Combining the latter estimates with \eqref{aki1} and using \eqref{eva1anull} we deduce that
\begin{align}\nonumber
\triangle_\delta(\V{h})&\le c_6\mathfrak{c}_\Geb^2\frac{|\V{h}|}{\delta}
\int_{\RR^3}\frac{1}{(|\V{p}|+\V{p}^2/2)^2}\Id\V{p}
\\\label{carlotta}
&\quad+c_6\mathfrak{c}_\Geb^2\sup_{\iota\in\NN}
\int_{\RR^3}|(\vr_\delta\omega_\iota^\eh\beta_\iota)(\V{k}+\V{h})
-(\vr_\delta\omega_\iota^\eh\beta_\iota)(\V{k})|^2\Id\V{k},
\end{align}
for all $\V{h}\in\RR^3$ satisfying $|\V{h}|\le1$ and all $\delta\in(0,1]$. 
Finally, we observe that our assumptions on $\ee_\iota$, $\mu_\iota$, $\eta_\iota$, and
$\UV_\iota$ together with the dominated convergence theorem imply that, 
in fact for every $\delta\ge0$, the sequence 
$\{\vr_\delta\omega_\iota^\eh\beta_\iota\}_{\iota\in\NN}$
converges in $L^2(\RR^3)$. In particular, its elements form a relatively compact set
in $L^2(\RR^3)$. By Kolmogorov's characterization of relatively compact sets in $L^2(\RR^3)$
the integral in the second line of \eqref{carlotta} goes to zero, as $\V{h}\to0$.
Altogether we now see that \eqref{cpt1b} is satisfied.
\end{proof}

\begin{cor}\label{corIRubd1}
Let $\{(\ee_\iota,\eta_\iota,\UV_\iota,\HV_{\iota})\}_{\iota\in\NN}$ be given as in 
Corollary~\ref{corIRbd} with the only exceptions that $\Geb$ is now assumed to be {\em unbounded}
and the boson mass $\mu\ge0$ is kept fixed. Set 
$$
\ee_\infty:=\lim_{\iota\to\infty},\quad\eta_\infty:=\lim_{\iota\to\infty}\eta_\iota,
\quad\UV_\infty:=\lim_{\iota\to\infty}\UV_\iota\in[0,\infty].
$$
Let $\HV_\infty$ denote the operator $\HV_{\Geb,\UV_\infty}$ defined by means of $\mu$, 
$\ee_\infty$, and $\eta_\infty$. Finally, let $\Sigma_\iota$ and $E_\iota$ denote the localization 
threshold and minimal energy of $\HV_\iota$, for all $\iota\in\NN\cup\{\infty\}$. Assume that
\begin{align*}
\Sigma_{\infty}>E_{\infty}.
\end{align*}
Then the following holds:
\begin{enumerate}
\item[{\rm(1)}] There exists $J\in\NN$ such that $\Sigma_{\iota}>E_{\iota}$, for all $\iota\ge J$. 
\item[{\rm(2)}] 
Assume that $\Phi_\iota$ is a normalized ground state eigenvector of 
$\HV_{\iota}$, for every $\iota\in\NN$ with $\iota\ge J$.
Then $\{\Phi_\iota\}_{\iota=J}^\infty$ contains a subsequence that converges in $L^2(\Geb,\sF)$.
\end{enumerate}
\end{cor}

\begin{proof}
Pick some $\tilde{\alpha},b>0$ such that 
$$
\Sigma_{\infty}-E_{\infty}\ge\tilde{\alpha}^2+4b.
$$ 
For every $\iota\in\NN\cup\{\infty\}$, let $\HV_{R,\iota}$ denote the operator 
$\HV_{\Geb_R,\UV_\iota}$ defined by means of $\mu$, $\ee_\iota$, $\eta_\iota$, and set 
$E_{R,\iota}:=\inf\sigma(\HV_{R,\iota})$; recall that $\Geb_R=\Geb\cap\{|\V{x}|>R\}$. 
Here we choose $R\ge1$ so large that 
\begin{align*}
E_{R,\infty}-E_{\infty}-\frac{8}{R^2}\ge2\tilde{\alpha}^2+3b.
\end{align*}
Lemma~\ref{lempert} implies that $\HV_{R,\iota}\to\HV_{R,\infty}$, $\iota\to\infty$, 
in norm resolvent sense. Since norm resolvent convergence entails convergence of the spectrum
\cite{ReedSimonI}, we see that $E_{R,\iota}\to E_{R,\infty}$. 
Likewise, $E_{\iota}\to E_{\infty}$, $\iota\to\infty$, 
by norm resolvent convergence. Therefore, we find some $J\in\NN$ such that, for all natural
numbers $\iota\ge J$,
\begin{align}\label{lalilu}
E_{R,\iota}-E_{\iota}&\le E_{R,\infty}- E_{\infty}+1\quad\text{and}\quad
E_{R,\iota}- E_{\iota}-\frac{8}{R^2}\ge\tilde{\alpha}^2+2b.
\end{align}
Since $\Sigma_{\iota}\ge E_{R,\iota}$, this implies Assertion~(1). 

To prove (2) we just have to substitute all arguments that exploited the boundedness of $\Geb$
in the proof of Proposition~\ref{corIRbd} by the following considerations: Notice first that the right hand sides 
of the inequalities in the proof of Proposition~\ref{corIRbd} depend on $\Geb$ only via the open ball 
$B$ and the quantity defined in \eqref{bdGeb}; furthermore, the constants in \eqref{eva1anull} 
contribute to the right hand sides of \eqref{carlotta0} and \eqref{carlotta}.
We shall now apply Lemma~\ref{lemeva} for each fixed $\iota\ge J$ and with $\Geb'=\Geb$, always 
using the parameter $R$ chosen above. Then the quantities $\lambda$ and $\Delta$ appearing in the
statement of Lemma~\ref{lemeva} become $\iota$-dependent,
$$
\lambda_\iota:=E_\iota+\frac{1}{2}(E_{R,\iota}-E_\iota),
\quad\Delta_\iota:=\frac{1}{2}(E_{R,\iota}-E_\iota)-\frac{4}{R^2},\quad\iota\in\NN.
$$
Thanks to \eqref{lalilu} we have, however, the uniform lower and upper bounds 
$$
\frac{\tilde{\alpha}^2}{2}+b\le\Delta_\iota\le\frac{1}{2}(E_{R,\infty}- E_{\infty}+1),\quad\iota\in\NN.
$$
Of course, $E_\iota\le\lambda_\iota$, whence the conditions (a) and (b) in
Lemma~\ref{lemeva} are trivially satisfied and (c) holds by assumption in the present situation.
Finally, we choose $\ve\in(0,1)$ such that $(1-\ve)\sqrt{\tilde{\alpha}^2+2b}=\tilde{\alpha}$.
In view of the preceding remarks and the first bound in \eqref{lalilu} we find 
$\iota$-independent upper bounds on the quantities called $\alpha$ and $M$ in Lemma~\ref{lemeva},
and $\tilde{\alpha}$ is a lower bound for $\alpha$.
Therefore, \eqref{eva1}--\eqref{eva1beta} yield the uniform bounds
\begin{align*}
\sup_{\iota\ge J}\|e^{\tilde{\alpha}\langle\hat{\V{x}}\rangle}\Phi_{\iota}\|&<\infty,\quad
\sup_{\iota\ge J}\|e^{\tilde{\alpha}\langle\hat{\V{x}}\rangle}\nabla\Phi_{\iota}\|<\infty,\quad
\sup_{\iota\ge J}\|e^{\tilde{\alpha}\langle\hat{\V{x}}\rangle}
a(e_{\hat{\V{x}}}\V{m}\beta_{\UV_\iota})\Phi_{\iota}\|<\infty.
\end{align*}
The first one clearly implies \eqref{cpt2a}. Together they entail uniform 
(in $\iota\ge J$) bounds on the expressions in \eqref{defUpsilon1}--\eqref{defUpsilon3}, 
which can be used as substitutes for 
the bounds \eqref{eva1anull} and \eqref{bdGeb} employed in the proof of Proposition~\ref{corIRbd}.
\end{proof}

To prove the third corollary of Proposition~\ref{propcpt} we need the following lemma:

\begin{lem}\label{lemconvEGebn0}
Assume that $\Geb$ is unbounded, and pick open sets $\Geb_\iota\subset\Geb$, $\iota\in\NN$, 
such that $\emptyset\not=\Geb_1\subset\Geb_2\subset\Geb_3\subset\ldots\:$, such that
$\bigcup_{\iota\in\NN}\Geb_\iota=\Geb$, and such that every compact subset of $\Geb$ 
is contained in some $\Geb_\iota$. Keep $\mu\ge0$, $\ee\in\RR$, $\UV\in[0,\infty]$, 
and $\eta$ fixed. Then
\begin{align}\label{convEGebn}
\lim_{\iota\to\infty}E_{\Geb_\iota,\UV}&=E_{\Geb,\UV}.
\end{align}
\end{lem}

\begin{proof}
By the variational principle and the fact that $\sD(\Geb_\iota,\fdom(\Id\Gamma(\omega)))$ is a
form core for $\HV_{\Geb_\iota,\UV}$ it is clear that $E_{\Geb_\iota,\UV}\ge E_{\Geb,\UV}$, 
for all $\iota\in\NN$. Now let $\ve>0$. Since $\sD(\Geb,\fdom(\Id\Gamma(\omega)))$ 
is a form core for $\HV_{\Geb,\UV}$, we find some normalized 
$\Psi\in\sD(\Geb,\fdom(\Id\Gamma(\omega)))$ such that 
$\hv_{\Geb,\UV}[\Psi]<E_{\Geb,\UV}+\ve$. There exists $\iota_0\in\NN$ such that 
$\supp(\Psi)\subset\Geb_\iota$, for all $\iota\ge\iota_0$, and we
conclude that $\Psi\restr_{\Geb_\iota}\in\dom(\hv_{\Geb_\iota,\UV})$ and
\begin{align*}
E_{\Geb_\iota,\UV}&=E_{\Geb_\iota,\UV}\|\Psi\restr_{\Geb_\iota}\|^2\le
\hv_{\Geb_\iota,\UV}[\Psi\restr_{\Geb_\iota}]
=\hv_{\Geb,\UV}[\Psi]<E_{\Geb,\UV}+\ve,\quad\iota\ge\iota_0,
\end{align*}
which proves \eqref{convEGebn}.
\end{proof}

\begin{cor}\label{corIRubd2}
In the situation of Lemma~\ref{lemconvEGebn0} we suppose in addition that
$\UV$ is finite and assume that the binding condition holds for $\Geb$, i.e.,
$\Sigma_{\Geb,\UV}>E_{\Geb,\UV}$. Furthermore, we assume that $\Phi_\iota$ is a 
normalized ground state eigenvector of $\HV_{\Geb_\iota,\UV}$, for every $\iota\in\NN$.
If every $\Phi_\iota$ is extended by $0$ to $\Geb$, then $\{\Phi_\iota\}_{\iota=J}^\infty$ 
contains a subsequence that converges in $L^2(\Geb,\sF)$.
\end{cor}

\begin{proof}
We choose $R>0$ and define $\lambda$ and $\Delta$ precisely as in the statement of
Lemma~\ref{lemeva}. Furthermore, we choose an arbitrary open ball $B\subset\Geb_1$ and pick some
$\ve\in(0,1)$. By our assumptions, 
every $\Geb_\iota$ satisfies the conditions (b) and (c) in Lemma~\ref{lemeva}. To verify Condition~(a) we 
employ Lemma~\ref{lemconvEGebn0} which implies that $E_{\Geb_\iota,\UV}\le\lambda$ for all 
$\iota\ge J$ and some $J\in\NN$. Therefore, the bounds \eqref{eva1}--\eqref{eva1beta} are
available with $\Geb'=\Geb_\iota$, for every $\iota\ge J$. They yield uniform
(in $\iota\ge J$) bounds on the expressions in \eqref{defUpsilon1}--\eqref{defUpsilon3}, which can be 
used as substitutes for \eqref{eva1anull} and \eqref{bdGeb} in the proof of Proposition~\ref{corIRbd}.
\end{proof}

%%%%%%%%%%%%%%%%%%%%%%%%%%%%%%%%%%%%%%%%%%%%%%%
%%%%%%%%%%%%%%%%%%%%%%%%%%%%%%%%%%%%%%%%%%%%%%%

\subsection{Construction of ground states in the general case}\label{ssecGSgen}

\noindent
In a chain of approximation steps we next drop the various restrictive hypotheses 
employed in Theorem~\ref{thmGSmass}. To this end we shall repeatedly combine the 
compactness results of the previous subsection with the following abstract lemma,
which is identical to \cite[Lem.~5.1]{KMS2011}. The lemma is a slightly improved version 
of a statement we learned from \cite{BFS1998b}.

\begin{lem}\label{lemSRC}
Let $A,A_1,A_2,\ldots$ be self-adjoint operators in some separable Hilbert space $\sK$ such that
$A_j\to A$ in the strong resolvent sense, as $j\to\infty$. For every $j\in\NN$, let $a_j\in\RR$
be an eigenvalue of $A_j$ and $\phi_j\in\dom(A_j)\setminus\{0\}$ a corresponding
eigenvector. Assume that $\{\phi_j\}_{j\in\NN}$ converges weakly to some non-zero $\phi\in\sK$.
Then $a:=\lim_{j\to\infty}a_j$ exists, $\phi\in\dom(A)$, and $A\phi=a\phi$.
If $a_j=\inf\sigma(A_j)$, for every $j\in\NN$, then $a=\inf\sigma(A)$.
\end{lem}

For technical reasons we first have to trade the positive mass required in Theorem~\ref{thmGSmass}
for a sharp infrared cutoff.

\begin{prop}[Ground states with IR cutoff, bounded $\Geb$, and finite~$\UV$]\label{propGSIRcutoff}
Let $\UV\in(0,\infty)$ and assume that $\Geb$ is bounded and that $\eta=0$ on 
$\{|\V{m}|\le\sigma\}$, for some $\sigma\in(0,\UV)$. Then $E_{\Geb,\UV}$ is an eigenvalue 
of both $\HV_{\Geb,\UV}$ and $\NV_{\Geb,\UV}$.
\end{prop}

\begin{proof}
It only remains to treat the case $\mu=0$ and it suffices to consider $\HV_{\Geb,\UV}$, because
$\HV_{\Geb,\UV}$ and $\NV_{\Geb,\UV}$ are unitarily equivalent via the Gross transformation 
$G_{\sigma,\UV}$.

Let $\mu=0$, so that $\omega=|\V{m}|$, 
and let $\{\mu_n\}_{n\in\NN}$ be a monotone zero-sequence of strictly positive 
real numbers. For every $n\in\NN$, put $\omega_n(\V{k}):=(\V{k}^2+\mu_n^2)^\eh$,
$\V{k}\in\RR^3$, and let $\HV_n$ be the operator obtained by doing the following 
replacements in the construction of $\HV_{\Geb,\UV}$,
\begin{align}\label{repl}
\mu\mapsto\mu_n,\quad
\eta\mapsto\eta_n:=\frac{\sigma^\eh}{(\sigma^2+\mu_n^2)^{\nf{1}{4}}}
\frac{\omega_n^\eh}{|\V{m}|^\eh}\eta,
\quad\ee\mapsto\ee_n:=\ee\frac{(\sigma^2+\mu_n^2)^{\nf{1}{4}}}{\sigma^\eh}.
\end{align}
Let $f_{\UV,n}$ be the coupling function obtained after all these replacements.
Then $f_{\UV,n}$ is actually {\em $n$-independent} and always equal to
$\ee|\V{m}|^\mh\eta1_{\{|\V{m}|\le\UV\}}$. Defining ${\QGV}$ by means of
$\omega=|\V{m}|$ and setting
$\sQ_{\Geb,n}:=\dom(\mathfrak{t}_\Geb^+)\cap L^2(\Geb,\fdom(\Id\Gamma(\omega_n)))$, 
we have $\sQ_{\Geb,n}=\sQ_{\Geb,1}\subset{\QGV}$, for all $n\in\NN$, and the
quadratic form of $\HV_n$, call it $\hv_n$, is simply given by
\begin{align*}
\hv_n[\Psi]=\hv_{\Geb,\UV}[\Psi]
+\int_\Geb\|\Id\Gamma(\omega_n)^\eh\Psi(\V{x})\|^2\Id\V{x}-
\int_\Geb\|\Id\Gamma(\omega)^\eh\Psi(\V{x})\|^2,\quad\Psi\in\sQ_{\Geb,1}.
\end{align*}
In fact, the replacement manoeuvre \eqref{repl} is only necessary to argue that $\hv_n$ 
can be dealt with by our previous results; notice that $\eta_n\le1$. 
In particular, we know from Theorem~\ref{thmGSmass} that
$\inf\sigma(\HV_n)$ is an eigenvalue of $\HV_n$; let $\Phi_n$ be a corresponding normalized
eigenvector. By Corollary~\ref{corIRbd}, $\{\Phi_{n}\}_{n\in\NN}$ contains a convergent subsequence, 
call it $\{\Phi_{n_j}\}_{j\in\NN}$, whose limit, call it $\Phi_\infty$, is normalized, too, of course.
Furthermore, the monotone convergence of quadratic forms,
$\hv_n[\Psi]\downarrow\hv_{\Geb,\UV}[\Psi]$,
$\Psi\in\sQ_{\Geb,1}$, and the fact that $\sQ_{\Geb,1}$ is a core for $\hv_{\Geb,\UV}$
imply that $\HV_n\to\HV_{\Geb,\UV}$, $n\to\infty$, in the strong resolvent sense;
see \cite[Thm.~S.15 and Thm.~S.16]{ReedSimonI}.
Applying Lemma~\ref{lemSRC} to the subsequence $\{\Phi_{n_j}\}_{j\in\NN}$,
we see that $\Phi_\infty\in\dom(\HV_{\Geb,\UV})$ and
$\HV_{\Geb,\UV}\Phi_\infty=E_{\Geb,\UV}\Phi_\infty$.
\end{proof}

\begin{prop}[Ground states for bounded $\Geb$ and finite $\UV$]\label{propGSbd}
Let $\UV\in(0,\infty)$ and suppose that $\Geb$ is bounded.
Then $E_{\Geb,\UV}$ is an eigenvalue of $\HV_{\Geb,\UV}$.
\end{prop}

\begin{proof}
Let $\HV_n$ denote the operator obtained upon putting $1_{\{|\V{m}|>1/n\}}\eta$ in place
of $\eta$ in the definition of $\HV_{\Geb,\UV}$. Then Lemma~\ref{lempert} implies that 
$\HV_n\to\HV_{\Geb,\UV}$ in norm resolvent sense, as $n\to\infty$. 
Invoking Proposition~\ref{propGSIRcutoff} we further find a
normalized ground state eigenvector $\Phi_n$ of $\HV_n$, for every $n\in\NN$. 
By Corollary~\ref{corIRbd}, $\{\Phi_n\}_{n\in\NN}$ contains a converging subsequence and we conclude
by applying Lemma~\ref{lemSRC} to that subsequence.
\end{proof}

In the next step we approximate an unbounded $\Geb$ by bounded open subsets.
To this end let us recall that we defined the operators $\wt{T}_{\Geb,\UV,t}$ on
$L^2(\RR^3,\sF)$ with the convention that a given $\Psi\in L^2(\Geb',\sF)$, defined
on an open subset $\Geb'\subset\RR^3$ not necessarily equal to $\Geb$, 
is first extended to $\RR^3$ by $0$ before we apply $\wt{T}_{\Geb,\UV,t}$ to it.
For later reference we further note that
\begin{align}\label{L2LpGeb}
\sup_{\UV\in[0,\infty]}\sup_{\cG\subset\RR^3\,\mathrm{open}}\sup\big\{
\|\wt{T}_{\cG,\UV,t}\Psi\|_p\,\big|\:\Psi\in L^2(\RR^3,\sF),\,\|\Psi\|_2\le1\big\}<\infty,
\end{align}
for all $t>0$ and $p\in[2,\infty]$, as an immediate consequence of Proposition~\ref{propGSmu}(2)
applied to $\HV_{\Geb,\UV}$. We also need a final technical lemma before we can continue
our construction of ground states:

\begin{lem}\label{lemconvEGebn}
Assume that $\Geb$ is unbounded, and pick open sets $\Geb_n\subset\Geb$, $n\in\NN$, such that
$\emptyset\not=\Geb_1\subset\Geb_2\subset\Geb_3\subset\ldots\:$, such that
$\bigcup_{n\in\NN}\Geb_n=\Geb$, and such that every compact subset of $\Geb$ 
is contained in some $\Geb_n$. Let $t>0$, $p\in[2,\infty)$, and $\Psi\in L^2(\RR^3,\sF)$. Then
\begin{align}\label{convTGebn}
\lim_{n\to\infty}\sup_{\UV\in[0,\infty]}
\|\wt{T}_{\cG_n,\UV,t}\Psi- \wt{T}_{\Geb,\UV,t}\Psi\|_{p}&=0,
\end{align}
where the norm is the one on $L^p(\RR^3,\sF)$.
\end{lem}

\begin{proof}
Let $\V{x}\in\RR^3$ and $\tau_{\Geb_n}(\V{x})$ be the first entry time of $\V{B}^{\V{x}}$ into 
$\Geb_n^c$. Let $\gamma\in\Omega$. Then the image of the path 
$\{\V{B}^{\V{x}}_s(\gamma):s\in[0,t]\}$ up to time $t$ is compact. Hence, it is contained in $\Geb$, 
if and only if it is contained in every $\Geb_n$ with $n\ge n_\gamma$, for some $n_\gamma\in\NN$.
This implies that $1_{\{\tau_{\Geb_n}(\V{x})>t\}}\to1_{\{\tau_{\Geb}(\V{x})>t\}}$ on $\Omega$,
as $n\to\infty$. Thus, $\EE[|1_{\{\tau_{\Geb_n}(\V{x})>t\}}-1_{\{\tau_{\Geb}(\V{x})>t\}}|^6]\to0$
by dominated convergence. Let $\Psi\in L^2(\RR^3,\sF)$. 
Employing H\"{o}lder's inequality, \eqref{kashmir},
and \eqref{WinLp} similarly as in the proof of Proposition~\ref{propGSmu}(2) we then find
\begin{align*}
\|\wt{T}_{\cG_n,\UV,t}\Psi- \wt{T}_{\Geb,\UV,t}\Psi\|_{p}
\le\sup_{\V{y}\in\RR^3}\EE\Big[e^{-6\int_0^tV(\V{B}_s^{\V{y}})\Id s}\Big]^{\nf{1}{6}}
\sup_{\V{z}\in\RR^3}\EE[\|\wt{W}_{\UV,t}(\V{z})\|^6]^{\nf{1}{6}}&
\\
\cdot\bigg(\int_{\RR^3}
\EE\big[|1_{\{\tau_{\Geb_n}(\V{x})>t\}}-1_{\{\tau_{\Geb}(\V{x})>t\}}|^6\big]^{\nf{p}{6}}
(e^{t\Delta/2}\|\Psi(\cdot)\|^2)(\V{x})^{\nf{p}{2}}\Id\V{x}\bigg)^{\nf{1}{p}}&.
\end{align*}
We recall that $e^{t\Delta/2}$ with $t>0$ maps $L^1(\RR^3)$ 
continuously into $L^{\nf{p}{2}}(\RR^3)$. Therefore, the right hand side of the previous estimation 
goes to zero as $n\to\infty$ by dominated convergence. The convergence is uniform
in $\UV\in[0,\infty]$ on account of \eqref{WinLp} where the constants are $\UV$-independent.
\end{proof}

\begin{prop}[Ground states for finite $\UV$]\label{propGSHfin}
Let $\UV\in(0,\infty)$ and assume that the binding condition \eqref{bindcond} is fulfilled.
Then $E_{\Geb,\UV}$ is an eigenvalue of $\HV_{\Geb,\UV}$.
\end{prop}

\begin{proof}
With Proposition~\ref{propGSbd} in mind we assume without loss of generality that $\Geb$ is unbounded.
Let $\Geb_n$, $n\in\NN$, be bounded and have all properties postulated in the statement of 
Lemma~\ref{lemconvEGebn}. By virtue of Proposition~\ref{propGSbd} we can, for every
$n\in\NN$, find a normalized ground state eigenvector of $\HV_{\Geb_n,\UV}$; 
we denote by $\Phi_n$ its extension by $0$ to $\Geb$. Thanks to the binding condition
and Corollary~\ref{corIRubd2} we know that $\{\Phi_n\}_{n\in\NN}$ contains a subsequence
converging in $L^2(\Geb,\sF)$, say $\{\Phi_{n_j}\}_{j\in\NN}$.
The relations \eqref{convEGebn}, \eqref{L2LpGeb}, \eqref{convTGebn}, 
and the Feynman-Kac formulas \eqref{FKDfor} for $\Geb_{n_j}$ and $\Geb$ now imply that
\begin{align*}
\Phi:=\lim_{j\to\infty}\Phi_{n_j}=\lim_{j\to\infty}e^{tE_{\Geb_{n_j},\UV}}
\wt{T}_{\Geb_{n_j},\UV,t}\Phi_{n_j}=e^{-t(\HV_{\Geb,\UV}-E_{\Geb,\UV})}\Phi,
\end{align*}
for every $t>0$. The claim now follows from the spectral calculus.
\end{proof}

Finally, we remove the ultraviolet cutoff in our existence results:

\begin{thm}[Ground states for $\HV_{\Geb,\UV}$; general case]\label{thmGSH}
Let $\UV\in[0,\infty]$ and assume that the binding condition \eqref{bindcond} is  fulfilled. 
Then $E_{\Geb,\UV}$ is an eigenvalue of $\HV_{\Geb,\UV}$. 
\end{thm}

\begin{proof}
In view of Proposition~\ref{propGSHfin} it suffices to consider $\UV=\infty$.
Pick any increasing sequence $\Lambda_n\uparrow\infty$, $n\to\infty$.
Then $\HV_{\Geb,\Lambda_n}\to\HV_{\Geb,\infty}$, $n\to\infty$, in the norm
resolvent sense and Corollary~\ref{corIRubd1} ensures that the binding conditions
$\Sigma_{\Geb,\UV_n}>E_{\Geb,\UV_n}$ hold, for all $n\ge n_0$ and some $n_0\in\NN$.
By Proposition~\ref{propGSHfin} every $\HV_{\Geb,\UV_n}$ with $n\ge n_0$ has a normalized ground state
eigenvector, say $\Phi_n$. Again from Corollary~\ref{corIRubd1} we infer that $\{\Phi_n\}_{n\ge n_0}$ 
contains a converging subsequence. We conclude with the help of Lemma~\ref{lemSRC}.
\end{proof}

\begin{thm}[Ground states for $\NV_{\Geb,\UV}$; general case]\label{thmGSN}
Let $\UV\in[0,\infty]$ and assume that the binding condition \eqref{bindcond} and the
infrared condition $\omega^{-\nf{3}{2}}\eta\in L^1_\loc(\RR^3)$ are fulfilled. 
Then $E_{\Geb,\UV}$ is an eigenvalue of $\NV_{\Geb,\UV}$.
\end{thm}

\begin{proof}
Under the given infrared condition $\NV_{\Geb,\UV}$ is unitarily equivalent to
$\HV_{\Geb,\UV}$ via the Gross transformation $G_{0,\UV}$. Therefore, the claim
follows from Theorem~\ref{thmGSH}.
\end{proof}

%%%%%%%%%%%%%%%%%%%%%%%%%%%%%%%%%%%%%%%%%%%%%%%
%%%%%%%%%%%%%%%%%%%%%%%%%%%%%%%%%%%%%%%%%%%%%%%

\subsection{Continuity Properties of Ground States}\label{sseccontGS}

\noindent
As mentioned in the introduction, the compactness argument employed repeatedly in the
previous subsection can also be used to study the $L^2$-continuity of ground state eigenvectors of 
the renormalized Hamiltonians with respect to parameters like $\ee$. Before we do this we have, 
however, to extend the crucial formulas of Proposition~\ref{prop-IR} to the case $\UV=\infty$. In particular,
we have to give a meaning to the annihilation operators corresponding to the components of
$e_{\V{x}}\V{m}\beta_\infty$, which are not in $L^2(\RR^3)$.

In fact, the definition of the ``pointwise" annihilation operator in the beginning of
Subsection~\ref{ssecIR} suggests the following extended definition of the ``smeared"
annihilation operator: Let $f\in L^2(\RR^3,\omega^{-1}\Id\V{k})$ and 
$\Psi\in\dom(a_\omega(f)):=\fdom(\Id\Gamma(\omega))$. Since 
$a\Psi\in L^2(\RR^3,L^2(\Geb,\sF);\omega^\eh\Id\V{k})$,
the following $L^2(\Geb,\sF)$-valued Bochner-Lebesgue integral exists,
\begin{align*}
a_\omega(f)\Psi:=\int_{\RR^3}\bar{f}(\V{k})(a\Psi)(\V{k})\Id\V{k}.
\end{align*} 
As soon as $f\in\dom(\omega^\mh)\subset L^2(\RR^3)$, we then have 
$a_\omega(f)=a(f)\!\!\upharpoonright_{\fdom(\Id\Gamma(\omega))}$. (This is always the case
if $\inf\omega>0$, of course.) In view of \eqref{fordGak} and the Cauchy-Schwarz inequality the
familiar relative bound for the annihilation operator is still satisfied,
\begin{align}\label{rbaomega}
\|a_\omega(f)\Psi\|&\le\|\omega^\mh f\|_{L^2(\RR^3)}\|\Id\Gamma(\omega)^\eh\Psi\|,
\quad\Psi\in\fdom(\Id\Gamma(\omega)).
\end{align}
This gives in particular a meaning to the components of
\begin{align*}
a_\omega(e_{\V{x}}\V{m}\beta_{\infty}):=
\big(a_\omega(e_{\V{x}}{m_1}\beta_{\infty}),a_\omega(e_{\V{x}}{m_2}\beta_{\infty}),
a_\omega(e_{\V{x}}{m_3}\beta_{\infty})\big),
\end{align*}
which are well-defined operators whose domains equal $\fdom(\Id\Gamma(\omega))$.
The new notation is very convenient to state the following extension of Proposition~\ref{prop-IR}.

\begin{prop}\label{propIRinfty}
Assume that the binding condition $\Sigma_{\Geb,\infty}>E_{\Geb,\infty}$ is fulfilled
and let $\Phi_{\Geb,\infty}$ be a ground state eigenvector of $\HV_{\Geb,\infty}$.
Then, for a.e. $\V{k}\in\RR^3\setminus\{0\}$, the formula \eqref{aki1} is still valid for $\UV=\infty$,
provided that $a_\omega(e_{\hat{\V{x}}}\V{m}\beta_{\infty})$ is put in place of
$a(e_{\hat{\V{x}}}\V{m}\beta_{\UV})$ in the definitions \eqref{defUpsilon2} and \eqref{defB3}.
\end{prop}

\begin{proof}
Let $\V{k}\in\RR^3\setminus\{0\}$. Since $\HV_{\Geb,\UV}\to\HV_{\Geb,\infty}$, 
$\UV\to\infty$, in norm resolvent sense, we know that $E_{\Geb,\UV}\to E_{\Geb,\infty}$, 
thus $R_{\Geb,\UV}(\V{k})\to R_{\Geb,\infty}(\V{k})$ in operator norm. Hence, also
\begin{align}\label{vernon0}
(\HV_{\Geb,\UV}-E_{\Geb,\UV})R_{\Geb,\UV}(\V{k})
\xrightarrow{\;\;\UV\to\infty\;\;}(\HV_{\Geb,\infty}-E_{\Geb,\infty})R_{\Geb,\infty}(\V{k}),
\end{align}
in operator norm, since 
\begin{align}\label{vernon0b}
(\HV_{\Geb,\UV'}-E_{\Geb,\UV'})R_{\Geb,\UV'}(\V{k})=\id-\omega(\V{k})R_{\Geb,\UV'}(\V{k}),
\quad\UV'\in[0,\infty].
\end{align}
Furthermore, the second limit relation stated in Theorem~\ref{thmrbUV}(4) together with
\eqref{rbaomega} shows that
\begin{align}\label{vernon1}
a_{\omega}(e_{\hat{\V{x}}}\V{m}\beta_{\UV})R_{\Geb,\UV}(\V{k})
\xrightarrow{\;\;\UV\to\infty\;\;}a_{\omega}(e_{\hat{\V{x}}}\V{m}\beta_{\infty})R_{\Geb,\infty}(\V{k})
\quad\text{in $\LO(L^2(\Geb,\sF))^3$.}
\end{align}

Next, let $n_0\in\NN$ and $\UV_n$, $\Phi_n$, $n\ge n_0$, be the same objects as in the proof 
of  Theorem~\ref{thmGSH}. Let $\{\Phi_{n_j}\}_{j\in\NN}$ be a converging subsequence of
$\{\Phi_n\}_{n\ge n_0}$ and put $\Phi:=\lim_{j\to\infty}\Phi_{n_j}$. Then
\begin{align}\nonumber
&\|a(e_{\hat{\V{x}}}\V{m}\beta_{\UV_{n_j}})\Phi_{n_j}
-a_\omega(e_{\hat{\V{x}}}\V{m}\beta_{\infty})\Phi\|
\\\label{vernon2}
&\le\||\V{m}|^\eh(\beta_{\UV_{n_j}}-\beta_{\infty})\|\|\Id\Gamma(\omega)^\eh\Phi\|
+\||\V{m}|^\eh\beta_{\infty}\|\|\Id\Gamma(\omega)^\eh(\Phi_{n_j}-\Phi)\|,
\end{align}
for every $j\in\NN$. On account of \eqref{pernille0} we further find 
\begin{align}\nonumber
&\|\nabla(\Phi_{n_j}-\Phi)\|^2+\|\Id\Gamma(\omega)^\eh(\Phi_{n_j}-\Phi)\|^2
\\\label{vernon3}
&\le c\hv_{\Geb,\UV_{n_j}}[\Phi_{n_j}-\Phi]+c'\|\Phi_{n_j}-\Phi\|^2,\quad j\in\NN,
\end{align}
with constants $c,c'>0$ depending only on $\ee$ and $V$. Here the term
\begin{align*}
\hv_{\Geb,\UV_{n_j}}[\Phi_{n_j}-\Phi]
&=E_{\Geb,\UV_{n_j}}+E_{\Geb,\infty}-2E_{\Geb,\UV_{n_j}}\Re\SPn{\Phi_{n_j}}{\Phi}
+\hv_{\Geb,\UV_{n_j}}[\Phi]-\hv_{\Geb,\infty}[\Phi]
\end{align*}
vanishes in the limit $j\to\infty$ in view of \eqref{defhKinftyV}.

Let us extend the notation introduced in \eqref{defUpsilon1} through \eqref{defB3} to
$\UV=\infty$ replacing $a$ by $a_\omega$ in \eqref{defUpsilon2} and \eqref{defB3}. 
Furthermore, let us indicate the $\UV$-dependence of these vectors and operators by an additional 
subscript $\UV$. To apply \eqref{vernon1} we shall employ the following
equivalent representation of the operator triple in \eqref{defB3},
\begin{align}\label{B3alt}
\V{\Xi}_{\UV,3}(\V{k})&=\{a_\omega(e_{\hat{\V{x}}}\V{m}\beta_{\UV})R_{\Geb,\UV}(\V{k})\}^*
\frac{e^{-i\V{k}\cdot\hat{\V{x}}}-\chi(\V{k})}{\langle\hat{\V{x}}\rangle},\quad \UV\in[0,\infty].
\end{align}
Since the terms $\langle\hat{\V{x}}\rangle$ and $1/\langle\hat{\V{x}}\rangle$ in
\eqref{defUpsilon2} and \eqref{defB2} (resp. \eqref{defUpsilon3} and \eqref{B3alt}) cancel each other,
the above relations \eqref{vernon1}--\eqref{vernon3} then imply that
\begin{align*}
\Xi_{\UV_{n_j},2}(\V{k})\V{k}\cdot\V{\Upsilon}_{\UV_{n_j},2}
&\xrightarrow{\;\;j\to\infty\;\;}\Xi_{\infty,2}(\V{k})\V{k}\cdot\V{\Upsilon}_{\infty,2},
\\
\V{\Xi}_{\UV_{n_j},3}(\V{k})\cdot\V{k}\Upsilon_{\UV_{n_j},3}
&\xrightarrow{\;\;j\to\infty\;\;}\V{\Xi}_{\infty,3}(\V{k})\cdot\V{k}\Upsilon_{\infty,3}.
\end{align*}
Let $N\in\NN$ and denote by $\vt_N$ the maximal operator of multiplication with the characteristic
function of $\{|\V{x}|\le N\}$ on $L^2(\Geb,\sF)$. By virtue of \eqref{vernon0b} and 
\eqref{resexploc2} we then know that the components of
$\vt_N(\HV_{\Geb,\UV}-E_{\Geb,\UV})R_{\Geb,\UV}(\V{k})\hat{\V{x}}$ 
extend to bounded operators on $L^2(\Geb,\sF)$, whose norms are bounded by a constant
depending solely on $\V{k}\not=0$ and $N$.
(Here we apply \eqref{resexploc2} with $V-E_{\Geb,\UV}+\omega(\V{k})$ substituted for $V$
and $z=0$.) Since the components of $\hat{\V{x}}\Phi$ are in $L^2(\Geb,\sF)$,
this in conjunction with \eqref{vernon0} implies 
\begin{align*}
&\vt_N\Xi_{\UV_{n_j},1}(\V{k})\V{k}\cdot\V{\Upsilon}_{\UV_{n_j},1}
\xrightarrow{\;\;j\to\infty\;\;}\vt_N\Xi_{\infty,1}(\V{k})\V{k}\cdot\V{\Upsilon}_{\infty,1}.
\end{align*}

Finally, we observe that the convergence 
$\Id\Gamma(\omega)^\eh(\Phi_{\UV_{n_j}}-\Phi)\to0$, $j\to\infty$, in
$L^2(\Geb,\sF)$, the relation \eqref{fordGak}, and the Riesz-Fischer theorem imply that 
$$
(a\Phi_{\UV_{m_i}})(\V{k})\xrightarrow{\;\;i\to\infty\;\;}(a\Phi)(\V{k})
\quad\text{in $L^2(\Geb,\sF)$, for a.e. $\V{k}\in\RR^3$,}
$$
where $\{m_i\}_{i\in\NN}$ is some subsequence of $\{n_j\}_{j\in\NN}$.

Putting all the above remarks together and applying \eqref{prop-IR} 
to every cutoff $\UV_{m_i}$ with $i\in\NN$, we arrive at
\begin{align*}
\vt_N(a\Phi)(\V{k})&=-\sum_{\ell=1}^2\vt_N\Xi_{\infty,\ell}(\V{k})
{\beta_{\infty}(\V{k})}\V{k}\cdot\V{\Upsilon}_{\UV,\ell}
-\vt_N\V{\Xi}_{\infty,3}(\V{k})\cdot\V{k}{\beta_{\infty}(\V{k})}\Upsilon_{\infty,3},
\end{align*}
for all $N\in\NN$ and every $\V{k}\not=0$ in the complement of some $N$-independent zero set.
\end{proof}

We are now in a position to study the norm continuity of ground state eigenvectors with respect to the 
coupling constant:

\begin{thm}\label{thmL2cont}
Assume that $\Geb$ is connected. Let $\{\ee_n\}_{n\in\NN}$ be a sequence in $\RR$ converging
to some $\ee\in\RR$. Keep $\mu$ and $\eta$ fixed and
assume that, for the operator $\HV:=\HV_{\Geb,\infty}$ defined by means of 
$\ee$, the binding condition $\Sigma_{\Geb,\infty}>E_{\Geb,\infty}$ holds.
Let $\Phi$ be the normalized, strictly positive ground state eigenvector of $\HV$, 
that exists according to Theorem~\ref{thmPerronFrobenius} and Theorem~\ref{thmGSH}. 
Let $\HV_n$ be the operator obtained upon putting $\ee_n$ in place of $\ee$ in the construction of 
$\HV$. Then there exists $n_0\in\NN$ such that, for all $n\ge n_0$, we find a 
normalized, strictly positive ground state eigenvector $\Phi_n$ of $\HV_n$, and 
$\lim_{n\to\infty}\Phi_n=\Phi$.
\end{thm}

\begin{proof}
With the extension of Proposition~\ref{prop-IR} given in Proposition~\ref{propIRinfty} at hand, we see that
the statement and proof of Corollary~\ref{corIRubd1} remain valid, if all ultraviolet cutoffs
$\UV_\iota$ appearing there are set equal to $\infty$. This shows that a sequence of
ground state eigenvectors $\{\Phi_n\}_{n\ge n_0}$ as in the statement exists.

Suppose for contradiction that $\{\Phi_n\}_{n\ge n_0}$ does not converge to $\Phi$. Then we
find some $\ve>0$ and a subsequence $\{\Phi_{n_\ell}\}_{\ell\in\NN}$ 
such that $\|\Phi_{n_\ell}-\Phi\|\ge\ve$, for all $\ell\in\NN$.
Then the just described modified version of Corollary~\ref{corIRubd1} further implies, however, that
$\{\Phi_{n_\ell}\}_{\ell\in\NN}$ contains another, strongly converging subsequence. Calling that 
sub-subsequence $\{\Phi_j'\}_{j\in\NN}$ and its strong limit $\Phi'$, we have $\|\Phi'-\Phi\|\ge\ve$.
Thanks to Lemma~\ref{lempert} we also know that $\HV_n\to\HV$, $n\to\infty$, in norm resolvent sense. 
Invoking Lemma~\ref{lemSRC} we see that $\Phi'$ is a normalized ground state eigenvector of $\HV$.
Since every $\Phi_j'$ is strictly positive, $\Phi'$ must be non-negative and, therefore,
$\Phi'=\Phi$; a contradiction!
\end{proof}

For convenience, we consider only the case $\Geb=\RR^3$ in the remaining part of this section.
The next two theorems complete the proofs of Theorems~\ref{introthmcont} and~\ref{introthmcontnonFock}
in the introduction:

\begin{thm}\label{thmcont}
In the situation of Theorem~\ref{thmL2cont} with $\Geb=\RR^3$, all ground state eigenvectors $\Phi$ and
$\Phi_n$, $n\in\NN$, $n\ge n_0$, have representatives which are continuous maps from $\RR^3$ into
$\sF$. Moreover, if $\{\V{x}_n\}_{n=n_0}^\infty$ is a sequence in $\RR^3$ converging to some
$\V{x}\in\RR^3$, then $\lim_{n\to\infty}\Phi_n(\V{x}_n)=\Phi(\V{x})$.
\end{thm}

\begin{proof}
We fix some $t>0$ in this proof and suppose without loss of generality that $n_0=1$.
We shall proceed in four steps.

{\em Step~1.} Assume that $V$ is continuous and bounded and let $\Psi\in L^2(\RR^3,\sF)$
be continuous and bounded as well. Denote by $\wt{T}_t^n$ the Feynman-Kac operator defined
by means of $\ee_n$, and by $\wt{T}_t$ the one corresponding to $\ee$. (Here we drop the
subscripts $\Geb=\RR^3$ and $\UV=\infty$.) Then the convergence
\begin{align}\label{corinna}
(\wt{T}^n_t\Psi)(\V{x}_n)\xrightarrow{\;\;n\to\infty\;\;}(\wt{T}_t\Psi)(\V{x}),
\end{align} 
follows in a straightforward
fashion from the dominated convergence theorem, if we take the continuity of
$\mathfrak{k}\ni h\mapsto F_{\nf{t}{2}}(h)\in\LO(\sF)$, \eqref{bdFt}, and \eqref{bduU} into account
and observe that $\ee_ne_{\V{x}_n}\to\ee e_{\V{x}}$ strongly as bounded operators on $\mathfrak{k}$.

{\em Step~2.} Let $\Psi$ be as in Step~1. For general Kato decomposable $V$, we find a sequence of 
bounded and continuous potentials $V_m:\RR^3\to\RR$, $m\in\NN$, such that
\begin{align}\label{approxVVn}
\sup_{\V{y}\in K}\EE\Big[\big|e^{-\int_0^tV_m(\V{B}_s^{\V{y}})\Id s}-
e^{-\int_0^tV(\V{B}_s^{\V{y}})\Id s}\big|^p\Big]\xrightarrow{\;\;m\to\infty\;\;}0,\quad p>0,
\end{align}
for all compact $K\subset\RR^3$; see, e.g., \cite[Prop.~2.3 and Lem.~C.6]{BHL2000}.
Let $\wt{T}^{n,m}_t$ be the Feynman-Kac operator defined
by means of $\ee_n$ and $V_m$, and $\wt{T}_t^n$ the one corresponding to $\ee_n$ and $V$,
where we set $\ee_\infty:=\ee$ and $\V{x}_\infty:=\V{x}$.
Then H\"{o}lder's inequality, \eqref{bdFt}, \eqref{bduU}, and \eqref{approxVVn} imply
\begin{align*}
\sup_{n\in\NN\cup\{\infty\}}\big\|(\wt{T}^{n,m}_t\Psi)(\V{x}_n)-(\wt{T}_t^n\Psi)(\V{x}_n)\big\|
\xrightarrow{\;\;m\to\infty\;\;}0.
\end{align*}
Hence, \eqref{corinna} holds for general $V$ as well.

{\em Step~3.} Let $\Psi\in L^2(\RR^3,\sF)$. Pick continuous and bounded
$\Psi_m\in L^2(\RR^3,\sF)$, $m\in\NN$, such that $\Psi_m\to\Psi$, as $m\to\infty$.
Define $\wt{T}_t^n$ as in Step~2.
Employing H\"{o}lder's inequality, \eqref{bdFt}, \eqref{kashmir}, and \eqref{bduU},
we find some $c_t>0$, depending only on $\sup_n|\ee_n|$ and $V$ besides $t$, such that
\begin{align*}
\sup_{n\in\NN\cup\{\infty\}}
\big\|\wt{T}^{n}_t(\Psi_m-\Psi)(\V{x}_n)\big\|\le c_t\|e^{t\Delta/2}\|_{1,\infty}^{\eh}
\|\Psi_m-\Psi\|_2\xrightarrow{\;\;m\to\infty\;\;}0.
\end{align*}
We conclude that \eqref{corinna} actually holds for general $V$ and all square-integrable $\Psi$.

{\em Step~4.} We now apply the result of Step~3 to the ground state eigenvectors $\Phi_n$ and
$\Phi$. Define $\wt{T}^{n}_t$ and $\wt{T}_t$ as in Step~1, but now for general $V$. Then
\begin{align*}
\big\|(\wt{T}^{n}_t\Phi_n)(\V{x}_n)-(\wt{T}_t\Phi)(\V{x})\big\|
&\le\big\|(\wt{T}^{n}_t\Phi)(\V{x}_n)-(\wt{T}_t\Phi)(\V{x})\big\|
+\|\wt{T}^n_t\|_{2,\infty}\|\Phi_n-\Phi\|_2,
\end{align*}
for all $n\in\NN$. Since $\|\wt{T}^n_t\|_{2,\infty}$ is uniformly bounded in $n\in\NN$
(again by H\"{o}lder's inequality, \eqref{bdFt}, \eqref{kashmir}, and \eqref{bduU}), we infer
from Theorem~\ref{thmL2cont} and Step~3 that the left hand side of the previous estimate goes to zero, 
as $n\to\infty$. Since $\HV_n\to\HV$ in norm resolvent sense, we also know that
$E_n:=\inf\sigma(\NV_n)\to E:=\inf\sigma(\NV)$, whence
$\Phi_n(\V{x}_n)=e^{tE_n}(\wt{T}_t^n\Phi_n)(\V{x}_n)$ converges to
$e^{tE}(\wt{T}_t\Phi)(\V{x})=\Phi(\V{x})$, as claimed.
\end{proof}

\begin{thm}\label{thmcontN}
Consider the case $\Geb=\RR^3$ and assume that, besides the binding condition
$\Sigma_{\RR^3,\infty}>E_{\RR^3,\infty}$, also the infrared regularity condition
$\omega^{-\nf{3}{2}}\eta\in L^2_\loc(\RR^3)$ is satisfied. Then the statements of 
Theorem~\ref{thmL2cont} and Theorem~\ref{thmcont} still hold true when $\NV:=\NV_{\RR^3,\infty}$ is
put in place of $\HV_{\RR^3,\infty}$ and $\HV_n$ is substituted by the operator $H_n$ obtained
upon replacing $\ee$ by $\ee_n$ in the construction of $\NV$.
\end{thm}

\begin{proof}
Under the condition $\omega^{-\nf{3}{2}}\eta\in L^2_\loc(\RR^3)$, the Nelson operators and their
non-Fock versions are unitarily equivalent, whence it is clear that the assertion of 
Theorem~\ref{thmL2cont} carries over to $\NV$ and $\NV_n$. To prove the convergence
$\Phi_n(\V{x}_n)\to\Phi(\V{x})$, $n\to\infty$, for ground state eigenvectors of 
$H_n$ and $H$, we can literally copy the proof of 
Theorem~\ref{thmcont}, just dropping the tildes everywhere.
\end{proof}

\begin{rem}\label{remcontrange}
Let $t>0$, $\Psi\in L^2(\RR^3,\sG)$, and write again $\HV=\HV_{\RR^3,\infty}$
and $\NV=\NV_{\RR^3,\infty}$.
If we choose a constant sequence of coupling constants $\ee_n=\ee$, $n\in\NN$, in the
proof of Theorem~\ref{thmcont}, then the result of its third step shows that $e^{-t\HV}\Psi$ 
has a unique continuous representative, which is given by the right hand side of the 
corresponding Feynman-Kac formula. The same remark applies to $e^{-tH}\Psi$. In the latter case
this re-proves a part of \cite[Thm.~8.8]{MatteMoeller2017}.
\end{rem}

%%%%%%%%%%%%%%%%%%%%%%%%%%%%%%%%%%%%%%%%%%%%%%%
%%%%%%%%%%%%%%%%%%%%%%%%%%%%%%%%%%%%%%%%%%%%%%%
%%%%%%%%%%%%%%%%%%%%%%%%%%%%%%%%%%%%%%%%%%%%%%%

\section{Absence of Ground States}\label{secabsence}

\noindent
In this short section we complement the existence results of the previous one by proving 
the non-existence of ground state eigenvectors of the massless ($\mu=0$)
Nelson operators $\NV_{\Geb,\UV}$ with $\UV\in(0,\infty]$ in the infrared singular case where
\begin{align}\label{IRsing}
\ee^2\int_{\{|\V{k}|<1\}}\frac{\eta(\V{k})^2}{\omega(\V{k})^3}\Id\V{k}=\infty.
\end{align}
If a binding condition is fulfilled, then this result is new only for $\UV=\infty$; see 
\cite{Hirokawa2006,Panati2009}.
Thanks to an additional argument based on the bound \eqref{resexploc2aN} we can,
however, drop the binding condition in the next theorem. Apart from this its proof is a simple 
modification of the one given for finite $\UV$ in \cite[Thm.~2.5(2)]{DerezinskiGerard2004} and 
\cite[\textsection5]{Panati2009}. Results based on path integration techniques proving the absence of 
ground states in Nelson type models with confining potentials and ultraviolet regularized interactions  
can be found in \cite{GHPS2009,GHPS2012,LHB2011,LorincziMinlosSpohn2002}.
Ultraviolet regularized fiber Hamiltonians
are treated non-perturbatively in \cite{Dam2018}. A non-perturbative proof of the absence
of ground states for renormalized fiber Hamiltonians in the massless Nelson model has been achieved
in~\cite{DamHinrichs2021}.

\begin{thm}\label{thmabsence}
Let $\mu=0$, $\UV\in(0,\infty]$, and assume that \eqref{IRsing} is satisfied.
Then $E_{\Geb,\UV}$ is not an eigenvalue of $\NV_{\Geb,\UV}$.
\end{thm}

\begin{proof}
We prove the theorem only for $\NV_{\Geb,\infty}$. Then the general case can be included by
choosing $\eta$ appropriately. Let $\Phi_\infty\in\Ran(1_{\{E_{\Geb,\infty}\}}(\NV_{\Geb,\infty}))$.
We have to show that $\Phi_\infty=0$. To this end we proceed in two steps. In the first one we
derive a formula for $(a\Phi_\infty)(\V{p})$ by modifying the usual ``pull-through type"
argument. In the second step we present the aforementioned version of the argument from
\cite{DerezinskiGerard2004,Panati2009} that avoids the use of a binding condition.

{\em Step 1.} We pick some $\chi\in C_0(\RR,\RR)$ with $\chi(0)=1$ 
and set $\wt{\chi}(t):=t\chi(t)$, $t\in\RR$. For every $\UV\in(0,\infty)$, we further put 
$\Phi_\UV:=\chi(\NV_{\Geb,\UV}-E_{\Geb,\UV})\Phi_\infty\in\dom(\NV_{\Geb,\UV})$. 
Then the convergence in norm resolvent sense
$\NV_{\Geb,\UV}-E_{\Geb,\UV}\to\NV_{\Geb,\infty}-E_{\Geb,\infty}$ entails
\begin{align}\label{bjarke0}
\Phi_\UV\longrightarrow\Phi_\infty,\quad\wt{\chi}(\NV_{\Geb,\UV}-E_{\Geb,\UV})\Phi_\infty
\longrightarrow0,\qquad\UV\to\infty.
\end{align}

Let $\V{p}\in\RR^3\setminus\{0\}$ and let $g\in L^2(\RR^3)$ have a compact support in
$\RR^3\setminus\{0\}$. We further set
\begin{align*}
I_{\Geb,\UV'}(\V{p}):=(\NV_{\Geb,\UV'}-E_{\Geb,\UV'}+\omega(\V{p}))^{-1},\quad\UV'\in(0,\infty].
\end{align*}
According to \cite{GriesemerWuensch2017} the form domain of every 
$\NV_{\Geb,\UV'}$ with finite or infinite $\UV'$ is contained in the form domain of 
$\Id\Gamma(\omega)$. (See also the explanation in the proof of Lemma~\ref{lemstrresconvdGN} below.) 
In particular, the range of $I_{\Geb,\UV'}(\V{p})$ is contained in the domain of $\ad(\omega^sg)$, 
for all $s\ge0$. With the help of  \eqref{OpformelNV} we can proceed along the lines of the proof of 
Proposition~\ref{prop-IR}, with $\Phi_\UV$ put in place of the ground state eigenvector $\Phi_{\Geb,\UV}$ 
considered there, to obtain the relation
\begin{align}\nonumber
\SPn{\ad(g)I_{\Geb,\UV}(\V{p})\Psi}{\wt{\chi}(\NV_{\Geb,\UV}-E_{\Geb,\UV})\Phi_\infty}
&=\SPn{\ad(g)\Psi}{\Phi_\UV}
\\\nonumber
&\quad+\SPn{\ad([\omega-\omega(\V{p})]g)I_{\Geb,\UV}(\V{p})\Psi}{\Phi_\UV}
\\\label{bjarke1}
&\quad+\SPn{I_{\Geb,\UV}(\V{p})\Psi}{\SPn{g}{e_{\hat{\V{x}}}f_\UV}\Phi_\UV},
\end{align}
for all $\Psi\in\sD(\Geb,\fdom(\Id\Gamma(\omega)))$ and finite $\UV$.

Let $\UV_0>0$ be so large that the open ball of radius $\UV_0$ about the origin in $\RR^3$ contains 
the support of $g$. Then the bounded multiplication operator defined by the function
\begin{align}\label{bjarke1b}
\SPn{g}{e_{\V{x}}f_\UV}=\int_{\RR^3}\ol{g(\V{k})}e^{-i\V{k}\cdot\V{x}}f_\infty(\V{k})\Id\V{k}
=:Q_{g}(\V{x}),\quad\V{x}\in\Geb,
\end{align}
actually is independent of $\UV\in[\UV_0,\infty)$. The subsequent 
Lemma~\ref{lemstrresconvdGN} further implies
\begin{align}\label{bjarke2}
\ad(\omega^sg)I_{\Geb,\UV}(\V{p})\Psi\xrightarrow{\;\;\UV\to\infty\;\;}
\ad(\omega^sg)I_{\Geb,\infty}(\V{p})\Psi,
\end{align}
for all $\Psi\in L^2(\Geb,\sF)$ and $s\ge0$. Passing to the limit $\UV\to\infty$ in \eqref{bjarke1} and
taking \eqref{bjarke0}, \eqref{bjarke1b}, \eqref{bjarke2}, and 
$\Phi_\infty\in\fdom(\Id\Gamma(\omega))$ into account, we deduce that
\begin{align}\nonumber
a(g)\Phi_\infty=I_{\Geb,\infty}(\V{p})a([\omega(\V{p})-\omega]g)\Phi_\infty
-I_{\Geb,\infty}(\V{p})Q_{g}\Phi_\infty.
\end{align}
As in the proof of Proposition~\ref{prop-IR}, an argument based on the Lebesgue point theorem
now leads to the first identity in
\begin{align}\nonumber
(a\Phi_\infty)(\V{p})&=-\ee\eta(\V{p}){\omega(\V{p})^\mh}I_{\Geb,\infty}(\V{p})
e^{-i\V{p}\cdot\hat{\V{x}}}\Phi_\infty
\\\label{IRmassless}
&=-\ee\eta(\V{p}){\omega(\V{p})^{-\nf{3}{2}}}\Phi_\infty
-\ee\eta(\V{p}){\omega(\V{p})^\mh}I_{\Geb,\infty}(\V{p})
(e^{-i\V{p}\cdot\hat{\V{x}}}-1)\Phi_\infty,
\end{align}
for a.e. $\V{p}\not=0$.

{\em Step 2.}
Let $N\in\NN$ and let $\vt_N$ denote multiplication with the characteristic function of
$\{|\V{x}|\le N\}$ in $L^2(\Geb,\sF)$. Define
$F_{\V{p}}(\V{x}):=\omega(\V{p})^\eh\langle\V{x}\rangle$, for all $\V{p},\V{x}\in\RR^3$.
Then $|\nabla F_{\V{p}}|\le\omega(\V{p})^\eh$. For non-zero $\V{p}$, Ex.~\ref{exresexpabs}
implies that $e^{-F_{\V{p}}}I_{\Geb,\infty}(\V{p})e^{F_{\V{p}}}$ extends to a bounded operator
on $L^2(\Geb,\sF)$ with norm $\le2/\omega(\V{p})$. Together with the elementary bounds
$$
e^{-F_{\V{p}}(\V{x})}|e^{-i\V{p}\cdot{\V{x}}}-1|\le|\V{p}||\V{x}|e^{-|\V{p}|^\eh|\V{x}|}
\le\frac{|\V{p}|^\eh}{e},
$$
this gives
\begin{align*}
0<|\V{p}|\le1\quad\Rightarrow\quad\|\eta(\V{p}){\omega(\V{p})^\mh}\vt_NI_{\Geb,\infty}(\V{p})
(e^{-i\V{p}\cdot\hat{\V{x}}}-1)\Phi_\infty\|&\le\frac{2e^N}{|\V{p}|}\|\Phi_\infty\|,
\end{align*}
which in conjunction with \eqref{IRmassless} permits to get
\begin{align*}
\int_{\{|\V{p}|\le1\}}\|(a\vt_N\Phi_\infty)(\V{p})
+\ee\eta(\V{p}){\omega(\V{p})^{-\nf{3}{2}}}\vt_N\Phi_\infty\|^2\Id\V{p}
&\le16\pi\ee^2e^{2N}\|\Phi_\infty\|^2.
\end{align*}
Under the condition \eqref{IRsing} the previous bound can, however, only be true provided that
$\vt_N\Phi_\infty=0$; compare \cite[Lem.~2.6]{DerezinskiGerard2004}. 
Since $N\in\NN$ was arbitrary, we find $\Phi_\infty=0$.
\end{proof}

The next lemma holds again for arbitrary boson masses $\mu\ge0$.

\begin{lem}\label{lemstrresconvdGN}
For every $\UV\in[0,\infty]$, the form domain of $\NV_{\Geb,\UV}$ is contained in
the form domain of $\Id\Gamma(\omega)$ and
\begin{align}\label{strresconvdGN}
\Id\Gamma(\omega)^\eh(\NV_{\Geb,\UV}-E_{\Geb,\UV}+\zeta)^{-1}\Psi
\xrightarrow{\;\;\UV\to\infty\;\;}
\Id\Gamma(\omega)^\eh(\NV_{\Geb,\infty}-E_{\Geb,\infty}+\zeta)^{-1}\Psi,
\end{align} 
for all $\Psi\in L^2(\Geb,\sF)$ and $\zeta>0$.
\end{lem}

\begin{proof}
The first assertion is proved in \cite{GriesemerWuensch2017}. (It follows from 
$\fdom(\HV_{\Geb,K,\UV})={\QGV}\subset\fdom(\Id\Gamma(\omega))$, 
$0<K<\UV\le\infty$, the relations
\eqref{trafoNelsonHamUV} and \eqref{defNelsonHam}, and Lemma~\ref{lemGrWubd}.)

To verify \eqref{strresconvdGN} we write 
$\theta:=\int_{\Geb}^\oplus(1+\Id\Gamma(\omega))\Id\V{x}$,
recall the notation \eqref{defGrosstrafo} for the Gross transformation, and pick some $K>0$.
In view of the second convergence relation asserted in Theorem~\ref{thmrbUV}(4), 
the strong convergence $G_{K,\UV}\to G_{K,\infty}$, $\UV\to\infty$, and the transformation formulas
\eqref{trafoNelsonHamUV} and \eqref{defNelsonHam} it suffices to show that
\begin{align}\label{Gthetaeh}
\sup_{K\le\UV\le\infty}\|\theta^\eh G_{K,\UV}^*\theta^\mh\|&\le1+\|\omega^\eh\beta_\infty\|,
\end{align}
as well as
\begin{align}\label{astronauta}
\theta^\eh G_{K,\UV}^*\theta^\mh\Psi\xrightarrow{\;\;\UV\to\infty\;\;}
\theta^\eh G_{K,\infty}^*\theta^\mh\Psi,\quad\Psi\in L^2(\Geb,\sF).
\end{align}
But \eqref{Gthetaeh} is a direct consequence of Lemma~\ref{lemGrWubd} (which presents a bound
from \cite{GriesemerWuensch2017}), while \eqref{astronauta} follows from 
Lemma~\ref{lemGrWubd}, Lemma~\ref{lemstrcontWeyldG}, and the dominated convergence theorem.
\end{proof}

%%%%%%%%%%%%%%%%%%%%%%%%%%%%%%%%%%%%%%%%%%%%%%%
%%%%%%%%%%%%%%%%%%%%%%%%%%%%%%%%%%%%%%%%%%%%%%%
%%%%%%%%%%%%%%%%%%%%%%%%%%%%%%%%%%%%%%%%%%%%%%%

\section{Path Measures Associated with Ground States}\label{secGibbs}

\noindent
{\em In this section we shall always assume that the binding condition 
\eqref{bindcond} and the infrared regularity condition $\omega^{-3}\eta^2\in L^1_\loc(\RR^3)$
are satisfied.} For simplicity we shall restrict our attention to the case 
\begin{align*}
\Geb=\RR^3\quad\text{and}\quad\UV=\infty,
\end{align*}
and we shall drop the subscripts $\Geb$ and $\UV$ in the notation, so that for instance 
\begin{align*}
\NV=\NV_{\RR^3,\infty},\quad E=E_{\RR^3,\infty},
\end{align*} 
just as in the introduction. (Then the ultraviolet regular case actually is
included since we still have the freedom to choose $\eta$.)
The symbol $\Phi$ will denote the continuous representative of the strictly 
positive normalized eigenvector of the Nelson operator $\NV$ corresponding to the
infimum of its spectrum $E$.
Finally, we fix some $\V{x}\in\RR^3$ throughout the whole section.

Our objective is to construct certain path measures associated with $\Phi$ and $\V{x}$
that permit to obtain nontrivial bounds on various expectation values with respect to 
$\Phi(\V{x})$. A similar analysis of ground state expectations in non-relativistic quantum field
theory has been pursued first in \cite{BHLMS2002}. The analogues of the path measures considered
in the latter article (for finite $\UV$) are obtained upon integrating the path measures constructed here
with respect to the probability density $\|\Phi(\cdot)\|$. 

%%%%%%%%%%%%%%%%%%%%%%%%%%%%%%%%%%%%%%%%%%%%%%%
%%%%%%%%%%%%%%%%%%%%%%%%%%%%%%%%%%%%%%%%%%%%%%%

\subsection{Martingales Associated with Ground States}\label{ssecGSmart}

\noindent
We start by defining the process
\begin{align*}
M_{t}(\V{x}):=e^{tE}W_{t}^V(\V{x})^*\Phi(\V{B}_t^{\V{x}}),\quad t\ge0,
\end{align*}
where we used the shorthand
\begin{align}\label{WVshorthand}
W_{t}^V(\V{x})&:=e^{-\int_0^tV(\V{B}_s^{\V{x}})\Id s}W_{\infty,t}(\V{x}).
\end{align}
Furthermore, we let $(\fG_t)_{t\ge0}$ denote the completion of the filtration generated by
the Brownian motion $\V{B}=(\V{B}_t)_{t\ge0}$, which is automatically right continuous.

\begin{lem}\label{lemMmart}
The process $(M_{t}(\V{x}))_{t\ge0}$ is a continuous $\sF$-valued martingale with respect to the
 filtration $(\fG_t)_{t\ge0}$. For each $t\ge0$, the random variable $M_{t}(\V{x})$ is in
 every $L^p(\Omega,\PP;\sF)$ with $1\le p<\infty$.
\end{lem}

\begin{proof}
On account of \eqref{kashmir}, \eqref{WinLp}, and the boundedness of $\Phi$ it is clear
that $M_{t}(\V{x})\in L^p(\Omega,\PP;\sF)$, for all $t\ge0$ and finite $p\ge1$.
For all $0\le s\le t<\infty$, a Markov property proven in 
\cite[Prop.~5.8 (4)]{MatteMoeller2017} further implies the first equality in
\begin{align*}
\EE^{\fF_s}[M_{t}(\V{x})]
&=e^{tE}W_{s}^V(\V{x})^*(T_{t-s}\Phi)(\V{B}_s^{\V{x}})=M_{s}(\V{x}),
\quad\text{$\PP$-a.s.;}
\end{align*}
recall that $T_{\tau}=T_{\RR^3,\infty,\tau}$ as we are dropping all subscripts $\RR^3$ and $\infty$
in the notation. The second equality follows from the relation $e^{(t-s)E}T_{t-s}\Phi=\Phi$.
From \cite[Rem.~5.4(1)\&(2)]{MatteMoeller2017} we infer that
$M_{s}(\V{x})$ is $\fG_s$-measurable and we conclude.
\end{proof}

%%%%%%%%%%%%%%%%%%%%%%%%%%%%%%%%%%%%%%%%%%%%%%%
%%%%%%%%%%%%%%%%%%%%%%%%%%%%%%%%%%%%%%%%%%%%%%%

\subsection{Construction of Path Measures Associated with Ground States}\label{ssecGSpathm}

\noindent
In what follows we suppose that $\V{B}_-=(\V{B}_{-t})_{t\ge0}$ be a three-dimensional Brownian 
motion on $(\Omega,\fF,(\fF_t)_{t\ge0},\PP)$ independent from $\V{B}$.  We shall put a minus sign in 
front of the time parameter of all probabilistic objects defined by means of $\V{B}_-$. This
leads to better looking formulas and facilitates comparison of our discussion with the previous
literature where the notion of two-sided Brownian motion was used. For instance, 
$(\fG_{-t})_{t\ge0}$, $({W}_{-t}^V(\V{x}))_{t\ge0}$, and $({M}_{-t}(\V{x}))_{t\ge0}$ are 
defined as before but with $\V{B}_-$ put in place of $\V{B}$. Then
Lemma~\ref{lemMmart} says that $({M}_{-t}(\V{x}))_{t\ge0}$ is $(\fG_{-t})_{t\ge0}$-adapted.
In particular, the path maps ${M}_{\bullet}(\V{x})$ and ${M}_{-\bullet}(\V{x})$ are independent.

For all $t\ge0$, we further put $\fH_t:=\sigma(\fG_t\cup\fG_{-t})$ and
\begin{align*}
m_{t}(\V{x}):=\SPn{{M}_{-t}(\V{x})}{{M}_{t}(\V{x})},\quad t\ge0.
\end{align*}

\begin{lem}\label{lemmmart}
The process $(m_{t}(\V{x}))_{t\ge0}$ is a positive martingale with respect to the filtration 
$(\fH_t)_{t\ge0}$ starting at $\|\Phi(\V{x})\|^2>0$. For each $t\ge0$, the random variable 
$m_{t}(\V{x})$ is in every $L^p(\Omega,\PP)$ with $1\le p<\infty$. In particular,
\begin{align}\label{Gibbs3}
\EE[m_{t}(\V{x})]&=\|\Phi(\V{x})\|^2,\quad t\ge0.
\end{align}
\end{lem}

\begin{proof}
The existence of the $p$-th moments of $m_t(\V{x})$ for all finite $p\ge1$ follows immediately 
from Lemma~\ref{lemMmart}. Let $t>s\ge0$. Since $\fH_s\subset\sigma(\fG_{-t}\cup\fG_s)$ and 
$M_{-t}(\V{x})$ is $\sigma(\fG_{-t}\cup\fG_s)$-measurable,
\begin{align*}
\EE^{\fH_s}[m_{t}(\V{x})]&=\EE^{\fH_s}\Big[\EE^{\sigma(\fG_{-t}\cup\fG_s)}[
\SPn{{M}_{-t}(\V{x})}{{M}_{t}(\V{x})}]\Big]
\\
&=\EE^{\fH_s}\Big[\SPB{{M}_{-t}(\V{x})}{\EE^{\sigma(\fG_{-t}\cup\fG_s)}[{M}_{t}(\V{x})]}\Big]
\\
&=\EE^{\sigma(\fG_{-s}\cup\fG_s)}\Big[\SPB{{M}_{-t}(\V{x})}{\EE^{\fG_s}[{M}_{t}(\V{x})]}\Big]
\\
&=\SPB{\EE^{\sigma(\fG_{-s}\cup\fG_s)}[{M}_{-t}(\V{x})]}{\EE^{\fG_s}[{M}_{t}(\V{x})]}
\\
&=\SPB{\EE^{\fG_{-s}}[{M}_{-t}(\V{x})]}{\EE^{\fG_s}[{M}_{t}(\V{x})]}.
\end{align*}
In the third step we used the independence of $\fG_{-t}$ and $\fG_s$ and the $\fG_{-t}$-independence
of $M_{t}(\V{x})$. Likewise, we exploited the independence of $\fG_{-s}$ and $\fG_s$ and the
$\fG_s$-independence of $M_{-t}(\V{x})$ in the last step. In view of Lemma~\ref{lemMmart} the term
in the last line equals $m_{s}(\V{x})$.
\end{proof}

From now on we choose $\Omega$ to be equal to 
$$
\Omega_{\Wi}:=C([0,\infty),\RR^6),
$$
as we shall exploit that $\Omega_{\Wi}$ is a Polish space when equipped with the topology 
associated with locally uniform convergence. The symbol
$\PP_{\Wi}$ will denote the completed Wiener measure on $\Omega_{\Wi}$ giving probability
one to the set of continuous paths starting at $0$. The symbol
$\EE_{\Wi}$ will denote the corresponding expectation; similarly for conditionial 
expectations. Furthermore, $\V{B}$ and $\V{B}_{-}$ will from now on stand for the first three
and the last three components, respectively, of the canonical evaluation process on $\Omega_{\Wi}$. 
We put 
\begin{align*}
\mr{\fF}_t:=\sigma(\V{B}_s,\V{B}_{-s}:\,0\le s\le t),\quad\text{for all $t\ge0$, and}\quad
\mr{\fF}_*:=\bigcup_{t\ge0}\mr{\fF}_t.
\end{align*} 
Then $\mr{\fF}_*$ generates the Borel $\sigma$-algebra associated with the Polish 
topology on $\Omega_{\Wi}$,
$$
\sigma(\mr{\fF}_*)=\fB(\Omega_{\Wi}).
$$
From now on the filtration $(\fF_t)_{t\ge0}$ is chosen to be the completion of $(\mr{\fF}_t)_{t\ge0}$,
which is indeed right continuous. Then $\fF_t=\fH_t$, for all $t\ge0$, where $\fH_t$ was defined 
in front of Lemma~\ref{lemmmart}.

We are now in a position to introduce path measures associated with $\Phi$ and $\V{x}$
by means of a standard construction: First, we define $\mu_{\V{x},*}:\mr{\fF}_*\to\RR$ by
\begin{align*}
\mu_{\V{x},*}(A)&:=\frac{1}{\|\Phi(\V{x})\|^2}
\EE_{\Wi}[1_Am_{t}(\V{x})],\quad A\in\mr{\fF}_t,\,t\ge0.
\end{align*}
The set function $\mu_{\V{x},*}$ is well-defined in this way since, by virtue of Lem.~\ref{lemmmart},
\begin{align*}
\EE_{\Wi}[1_Am_{t}(\V{x})]&=\EE_{\Wi}\Big[1_A\EE_{\Wi}^{\fF_s}[m_{t}(\V{x})]\Big]
=\EE_{\Wi}[1_Am_{s}(\V{x})],\quad A\in\mr{\fF}_s,\,0\le s\le t.
\end{align*}
In view of \eqref{Gibbs3}, each restriction $\mu\restr_{\mr{\fF}_t}$ is a probability 
measure on $\mr{\fF}_t$, i.e., $\mu_{\V{x},*}$ is a {\em promeasure} in the nomenclature of
\cite{HackenbrochThalmaier1994}. By a wellknown result 
(see, e.g., \cite[Satz~1.25${}'$]{HackenbrochThalmaier1994})
every such promeasure on $\mr{\fF}_*$ is automatically $\sigma$-additive.
(Here the fact that $\Omega_{\Wi}$ or, more precisely, the spaces $C([0,t],\RR^6)$, $t>0$, are
Polish is used.)

\begin{defn}\label{defnpathmeasGSx}
By the preceding remarks and Carath\'{e}odory's extension theorem, $\mu_{\V{x},*}$ has a unique 
extension to a probability measure on $(\Omega_{\Wi},\fB(\Omega_{\Wi}))$.  We denote this extension 
by $\mu_{\V{x},\infty}$ and call it the {\em path measure associated with $\Phi$ and $\V{x}$.}
\end{defn}

%%%%%%%%%%%%%%%%%%%%%%%%%%%%%%%%%%%%%%%%%%%%%%%%%%

\subsection{Analysis of Ground State Expectations Via Path Measures}\label{ssecGSEE}

\noindent
To analyze expectations with respect to $\mu_{\V{x},\infty}$, it is helpful to introduce a 
new family of measures, given by more explicit formulas, that converges to 
$\mu_{\V{x},\infty}$ in a suitable sense. For this we pick an arbitrary square-integrable, 
non-negative function $g:\RR^3\to\RR$, that is not a.e. equal to zero. For every $t>0$, we then define a probability measure 
$\mu_{\V{x},t}$ on $\fB(\Omega_{\Wi})$ by
\begin{align}\label{maike1}
\mu_{\V{x},t}(A):=\frac{1}{Z_{\V{x},t}}\EE_{\Wi}[1_A\sL_{t}(\V{x})g(\V{B}_{-t}^{\V{x}})
g(\V{B}_t^{\V{x}})],\quad A\in\fB(\Omega_{\Wi}),
\end{align}
with a normalization constant $Z_{\V{x},t}$ and
\begin{align*}
\sL_{t}(\V{x})&:=e^{-\int_0^t(V(\V{B}_s^{\V{x}})+V(\V{B}_{-s}^{\V{x}}))\Id s
+u_{t}+{u}_{-t}+\SPn{{U}_{-t}^-}{{U}_{t}^-}}.
\end{align*}
Here we continue using our convention to put a minus sign in front of the time parameter of all
processes that are defined by the earlier formulas but with $\V{B}_-$ put in place of $\V{B}$.
Let us explain why $Z_{\V{x},t}$ is indeed strictly positive: Recall first that
$\epsilon(0)=(1,0,0,\ldots \ )$ is the vacuum vector in $\sF$. Employing \eqref{Fexpv}, 
\eqref{Fstarexpv}, and \eqref{defW} we deduce that
\begin{align}\label{maike2a}
W_{t}^V(\V{x})^*\epsilon(0)&=e^{u_t-\int_0^tV(\V{B}_s^{\V{x}})\Id s}\epsilon(-e_{\V{x}}U_t^{-})
=e^{u_t-\int_0^tV(\V{B}_s^{\V{x}})\Id s}\Gamma(e_{\V{x}})\epsilon(-U_{t}^-),
\end{align}
and analogously for $-t$. This permits to get 
\begin{align}\label{maike2}
\sL_{t}(\V{x})g(\V{B}_{-t}^{\V{x}})g(\V{B}_t^{\V{x}})
&=\SPn{{W}_{-t}^V(\V{x})^*g(\V{B}_{-t}^{\V{x}})\epsilon(0)}{
W_{t}^V(\V{x})^*g(\V{B}_t^{\V{x}})\epsilon(0)},
\end{align}
for all $t>0$. The independence of random variables indexed by $t$ and $-t$ now implies
\begin{align}\nonumber
Z_{\V{x},t}&=\EE_{\Wi}\Big[\SPn{{W}_{-t}^V(\V{x})^*g(\V{B}_{-t}^{\V{x}})
\epsilon(0)}{W_{t}^V(\V{x})^*g(\V{B}_t^{\V{x}})\epsilon(0)}\Big]
\\\nonumber
&=\SPB{\EE_{\Wi}[{W}_{-t}^V(\V{x})^*g(\V{B}_{-t}^{\V{x}})\epsilon(0)]}{\EE_{\Wi}[
W_{t}^V(\V{x})^*g(\V{B}_t^{\V{x}})\epsilon(0)]}
\\\label{maike3}
&=\|(T_{t}g\epsilon(0))(\V{x})\|^2.
\end{align}
Here the vector $g\epsilon(0)\in L^2(\RR^3,\sF)$ is non-negative and non-zero. 
Since, for $t>0$, we know that $T_{t}$ is positivity improving and 
elements in its range are continuous, we deduce that $(T_{t}g\epsilon(0))(\V{x})$ is
strictly positive and in particular has a non-vanishing norm.

The connection of the measures in \eqref{maike1} to the ground state path measure 
$\mu_{\V{x},\infty}$ is revealed by the
following theorem which is an analogue (here for fixed $\V{x}$ and without ultraviolet cutoff) 
of a result that first appeared in \cite{BHLMS2002}. Its
proof requires, however, a new discussion of certain convergence properties.

\begin{thm}\label{thmlocconvmu}
The family of probability measures $\{\mu_{\V{x},t}\}_{t>0}$ converges {\em locally} to 
$\mu_{\V{x},\infty}$ in the sense that $\mu_{\V{x},t}(A)\to\mu_{\V{x},\infty}(A)$, 
$t\to\infty$, for all $A\in\mr{\fF}_*$.
\end{thm}

\begin{proof}
Let $t>0$ and $A\in\mr{\fF}_t$. For all $\tau\ge t$, the formulas \eqref{maike1} and 
\eqref{maike2} imply
\begin{align*}
\mu_{\V{x},\tau}(A)&=\frac{1}{Z_{\V{x},\tau}}\EE_{\Wi}\bigg[1_A\EE_{\Wi}^{\fF_t}\Big[
\SPn{{W}_{-\tau}^V(\V{x})^*g(\V{B}_{-\tau}^{\V{x}})\epsilon(0)}{
W_{\tau}^V(\V{x})^*g(\V{B}_\tau^{\V{x}})\epsilon(0)}\Big]\bigg].
\end{align*}
Here the Markov property derived in \cite[Prop.~5.8(4)]{MatteMoeller2017} shows that
\begin{align*}
\EE_{\Wi}^{\fF_t}[W_{\tau}^V(\V{x})^*g(\V{B}_\tau^{\V{x}})\epsilon(0)]=W_{t}^V(\V{x})^*
(T_{\tau-t}g\epsilon(0))(\V{B}_t^{\V{x}}),
\end{align*}
and analogously for the objects indexed by $-\tau$. Thus,
\begin{align}\label{Gibbs17}
\mu_{\V{x},\tau}(A)&=\EE_{\Wi}\Big[1_A\SPn{{W}_{-t}^V(\V{x})^*\Psi_{t,\tau}(\V{B}^{\V{x}}_{-t})}{
W_{t}^V(\V{x})^*\Psi_{t,\tau}({\V{B}}^{\V{x}}_t)}\Big],
\end{align}
for all $\tau\ge t$, with (recall also \eqref{maike3})
\begin{align*}
\Psi_{t,\tau}&:=\frac{1}{\|(T_{\tau}g\epsilon(0))(\V{x})\|}T_{\tau-t}g\epsilon(0).
\end{align*}
Next, we claim that
\begin{align}\label{Gibbs18}
\Psi_{t,\tau}&\xrightarrow{\;\;\tau\to\infty\;\;}
\frac{e^{tE}}{\|\Phi(\V{x})\|}\Phi\quad\text{in $L^2(\RR^3,\sF)$.}
\end{align}
In fact, the spectral calculus implies $\Upsilon_s\to\Phi$, $s\to\infty$, in $L^2(\RR^3,\sF)$ with 
$$
\Upsilon_s:=\frac{1}{\|T_{s}g\epsilon(0)\|}T_{s}g\epsilon(0),\quad s\ge0,
$$
and since $T_{t}$ maps $L^2(\RR^3,\sF)$ continuously into $C_b(\RR^3,\sF)$,
the space of bounded continuous functions from $\RR^3$ to $\sF$, it follows that
\begin{align*}
\Psi_{t,\tau}=\frac{1}{\|(T_{t}\Upsilon_{\tau-t})(\V{x})\|}\Upsilon_{\tau-t}
\xrightarrow{\;\;\tau\to\infty\;\;}\frac{1}{\|(T_{t}\Phi)(\V{x})\|}\Phi\quad\text{in $L^2(\RR^3,\sF)$.}
\end{align*}
Here we used again that 
$(T_{t}T_{\tau-t}g\epsilon(0))(\V{x})=(T_{\tau}g\epsilon(0))(\V{x})$ is
strictly positive, for all $\tau\ge t>0$. We conclude the justification of \eqref{Gibbs18} by using
$(T_{t}\Phi)(\V{x})=e^{-tE}\Phi(\V{x})$.

Since $\Psi_{t,\tau}=T_{t}\Psi_{2t,\tau}$, for all $\tau\ge2t$, and $T_{t}$ is continuous from 
$L^2(\RR^3,\sF)$ into $C_b(\RR^3,\sF)$, we may further conclude that, 
pointwise on $\Omega_{\Wi}$, 
\begin{align*}
\Psi_{t,\tau}({\V{B}}^{\V{x}}_t)&\xrightarrow{\;\;\tau\to\infty\;\;}\frac{e^{2tE}}{\|\Phi(\V{x})\|}
(T_{t}\Phi)({\V{B}}^{\V{x}}_t)=\frac{e^{tE}}{\|\Phi(\V{x})\|}\Phi({\V{B}}^{\V{x}}_t),
\\
\sup_{\tau\ge2t}\|\Psi_{t,\tau}({\V{B}}^{\V{x}}_t)\|&\le\|T_{t}\|_{2,\infty}
\sup_{\tau\ge2t}\|\Psi_{2t,\tau}\|<\infty,
\end{align*}
and similarly for $\V{B}_-$. Since $\|{W}_{-t}^V(\V{x})\|\|W_{t}^V(\V{x})\|$ is $\PP_{\Wi}$-integrable,
we may therefore pass to the limit $\tau\to\infty$ under the expectation in \eqref{Gibbs17} with the 
help of the dominated convergence theorem to see that 
$\mu_{\V{x},\tau}(A)\to\mu_{\V{x},*}(A)=\mu_{\V{x},\infty}(A)$.
\end{proof}

The previous theorem implies a formula, Eqn.~\eqref{clara100} in the next corollary, 
for expectation values with respect to $\Phi(\V{x})$. This formula will be the starting point for 
the proof of Theorem~\ref{introthmsupexpNum}.
Upon integrating \eqref{clara100} with respect to $\V{x}$ we also get a formula for the ground
state expectation of any bounded decomposable operator on $L^2(\RR^3,\sF)$, at least when the
somewhat implicit and strong measurability and convergence conditions in the next corollary can be 
verified for every $\V{x}\in\RR^3$.

\begin{cor}[Ground state expectations via path measures]\label{corGSEEPM}
Let $\cO\in\LO(\sF)$ and define
\begin{align}\label{defwtKO}
\wt{K}_{\V{x},t}[\cO]&:=
\frac{\SPn{\epsilon(-{U}_{-t}^-)}{\Gamma(e_{\V{x}}^*)\cO\Gamma(e_{\V{x}})\epsilon(-U_{t}^-)}}{
\SPn{\epsilon(-{U}_{-t}^-)}{\epsilon(-U_{t}^-)}},\quad t\ge0.
\end{align}
Assume that $(\wt{K}_{\V{x},t}[\cO])_{t\ge0}$ has a $(\mr{\fF}_t)_{t\ge0}$-adapted modification
that we henceforth denote by $({K}_{\V{x},t}[\cO])_{t\ge0}$.
Assume further there exists a bounded $\fB(\Omega_{\Wi})$-measurable function
$K_{\V{x},\infty}[\cO]:\Omega_{\Wi}\to\RR$ such that
$K_{\V{x},t}[\cO]\to K_{\V{x},\infty}[\cO]$, $t\to\infty$, 
{\em uniformly on all of $\Omega_{\Wi}$}. Then
\begin{align}\label{clara100}
\SPn{\Phi(\V{x})}{\cO\Phi(\V{x})}
=\|\Phi(\V{x})\|^2\int_{\Omega_{\Wi}}K_{\V{x},\infty}[\cO]\Id\mu_{\V{x},\infty}.
\end{align}
\end{cor}

\begin{proof}
The discussion in the proof of Theorem~\ref{thmlocconvmu} shows that
\begin{align*}
\frac{1}{\|(T_{t}g\epsilon(0))(\V{x})\|}(T_{t}g\epsilon(0))(\V{x})&=(T_{1}\Psi_{1,t})(\V{x})
\\
&\xrightarrow{\;\;t\to\infty\;\;}\frac{e^{E}}{\|\Phi(\V{x})\|}
(T_{1}\Phi)(\V{x})=\frac{1}{\|\Phi(\V{x})\|}\Phi(\V{x}).
\end{align*}
Employing the defining formula for $T_t$ and the independence of the processes indexed by
$t$ and $-t$ as in \eqref{maike3}, we further find
\begin{align*}
\SPn{(T_{t}g\epsilon(0))(\V{x})}{\cO(T_{t}g\epsilon(0))(\V{x})}
&=\EE_{\Wi}\Big[\SPn{W_{-t}^V(\V{x})^*g(\V{B}_{-t}^{\V{x}})\epsilon(0)}{\cO
W_{t}^V(\V{x})^*g(\V{B}_{t}^{\V{x}})\epsilon(0)}\Big],
\end{align*}
for all $t>0$. Combining this with \eqref{maike2a} and taking the $\mr{\fF}_t$-measurability of 
$K_{\V{x},t}[\cO]$ into account in the last step, we thus obtain
\begin{align*}
\frac{\SPn{\Phi(\V{x})}{\cO\Phi(\V{x})}}{\|\Phi(\V{x})\|^2}
&=\lim_{t\to\infty}\frac{\SPn{(T_{t}g\epsilon(0))(\V{x})}{
\cO(T_{t}g\epsilon(0))(\V{x})}}{\|(T_{t}g\epsilon(0))(\V{x})\|^2},
\\
&=\lim_{t\to\infty}\frac{1}{Z_{\V{x},t}}\EE_{\Wi}\Big[\sL_{t}(\V{x})g(\V{B}_{-t}^{\V{x}})
g(\V{B}_t^{\V{x}})\wt{K}_{\V{x},t}[\cO]\Big]
\\
&=\lim_{t\to\infty}\int_{\Omega_{\Wi}}{K}_{\V{x},t}[\cO]\Id\mu_{\V{x},t}.
\end{align*}
Now the assertion follows from Theorem~\ref{thmlocconvmu} and the postulated uniform
convergence of ${K}_{\V{x},t}[\cO]$. In fact, given $\ve>0$, we pick $s>0$ so large that
$|{K}_{\V{x},t}[\cO]-{K}_{\V{x},r}[\cO]|<\ve$, for all $s\le r,t\le\infty$. Then
\begin{align*}
&\bigg|\int_{\Omega_{\Wi}}{K}_{\V{x},t}[\cO]\Id\mu_{\V{x},t}
-\int_{\Omega_{\Wi}}{K}_{\V{x},\infty}[\cO]\Id\mu_{\V{x},\infty}\bigg|
\\
&\le2\ve+\bigg|\int_{\Omega_{\Wi}}{K}_{\V{x},s}[\cO]\Id\mu_{\V{x},t}
-\int_{\Omega_{\Wi}}{K}_{\V{x},s}[\cO]\Id\mu_{\V{x},\infty}\bigg|,
\end{align*}
for all $t\ge s$. As $t\to\infty$, the term in the second line converges to $2\ve$ by
Theorem~\ref{thmlocconvmu} and the $\mr{\fF}_s$-measurability of $K_{\V{x},s}[\cO]$ and we conclude.
\end{proof}

When combined with the following lemma, the previous corollary can be used to study 
ground state expectations of certain {\em unbounded} observables, {\em without} a priori information
on whether $\Phi(\V{x})$ is in their domain or not. Here the crucial point is that the limiting measure 
$\mu_{\V{x},\infty}$ permits to construct holomorphic functions as the one called $g$ in the next 
lemma. The lemma is an abstracted version of an observation made in 
\cite[Thm.~10.12]{Hiroshima2004}. In its statement and proof $D_r(z)$ denotes the open disc of radius 
$r>0$ about $z\in\CC$ in the complex plane.

\begin{lem}\label{lemholext}
Let $A$ be a non-negative self-adjoint operator in some Hilbert space $\sK$ and let $\phi\in\sK$.
Suppose there exist $R>0$ and a holomorphic function $g:D_R(0)\to\CC$ such that
\begin{align}\label{gunnar}
\SPn{\phi}{e^{-zA}\phi}&=g(z),
\end{align}
for all $z\in D_R(0)$ with $\Re[z]>0$. Then $\phi\in\dom(e^{sA/2})$, for all $s\in(0,R)$, and 
\begin{align}\label{gunnar2}
\SPn{e^{-xA/2}\phi}{e^{-xA/2-iyA}\phi}&=g(z),\quad z\in D_R(0),\,x=\Re[z],\,y=\Im[z].
\end{align}
\end{lem}

\begin{proof}
Let $r\in(0,R)$. Then we find real numbers $0<\xi<\rho<R$ such that
$D_r(0)\subset D_\rho(\xi)\subset D_R(0)$. Let $g(z)=\sum_{n=0}^\infty b_n(z-\xi)^n$ denote the 
Taylor series of $g$ at $\xi$, whose radius of convergence is larger than $\rho$. Furthermore, let 
$\vs$ be the spectral measure of $A$ corresponding to $\phi$. Then Fubini's theorem implies 
\begin{align}\label{Gibbs78}
\SPn{\phi}{e^{-zA}\phi}=\sum_{n=0}^\infty(\xi-z)^n\frac{1}{n!}
\int_0^\infty\lambda^n e^{-\xi\lambda}\Id\vs(\lambda),\quad z\in D_\xi(\xi).
\end{align}
Comparing coefficients we infer from the validity of \eqref{gunnar} for $z\in D_\xi(\xi)$ that
$$
b_n=\frac{(-1)^n}{n!}\int_0^\infty\lambda^ne^{-\xi\lambda}\Id\vs(\lambda),\quad n\in\NN_0.
$$
Since the Taylor series of $g$ at $\xi$ converges absolutely on $D_\rho(\xi)$, we see {\em a posteriori} 
that the series on the right hand side of \eqref{Gibbs78} actually converges absolutely for all 
$z\in D_\rho(\xi)$. This permits to invoke Tonelli's theorem to argue that
\begin{align*}
\int_0^\infty e^{-x\lambda}\Id\vs(\lambda)&=\sum_{n=0}^\infty(\xi-x)^n\frac{1}{n!}
\int_0^\infty\lambda^ne^{-\xi\lambda}\Id\vs(\lambda)<\infty,\quad -r<x<0.
\end{align*}
Since $r\in(0,R)$ was arbitrary, this shows that $\phi\in\dom(e^{sA/2})$, for all $s\in(0,R)$.
The spectral calculus now implies that 
$D_R(0)\ni z\mapsto\int_0^\infty e^{-z\lambda}\Id\vs(\lambda)$
is holomorphic and equal to the function defined by the left hand side of \eqref{gunnar2}.
The identity theorem for holomorphic functions finally entails the equality in \eqref{gunnar2}.
\end{proof}

%%%%%%%%%%%%%%%%%%%%%%%%%%%%%%%%%%%%%%%%%%%%%%%
%%%%%%%%%%%%%%%%%%%%%%%%%%%%%%%%%%%%%%%%%%%%%%%

\subsection{Super-Exponential Decay of Boson Numbers in Ground States}\label{ssecsuperexpN}

\noindent
As an application of Corollary~\ref{corGSEEPM} and the succeeding lemma
we prove Theorem~\ref{introthmsupexpNum} at the end of this subsection. For $0<L<\infty$, we shall 
consider the second quantization of the characteristic function 
\begin{align}\label{defchiL}
\chi_L&:=1_{\{|\V{m}|\le L\}}.
\end{align}
For all $t\ge0$ and $z\in\CC$ with $\Re[z]\ge0$, we infer from \eqref{SGdGamma} 
and \eqref{defwtKO} that
\begin{align}\label{carla1}
\wt{K}_{\V{x},t}[e^{-z\Id\Gamma(\chi_L)}]&=e^{\SPn{U_{-t}^-}{(e^{-z\chi_L}-1)U_{t}^-}}.
\end{align}
A direct consequence of \cite[Lem.~3.11]{MatteMoeller2017} is that, $\PP_{\Wi}$-a.s.,
\begin{align}\label{forUminus}
U_{t}^-&=\int_0^te^{-s\omega-i\V{m}\cdot\V{B}_s}f_L\Id s
+(1-e^{-t\omega-i\V{m}\cdot\V{B}_t})\beta_{L,\infty}-M^-_{L,t},\quad t\ge0,
\end{align}
where
\begin{align*}
M_{L,t}^-&:=\int_0^te^{-s\omega-i\V{m}\cdot\V{B}_s}i\V{m}\beta_{L,\infty}\Id\V{B}_s,\quad t\ge0,
\end{align*}
is a square-integrable $L^2(\RR^3)$-valued martingale; see \cite[Lem.~3.10]{MatteMoeller2017}.
The process $U_{-t}^-$ is given by the same formulas with $\V{B}_-$ substituted for $\V{B}$. 
Notice that the last two members of the right hand side of \eqref{forUminus} vanish when they
are multiplied with $(e^{-z\chi_L}-1)$; recall \eqref{deffbeta}. This implies that the exponent on the right
hand side of \eqref{carla1} can $\PP_{\Wi}$-a.s. be written as
\begin{align*}
\SPn{{U}_{-t}^-}{(e^{-z\chi_L}-1){U}_{t}^-}&=(e^{-z}-1)w_{L,t},\quad t\ge0,
\end{align*}
where $w_{L,t}$ is the bounded, $\mr{\fF}_t$-measurable random variable given by
\begin{align*}
w_{L,t}:=\ee^2\int_0^t\int_0^t\int_{\{|\V{k}|\le L\}}
e^{-(r+s)\omega(\V{k})-i\V{k}\cdot(\V{B}_r-\V{B}_{-s})}
\frac{\eta(\V{k})^2}{\omega(\V{k})}\Id\V{k}\,\Id r\,\Id s.
\end{align*}

\begin{prop}\label{kettoren}
Let $0<L<\infty$. Then, as $t$ goes to infinity, $w_{L,t}$ converges uniformly on $\Omega_{\Wi}$ 
to some bounded, $\fB(\Omega_{\Wi})$-measurable random variable 
$w_{L,\infty}:\Omega_{\Wi}\to\CC$. For every $z\in\CC$, we further have 
$\Phi(\V{x})\in\dom(e^{z\Id\Gamma(\chi_L)})$ and 
\begin{align}\label{numren}
\SPn{\Phi(\V{x})}{e^{-z\Id\Gamma(\chi_L)}\Phi(\V{x})}&=
\|\Phi(\V{x})\|^2\int_{\Omega_{\Wi}}e^{(e^{-z}-1)w_{L,\infty}}\Id\mu_{\V{x},\infty}.
\end{align}
\end{prop}

\begin{proof} 
We start by observing that
\begin{align*}
|w_{L,t}-w_{L,t'}|&\le2\ee^2\int_{\{|\V{k}|\le L\}}e^{-t\omega(\V{k})}
\frac{\eta(\V{k})^2}{\omega(\V{k})^3}\Id\V{k},\quad t'>t>0.
\end{align*}
By virtue of the infrared condition this implies the first statement, which further reveals that the right 
hand side of \eqref{numren} defines an entire function of $z\in\CC$. For all $z\in\CC$ with
$\Re[z]\le0$, the identity \eqref{numren} is now a direct consequence of Corollary~\ref{corGSEEPM}.
By virtue of Lemma~\ref{lemholext} we finally extend \eqref{numren} to all $z\in\CC$.
\end{proof}

\begin{proof}[Proof of Theorem~\ref{introthmsupexpNum}.]
We choose $L=1$ in \eqref{defchiL} and set $\ol{\chi}_1:=1-\chi_1$. In view of 
Theorems~\ref{introthmexpdec} and~\ref{introthmsupexp} it suffices to treat the massless case
$\mu=0$, where $\ol{\chi}_1\le\omega\wedge1$. Let $r>0$. 
By virtue of a Cauchy-Schwarz inequality and Proposition~\ref{kettoren} we obtain
\begin{align*}
\|e^{r\Id\Gamma(1)}\Phi(\V{x})\|&\le\|e^{2r\Id\Gamma(\chi_1)}\Phi(\V{x})\|^\eh
\|e^{2r\Id\Gamma(\ol{\chi}_1)}\Phi(\V{x})\|^\eh
\\
&\le\|\Phi(\V{x})\|^\eh e^{(e^{4r}-1)\|w_{1,\infty}\|_\infty/4}
\|e^{2r\Id\Gamma(\omega\wedge1)}\Phi(\V{x})\|^\eh,
\end{align*}
and we conclude by applying the bound in 
Theorem~\ref{introthmexpdec} (resp. Theorem~\ref{introthmsupexp}).
\end{proof}

%%%%%%%%%%%%%%%%%%%%%%%%%%%%%%%%%%%%%%%%%%%%%%%
%%%%%%%%%%%%%%%%%%%%%%%%%%%%%%%%%%%%%%%%%%%%%%%

\subsection{On Gaussian Domination}\label{ssecGaussdom}

\noindent
Another application of Corollary~\ref{corGSEEPM} and Lemma~\ref{lemholext} is the following proof
of our last main theorem.

\begin{proof}[Proof of Theorem~\ref{introthmGaussdomlb}.]
Let $h\in L^2(\RR^3)$. Plugging the Weyl operator $\sW(-ish)=e^{-is\vp(h)}$ into 
\eqref{defwtKO} and taking \eqref{defWeyl} into account, we obtain
\begin{align*}
\wt{K}_{\V{x},t}[\sW(-ish)]&=e^{-s^2\|h\|^2/2+is\tilde{\alpha}_{\V{x},t}[h]},
\end{align*}
for all $t\ge0$ and $s\in\RR$, with
\begin{align*}
\tilde{\alpha}_{\V{x},t}[h]&:=\SPn{h}{e_{\V{x}}U_{t}^-}+\SPn{e_{\V{x}}U_{-t}^-}{h}.
\end{align*}
We now proceed in three steps. In the first step we prove \eqref{lbGauss} imposing a technical extra 
condition on $h$ that is dropped in the second one. In the third step we verify \eqref{notinGauss}.

{\em Step 1.} Assume in addition that that $\omega^{\epsilon} h\in L^2(\RR^3)$, for some 
$\epsilon>0$. Then we infer from \cite[Lem.~3.6]{MatteMoeller2017} that 
$\tilde{\alpha}_{\V{x}}[h]$ is indistinguishable from the $(\mr{\fF}_t)_{t\ge0}$-adapted
complex-valued integral process
\begin{align}\nonumber
\alpha_{\V{x},t}[h]&:=\ee\int_0^t\int_{\RR^3}e^{-s\omega(\V{k})-i\V{k}\cdot\V{B}_s^{\V{x}}}
\ol{h(\V{k})}\omega(\V{k})^\mh\eta(\V{k})\Id\V{k}\,\Id s
\\\label{claudi2000}
&\quad+\ee\int_0^t\int_{\RR^3}e^{-s\omega(\V{k})+i\V{k}\cdot\V{B}_{-s}^{\V{x}}}h(\V{k})
\omega(\V{k})^\mh\eta(\V{k})\Id\V{k}\,\Id s,\quad t\ge0.
\end{align}
In fact, thanks to the infrared condition $\omega^{-\nf{3}{2}}\eta\in L^2_\loc(\RR^3)$ and the
additional condition on $h$, both double Lebesgue integrals in the previous formula exist also for 
$t=\infty$ and define a bounded $\fB(\Omega_{\Wi})$-measurable function 
$\alpha_{\V{x},\infty}[h]:\Omega_{\Wi}\to\CC$.
Furthermore, $\alpha_{\V{x},t}[h]\to\alpha_{\V{x},\infty}[h]$ uniformly on $\Omega_{\Wi}$, 
as $t\to\infty$. Hence, Corollary~\ref{corGSEEPM} applies and yields
\begin{align*}
\SPn{\Phi(\V{x})}{e^{-is\vp(h)}\Phi(\V{x})}
&=\|\Phi(\V{x})\|^2\int_{\Omega_{\Wi}}e^{-s^2\|h\|^2/2+is\alpha_{\V{x},\infty}[h]}
\Id\mu_{\V{x},\infty},
\end{align*}
for all $s\in\RR$. Upon integrating the above identity over $\RR$ with respect to the measure
$(2\pi)^{\mh}e^{-s^2/2}\Id s$ and employing the spectral calculus and Fubini's theorem, 
we deduce that
\begin{align}\nonumber
\SPn{\Phi(\V{x})}{&e^{-z\vp(h)^2/2}\Phi(\V{x})}
\\\label{frieda1}
&=\frac{\|\Phi(\V{x})\|^2}{\sqrt{1+z\|h\|^2}}
\int_{\Omega_{\Wi}}e^{-z\alpha_{\V{x},\infty}[h]^2/2(1+z\|h\|^2)}\Id\mu_{\V{x},\infty}.
\end{align}
Here we also put $z^\eh h$ in place of $h$, which is allowed for all $z\in[0,\infty)$. 
Since $\alpha_{\V{x},\infty}[h]$ is a bounded random variable, 
the right hand side of \eqref{frieda1} is, however,
well-defined and analytic as a function of $z$ on the disc  $D_h:=\{z\in\CC:|z|<1/\|h\|^2\}$. 
(Here we choose the branch of the complex square root slit on the negative half-axis in the 
denominator on the right hand side.) By the spectral calculus, the left hand side of \eqref{frieda1}
is well-defined and analytic on $\{z\in\CC:\Re[z]>0\}$. By the identity theorem for holomorphic 
functions the left and right hand sides of \eqref{frieda1} agree on $\{z\in D_h:\Re[z]>0\}$.
Employing Lemma~\ref{lemholext} we conclude that $\Phi(\V{x})\in\dom(e^{r\vp(h)^2/4})$,
for every $r\in(0,1/\|h\|^2)$, and 
\begin{align}\nonumber
\SPn{e^{-x\vp(h)^2/4}\Phi(\V{x})}{&e^{-x\vp(h)^2/4-iy\vp(h)^2/2}\Phi(\V{x})}
\\
&=\frac{\|\Phi(\V{x})\|^2}{\sqrt{1+z\|h\|^2}}
\int_{\Omega_{\Wi}}e^{-z\alpha_{\V{x},\infty}[h]^2/2(1+z\|h\|^2)}\Id\mu_{\V{x},\infty},
\end{align}
for all $z\in D_h$ with $x:=\Re[z]$ and $y:=\Im[z]$.

Now assume that $h\in\mathfrak{r}$, where $\mathfrak{r}$ is the completely real subspace
of $L^2(\RR^3)$ defined in \eqref{realsubspace}. Then $\alpha_{\V{x},\infty}[h]$ is a
{\em real-valued} random variable: For taking the complex conjugate of $\alpha_{\V{x},\infty}[h]$
leads to the same result as substituting $\V{k}\to-\V{k}$ in the two integrals on the right hand side
of \eqref{claudi2000}. Also assuming $\|h\|<1/2$ 
and choosing $z=-4$, we see that $e^{2\alpha_{\V{x},\infty}[h]^2/(1-4\|h\|^2)}\ge1$ and
obtain the desired lower bound \eqref{lbGauss} under
the extra condition $\omega^{\epsilon} h\in L^2(\RR^3)$.

{\em Step~2.} Let $h\in\mathfrak{r}$ satisfy $\|h\|<1/2$ but otherwise be arbitrary. 
Pick some $\epsilon>0$ and $h_n\in\mathfrak{r}$
satisfying $\|h_n\|\le\|h\|$, $\omega^\epsilon  h_n\in L^2(\RR^3)$, $n\in\NN$, 
and $h_n\to h$, $n\to\infty$. (E.g., $h_n=1_{\{\omega<n\}}h$.)
Pick some $\alpha>1$ such that $4\alpha^2\|h\|^2<1$. 
Define $\Theta(t):=c(t/R)e^{t^2}$, $t\in\RR$, for some $R>0$ to be fixed later on, where
$c:\RR\to[0,1]$ is continuous with $\supp(c)\subset[-2,2]$ and $c=1$ on $[-1,1]$.
If $g\in\{h\}\cup\{h_n|\,n\in\NN\}$, then the spectral calculus and 
Lemma~\ref{leminvGauss} entail
\begin{align*}
\|e^{\vp(g)^2}\Phi(\V{x})-\Theta(\vp(g))\Phi(\V{x})\|
&\le e^{-(\alpha-1)R^2}\|e^{\alpha\vp(g)^2}\Phi(\V{x})\|
\\
&\le\frac{e^{-(\alpha-1)R^2}}{\sqrt{1-4\alpha^2\|h\|^2}}\|e^{\Id\Gamma(1)/(\alpha-1)}\Phi(\V{x})\|,
\end{align*}
where the term in the second line is well-defined by Theorem~\ref{introthmsupexpNum}.
Here we also used the condition $\|h_n\|\le\|h\|$ in the last step. Let $\delta>0$.
Then the previous bound permits to fix a sufficiently large $R>0$ such that
\begin{align*}
\|e^{\vp(g)^2}\Phi(\V{x})-\Theta(\vp(g))\Phi(\V{x})\|<\frac{\delta}{3},\quad 
g\in\{h\}\cup\{h_n|\,n\in\NN\}.
\end{align*}
Since $h_n\to h$, $n\to\infty$, in $L^2(\RR^3)$, the strong continuity of the Weyl representation
further implies that $\vp(h_n)\to\vp(h)$ in the strong resolvent sense, whence
$\Theta(\vp(h_n))\to\Theta(\vp(h))$ strongly; see, e.g., \cite[Thm.~VIII.20(b)]{ReedSimonI}.
Thus, we find some $m\in\NN$ such that 
$\|\Theta(\vp(h_m))\Phi(\V{x})-\Theta(\vp(h))\Phi(\V{x})\|<\delta/3$. 
Combining these remarks and applying \eqref{lbGauss} to $h_m$ we find
\begin{align*}
\|e^{\vp(h)^2}\Phi(\V{x})\|&>\|e^{\vp(h_m)^2}\Phi(\V{x})\|-\delta
\ge\frac{\|\Phi(\V{x})\|}{\sqrt[4]{1-4\|h_m\|^2}}-\delta
\ge\frac{\|\Phi(\V{x})\|}{\sqrt[4]{1-4\|h\|^2}}-\delta.
\end{align*}
Since $\delta>0$ was arbitrary, this concludes the proof of \eqref{lbGauss}.

{\em Step 3.} Let $h\in\mathfrak{r}$ satisfy $\|h\|\ge1/2$. Then we find some $\delta\in(0,1]$
such that $4\delta\|h\|^2=1$. We further pick $\delta_n\in(0,1)$ such that $\delta_n\uparrow\delta$,
as $n\to\infty$. Let $\nu_{\Phi(\V{x})}$ be the spectral measure of $\vp(h)$ associated
with $\Phi(\V{x})$. Then the monotone convergence theorem and \eqref{lbGauss} imply
\begin{align*}
\int_\RR e^{2t^2}\Id\nu_{\Phi(\V{x})}&\ge\int_\RR e^{2\delta t^2}\Id\nu_{\Phi(\V{x})}
\\
&=\lim_{n\to\infty}\int_\RR e^{2\delta_n t^2}\Id\nu_{\Phi(\V{x})}
\ge\lim_{n\to\infty}\frac{\|\Phi(\V{x})\|^2}{\sqrt{1-4\delta_n\|h\|^2}}=\infty,
\end{align*}
which proves \eqref{notinGauss}
\end{proof}

%%%%%%%%%%%%%%%%%%%%%%%%%%%%%%%%%%%%%%%%%%%%%%%
%%%%%%%%%%%%%%%%%%%%%%%%%%%%%%%%%%%%%%%%%%%%%%%
%%%%%%%%%%%%%%%%%%%%%%%%%%%%%%%%%%%%%%%%%%%%%%%

\appendix

%%%%%%%%%%%%%%%%%%%%%%%%%%%%%%%%%%%%%%%%%%%%%%%
%%%%%%%%%%%%%%%%%%%%%%%%%%%%%%%%%%%%%%%%%%%%%%%
%%%%%%%%%%%%%%%%%%%%%%%%%%%%%%%%%%%%%%%%%%%%%%%

\section{Some Properties of Weyl Operators}

\noindent
Here we collect some technical results on the Weyl representation that are used in 
Section~\ref{secabsence} and in the succeeding Appendix~\ref{appGrosstrafo}.
The reader should keep in mind that $\sW(g)^*=\sW(-g)$, $g\in L^2(\RR^3)$, in what follows.

\begin{lem}\label{lemWvp}
Let $f,g\in L^2(\RR^3)$. Then $\sW(\pm g)\dom(\vp(f))=\dom(\vp(f))$ and
\begin{align}\label{Wvp}
\sW(g)\vp(f)\sW(g)^*&=\vp(f)-2\Re\SPn{f}{g}.
\end{align}
\end{lem}

\begin{proof}
By the spectral calculus and the Weyl relations \eqref{Weylrel} both sides of \eqref{Wvp}
generate the same strongly continuous unitary group.
\end{proof}

\begin{lem}\label{lemWdG}
Let $\vk$ be a maximal, non-negative, and invertible multiplication operator in $L^2(\RR^3)$ 
and let $g\in\dom(\vk)$. Then $\vp(\vk g)$ is infinitesimally form bounded with respect to
$\Id\Gamma(\vk)$, $\sW(\pm g)\fdom(\Id\Gamma(\vk))=\fdom(\Id\Gamma(\vk))$ and, 
for all $\psi\in\fdom(\Id\Gamma(\vk))$,
\begin{align}\label{WdGammaform}
\|\Id\Gamma(\vk)^\eh\sW(g)^*\psi\|^2=\|\Id\Gamma(\vk)^\eh\psi\|^2-\SPn{\psi}{\vp(\vk g)\psi}
+\SPn{g}{\vk g}\|\psi\|^2.
\end{align}
\end{lem}

\begin{proof}
If we consider only $\psi\in\dom(\Id\Gamma(\vk))$ in \eqref{WdGammaform}, then detailed proofs 
of all statements can be found, e.g., in \cite[Lem.~C.3 and Lem.~C.4]{GriesemerWuensch2017}.
Since $\dom(\Id\Gamma(\vk))$ is a form core for $\Id\Gamma(\omega)$ and 
$\Id\Gamma(\vk)^\eh\sW(g)^*$ is closed, it is, however, clear that \eqref{WdGammaform} 
extends to all $\psi\in\fdom(\Id\Gamma(\vk))$.
\end{proof}

The next lemma summarizes results from \cite[Lem.~C.4 and Cor.~C.5]{GriesemerWuensch2017}: 

\begin{lem}\label{lemGrWubd}
Let $g\in\fdom(\omega)$. Then $\sW(g)\fdom(\Id\Gamma(\omega))=\fdom(\Id\Gamma(\omega))$ and
\begin{align}\label{GrWubd}
\|(1+\Id\Gamma(\omega))^\eh\sW(g)(1+\Id\Gamma(\omega))^\mh\|&\le1+\|\omega^\eh g\|.
\end{align}
\end{lem}

The previous two lemmas permit to complement the strong continuity of the Weyl
representation by the following result:

\begin{lem}\label{lemstrcontWeyldG}
Let $g,g_n\in\fdom(\omega)$, $n\in\NN$, such that $f_n\to f$ and $\omega^\eh f_n\to\omega^\eh f$
in $L^2(\RR^3)$, as $n\to\infty$. Then
\begin{align*}
\Id\Gamma(\omega)^\eh(\sW(g_n)-\sW(g))(1+\Id\Gamma(\omega))^\mh\psi
&\xrightarrow{\;\;n\to\infty\;\;}0,\quad\psi\in\sF.
\end{align*}
\end{lem}

\begin{proof}
Set $h_n:=g_n-g$, $n\in\NN$. In view of the Weyl relations and \eqref{GrWubd} it suffices to show
that $\Id\Gamma(\omega)^\eh(\sW(h_n)-\id)\phi\to0$, $n\to\infty$, for every 
$\phi\in\dom(\Id\Gamma(\omega))$. So let $\phi$ be in the domain of $\Id\Gamma(\omega)$.
Applying Lemma~\ref{lemWdG} to every $\vk=\omega_m:=m\wedge\omega$ with $m\in\NN$, we obtain
\begin{align*}
\|\Id\Gamma(\omega_m)^\eh\sW(h_n)\phi-\Id\Gamma(\omega_m)^\eh\phi\|^2
&=2\|\Id\Gamma(\omega_m)^\eh\phi\|^2+2\Re\SPn{\phi}{a(\omega_mh_n)\phi}
\\
&\quad-2\Re\SPn{\Id\Gamma(\omega_m)\phi}{\sW(h_n)\phi}
+\|\omega_m^\eh h_n\|^2\|\phi\|^2
\\
&\le2\|\Id\Gamma(\omega)^\eh\phi\|^2+2\|\phi\|\|\omega^\eh h_n\|\|\Id\Gamma(\omega)^\eh\phi\|
\\
&\quad-2\Re\SPn{\Id\Gamma(\omega_m)\phi}{\sW(h_n)\phi}
+\|\omega^\eh h_n\|^2\|\phi\|^2.
\end{align*}
Passing to the limit $m\to\infty$ in the terms in the first and last lines we find
\begin{align*}
\|\Id\Gamma(\omega)^\eh\sW(h_n)\phi-\Id\Gamma(\omega)^\eh\phi\|^2
&\le2\|\Id\Gamma(\omega)^\eh\phi\|^2+2\|\phi\|\|\omega^\eh h_n\|\|\Id\Gamma(\omega)^\eh\phi\|
\\
&\quad-2\Re\SPn{\Id\Gamma(\omega)\phi}{\sW(h_n)\phi}
+\|\omega^\eh h_n\|^2\|\phi\|^2.
\end{align*}
Since $\sW(h_n)\to\id$ strongly and $\|\omega^\eh h_n\|\to0$, the right hand side of the previous
inequality goes to zero as $n\to\infty$.
\end{proof}

%%%%%%%%%%%%%%%%%%%%%%%%%%%%%%%%%%%%%%%%%%%%%%%
%%%%%%%%%%%%%%%%%%%%%%%%%%%%%%%%%%%%%%%%%%%%%%%
%%%%%%%%%%%%%%%%%%%%%%%%%%%%%%%%%%%%%%%%%%%%%%%

\section{Gross Transformation of the Nelson Hamiltonian with Ultraviolet Cutoff}\label{appGrosstrafo}

\noindent
In this appendix we verify the assertions on the Gross transformed Nelson form stated in
Proposition~\ref{propGrosstrafo}. Before we prove the latter proposition, we present a lemma which,
together with Lemma~\ref{lemWvp} and Lemma~\ref{lemWdG},
explains how the various terms in the Nelson form $\nv_{\Geb,\UV}$ 
transform under $G_{K,\UV}$. Recall that $G_{K,\UV}$ is defined in \eqref{defGrosstrafo}.

\begin{lem}\label{lemnablaGrossPsi}
Let $0\le K<\UV<\infty$ be such that $\beta_{K,\UV}\in L^2(\RR^3)$ and let $\Psi\in\QGV$.
Then $G_{K,\UV}^*\Psi\in\dom(\mathfrak{t}_{\Geb}^+)$ and, for a.e. $\V{x}\in\RR^3$,
\begin{align}\label{nablaGrossPsi}
\nabla(G_{K,\UV}^*\Psi)(\V{x})&=\sW(-e_{\V{x}}\beta_{K,\UV})
\big(\nabla\Psi(\V{x})+i\vp(e_{\V{x}}\V{m}\beta_{K,\UV})\Psi(\V{x})\big).
\end{align}
\end{lem}

\begin{proof}
Thanks to the sharp ultraviolet cutoff at $\UV<\infty$, the map
$\RR^3\ni\V{x}\mapsto e_{\V{x}}\beta_{K,\UV}\in L^2(\RR^3)$ is smooth.
We also recall that $L^2(\RR^3)\ni h\mapsto\epsilon(h)\in\sF$ is analytic.
So, if $\Psi$ is equal to $g\epsilon(h)$ with $g\in C_0^\infty(\Geb)$ and $h\in\dom(\omega)$, then
$G_{K,\Lambda}^*\Psi$ is manifestly smooth in view of \eqref{defWeyl} and, furthermore,
\eqref{nablaGrossPsi} is a consequence of \eqref{karl}, \eqref{Wvp}, and
$\SPn{e_{\V{x}}\beta_{K,\UV}}{\V{m}e_{\V{x}}\beta_{K,\UV}}=0$. 
For general $\Psi\in\QGV$, we can employ \cite[Cor.~4.6]{Matte2017} according to which we find 
$\Psi_n\in\mathrm{span}\{g\epsilon(h)|\,g\in C_0^\infty(\Geb),\,h\in\dom(\omega)\}$,
$n\in\NN$, such that $\cfG[\Psi_n-\Psi]\to0$, as $n\to\infty$. 
Then $G_{K,\UV}^*\Psi_n\in C_0^\infty(\Geb,\sF)$, for every $n\in\NN$. Plugging
$\Psi_n$ into \eqref{nablaGrossPsi} and using \eqref{rbvp} we see that
$\nabla(G_{K,\UV}^*\Psi_n)$ converges in $L^2(\Geb,\sF^3)$ to 
the right hand side of \eqref{nablaGrossPsi}, which therefore equals the
weak gradient of $G_{K,\UV}^*\Psi$. Finally, it is clear that
\begin{align*}
\int_{\Geb}V_+(\V{x})\|G_{K,\UV}^*(\Psi_n-\Psi)(\V{x})\|^2\Id\V{x}\le
\cfG[\Psi_n-\Psi]\xrightarrow{\;\;n\to\infty\;\;}0,
\end{align*}
so that $G_{K,\UV}^*\Psi\in\dom(\mathfrak{t}_{\Geb}^+)$; recall \eqref{deftGplus}.
\end{proof}

\begin{proof}[Proof of Proposition~\ref{propGrosstrafo}.]
Combining Lemma~\ref{lemGrWubd} and Lemma~\ref{lemnablaGrossPsi} (which also holds for
the coupling constant $-\ee$), we first see that $G_{K,\UV}$ and $G_{K,\UV}^*$ map
$\QGV$ into itself. The identity \eqref{trafoNelsonform} is then a direct consequence of
\eqref{Wvp}, \eqref{WdGammaform}, \eqref{nablaGrossPsi}, and the relations
\begin{align*}
\vp(e_{\V{x}}f_{\UV})\psi-\vp(e_{\V{x}}\omega\beta_{K,\UV})\psi
&=\vp(e_{\V{x}}f_{K})\psi+\frac{1}{2}\vp(e_{\V{x}}\V{m}^2\beta_{K,\UV})\psi,
\\
\Re\SPn{e_{\V{x}}f_{\UV}}{e_{\V{x}}\beta_{K,\UV}}&=E_{\UV}^{\ren}-E_{K}^{\ren},
\\
\SPn{e_{\V{x}}\beta_{K,\UV}}{e_{\V{x}}\omega\beta_{K,\UV}}
-(E_{\UV}^{\ren}-E_{K}^{\ren})&=-\frac{1}{2}\|\V{m}\beta_{K,\UV}\|^2,
\end{align*}
valid for all $\V{x}\in\RR^3$ and $\psi\in\fdom(\Id\Gamma(\omega))$ in the first line.
\end{proof}

%%%%%%%%%%%%%%%%%%%%%%%%%%%%%%%%%%%%%%%%%%%%%%%
%%%%%%%%%%%%%%%%%%%%%%%%%%%%%%%%%%%%%%%%%%%%%%%
%%%%%%%%%%%%%%%%%%%%%%%%%%%%%%%%%%%%%%%%%%%%%%%

\section{Relative Bounds Needed to Remove the Ultraviolet Cutoff}\label{apprbUV}

\noindent
In this appendix we derive the relative bounds employed in the main text to construct
renormalized operators and to study the infrared behavior of ground state eigenvectors.
In particular, we shall prove Proposition~\ref{proprbUV1} and Proposition~\ref{proprbUV2}.
As mentioned earlier, the estimations below are simple modifications of the ones used by Nelson
in \cite{Nelson1964}, which are re-obtained in essence by setting $K=L$. We also implement
later extensions to the case $\mu=0$ of \cite{GriesemerWuensch2017,HHS2005}, where
bounds similar to \eqref{myresluger1} and \eqref{myresluger2} have been applied as well.

To start with we verify the basic relation splitting $\hv_{\Geb,K,\UV}$ into a $\UV$-independent, 
well-understood comparison term and a $\UV$-dependent perturbation.

\begin{lem}\label{lemLsplit}
Let $0\le K\le L\le \UV<\infty$. Then \eqref{defqKLV} holds true for all $\Psi\in{\QGV}$.
\end{lem}

\begin{proof}
For every $\Psi\in\QGV$, we find 
$\Psi_n\in\mathrm{span}\{g\epsilon(h)|\,g\in C_0^\infty(\Geb),\,h\in\dom(\omega)\}$, $n\in\NN$,
such that $\Psi_n\to\Psi$ in $L^2(\Geb,\sF)$ and $\cfG[\Psi_n-\Psi]\to0$, as $n\to\infty$;
see, e.g., \cite[Cor.~4.6]{Matte2017}. In view of the relative bounds \eqref{rba}--\eqref{rbvp} it
therefore suffices to verify \eqref{defqKLV} for every 
$\Psi\in\mathrm{span}\{g\epsilon(h)|\,g\in C_0^\infty(\Geb),\,h\in\dom(\omega)\}$.
So let $\Psi$ be in the latter space in what follows and let $\V{x}\in\Geb$. Then the
definition \eqref{deftLambdax} entails
\begin{align}\nonumber
\mathfrak{t}_{K,\UV,\V{x}}[\Psi]-\mathfrak{t}_{K,L,\V{x}}[\Psi]
&=\Re\SPn{i\vp(e_{\V{x}}\V{m}\beta_{L,\UV})\Psi(\V{x})}{\nabla\Psi(\V{x})
+i\vp(e_{\V{x}}\V{m}\beta_{K,L})\Psi(\V{x})}
\\\label{split77}
&\quad+\frac{1}{2}\|\vp(e_{\V{x}}\V{m}\beta_{L,\UV})\Psi(\V{x})\|^2.
\end{align}
After normal ordering the term in the second line of the previous identity reads
\begin{align*}
\frac{1}{2}\|\vp(e_{\V{x}}\V{m}\beta_{L,\UV})\Psi(\V{x})\|^2
&=
\Re\SPn{\ad(e_{\V{x}}\V{m}\beta_{L,\UV})\Psi(\V{x})}{a(e_{\V{x}}\V{m}\beta_{L,\UV})\Psi(\V{x})}
\\
&\quad
+\|a(e_{\V{x}}\V{m}\beta_{L,\UV})\Psi(\V{x})\|^2
+\frac{1}{2}\|\V{m}\beta_{L,\UV}\|^2\|\Psi(\V{x})\|^2.
\end{align*}
In the first term on the right hand side of \eqref{split77}
we split up the field operator in the left entry of the scalar product as 
$\vp(e_{\V{x}}\V{m}\beta_{L,\UV})=\ad(e_{\V{x}}\V{m}\beta_{L,\UV})+a(e_{\V{x}}\V{m}\beta_{L,\UV})$ 
and re-write the contribution of the creation operator as
\begin{align}\nonumber
\Re&\SPn{i\ad(e_{\V{x}}\V{m}\beta_{L,\UV})\Psi(\V{x})}{\nabla\Psi(\V{x})
+i\vp(e_{\V{x}}\V{m}\beta_{K,L})\Psi(\V{x})}
\\\nonumber
&=\Re\SPn{\nabla\Psi(\V{x})
+i\vp(e_{\V{x}}\V{m}\beta_{K,L})\Psi(\V{x})}{ia(e_{\V{x}}\V{m}\beta_{L,\UV})\Psi(\V{x})}
\\\label{split78}
&\quad-\Re\SPn{\ad(e_{\V{x}}\V{m}^2\beta_{L,\UV})\Psi(\V{x})}{\Psi(\V{x})}
+\nabla\Re\SPn{ia(e_{\V{x}}\V{m}\beta_{L,\UV})\Psi(\V{x})}{\Psi(\V{x})}.
\end{align}
Here we used a Leibniz rule and exploited that $\beta_{K,L}$ and
$\beta_{L,\UV}$ have disjoint supports up to a zero set, so that $\vp(e_{\V{x}}\V{m}\beta_{K,L})$
and $\ad(e_{\V{x}}\V{m}\beta_{K,L})$ commute when they are applied to $\Psi(\V{x})$.
After an integration with respect to $\V{x}\in\Geb$, the term $\nabla\Re\ldots$
in the last line of \eqref{split78} vanishes and 
a few further easy manipulations finish the proof of \eqref{defqKLV}.
\end{proof}

The next lemma provides the key estimate in Nelson's renormalization strategy in a variant
suitable for arbitrary choices of the boson mass $\mu\ge0$.

\begin{lem}\label{lemtertius}
Let $\beta,\beta':\RR^3\to\RR$ be measurable such that $|\V{m}|^\eh\beta$ and $|\V{m}|\beta$
are square-integrable and similarly for $\beta'$. Then 
\begin{align}\nonumber
|\SPn{\ad(\V{m}\beta)\psi}{a(\V{m}{\beta}')\psi}|
&\le6\|(|\V{m}|^\eh\vee|\V{m}|^{\nf{3}{4}})\beta\|
\|(|\V{m}|^\eh\vee|\V{m}|^{\nf{3}{4}}){\beta}'\|
\\\label{tertius}
&\quad\,\cdot\|(1+\Id\Gamma(\omega))^\eh\psi\|\|\Id\Gamma(\omega)^\eh\psi\|,\quad\psi\in\fdom(\Id\Gamma(\omega)).
\end{align}
\end{lem}

\begin{proof}
We split the function in the creation operator in an infrared part,
$\beta_{<}:=1_{\{|\V{m}|\le1\}}\beta$, and an ultraviolet part, $\beta_{>}:=1_{\{|\V{m}|>1\}}\beta$.
The infrared part is dealt with by the standard bounds \eqref{rba} and \eqref{rbad} which imply
\begin{align*}
&|\SPn{\ad(\V{m}\beta_<)\psi}{a(\V{m}\beta')\psi}|
\\
&\le2\|(\omega^\mh\vee1)\V{m}\beta_{<}\|\|(1+\Id\Gamma(\omega))^\eh\psi\|
\|\omega^\mh\V{m}\beta'\|\|\Id\Gamma(\omega)^\eh\psi\|,
\end{align*} 
where $\psi\in\fdom(\Id\Gamma(\omega))$. Define $\beta_<'$ and $\beta_>'$ in the same way
as $\beta_<$ and $\beta_>$.
Then the canonical commutation relations and an analogous estimate further yield
\begin{align}\nonumber
&|\SPn{\ad(\V{m}\beta_{>})\phi}{a(\V{m}\beta_{<}')\phi}|
\\\nonumber
&=|\SPn{\ad(\V{m}\beta_{<}')\phi}{a(\V{m}\beta_{>})\phi}|
\\\label{tertius2}
&\le2\|(\omega^\mh\vee1)\V{m}\beta_{<}'\|\|(1+\Id\Gamma(\omega))^\eh\phi\|
\|\omega^\mh\V{m}\beta_{>}\|\|\Id\Gamma(\omega)^\eh\phi\|,
\end{align}
for all $\phi\in\dom(\Id\Gamma(\omega))$.
For the remaining part we shall employ the following bound, valid for any
$f,g\in\dom(\omega^{-\nf{1}{4}})$ and
$\phi\in\dom(\Id\Gamma(1_S\omega))$, where $S\subset\RR^3$ is a measurable set such that
$f=g=0$ a.e. on the complement of $S$,
\begin{align}\nonumber
\big\|(2+\Id\Gamma&(\omega\wedge1))^\mh a(f)a(g)\phi\big\|
\\\label{myresluger1}
&\le\|\omega^{-\nf{1}{4}}f\|\|\omega^{-\nf{1}{4}}g\|
\sup_{\ve>0}\|(\ve+\Id\Gamma(\omega\wedge1))^\mh
\Id\Gamma(1_S\omega^\eh)\phi\|.
\end{align}
It is obtained upon successively multiplying the inequalities
\begin{align*}
&\frac{\left|
\int_{\RR^3}\int_{\RR^3}\bar{f}(\V{k}_{1})\bar{g}(\V{k}_{2})\phi^{(n+2)}
(\V{k}_1,\ldots,\V{k}_{n+2})\Id\V{k}_{1}\Id\V{k}_{2}\right|^2}{
2+\ve+\sum_{j=3}^{n+2}(\omega\wedge1)(\V{k}_j)}
\\
&\le\|\omega^{-\nf{1}{4}}f\|^2\|\omega^{-\nf{1}{4}}g\|^2\int_{S}\int_{S}
\frac{\omega(\V{k}_1)^\eh\omega(\V{k}_2)^\eh|\phi^{(n+2)}(\V{k}_1,\ldots,\V{k}_{n+2})|^2}{
\ve+\sum_{j=1}^{n+2}(\omega\wedge1)(\V{k}_j)}\Id\V{k}_{1}\Id\V{k}_{2}
\end{align*}
with $(n+2)(n+1)$, integrating over $(\RR^3)^n$, summing with respect to $n\in\NN_0$, 
and exploiting the permutation symmetry of $\phi^{(n+2)}$, 
where $\phi=(\phi^{(n)})_{n=1}^\infty\in\sF$.
We apply \eqref{myresluger1} with $S=\{|\V{m}|\ge1\}$ and combine it with the elementary estimates
\begin{align}\nonumber
\sum_{j=1}^n1_{\{|\V{m}|\ge1\}}(\V{k}_j)\omega^\eh(\V{k}_j)
&\le\Big(\sum_{j=1}^n1_{\{|\V{m}|\ge1\}}(\V{k}_j)\Big)^\eh
\Big(\sum_{j=1}^n\omega(\V{k}_j)\Big)^\eh
\\\label{myresluger2}
&\le\Big(\sum_{j=1}^n(\omega\wedge1)(\V{k}_j)\Big)^\eh
\Big(\sum_{j=1}^n\omega(\V{k}_j)\Big)^\eh,
\end{align}
for all $\V{k}_1,\ldots,\V{k}_n\in\RR^3$ and $n\in\NN$. This leads to
\begin{align}\nonumber
&|\SPn{\ad(\V{m}\beta_{>})\phi}{a(\V{m}\beta_{>}')\phi}|
\\\nonumber
&\le\|(2+\Id\Gamma(\omega\wedge1))^\eh\phi\|
\big\|(2+\Id\Gamma(\omega\wedge1))^\mh a(\V{m}\beta_{>})a(\V{m}\beta_{>}')\phi\big\|
\\\label{tertius3}
&\le\|\omega^{-\nf{1}{4}}\V{m}\beta_{>}\|\|\omega^{-\nf{1}{4}}\V{m}\beta_{>}'\|
\|(2+\Id\Gamma(\omega))^\eh\phi\|\|\Id\Gamma(\omega)^\eh\phi\|,
\end{align}
for all $\phi\in\dom(\Id\Gamma(\omega))$. We finally observe that
the terms in the first and last lines of \eqref{tertius2} and \eqref{tertius3} are well-defined
and continuous on $\fdom(\Id\Gamma(\omega))$ as functions of $\phi$.
\end{proof}

\begin{lem}\label{lemvx1}
Let $0\le K\le L\le\UV<\infty$ and $\Psi\in\QGV$. Then the following bound holds, for
a.e. $\V{x}\in\RR^3$ and all $\ve>0$,
\begin{align*}
&|\mathfrak{v}_{K,L,\Lambda,\V{x}}[\Psi]|
\\
&\le\frac{\ve}{2}\|\nabla\Psi(\V{x})+i\vp(e_{\V{x}}\V{m}\beta_{K,L})\Psi(\V{x})\|^2
+\Big(1+\frac{2}{\ve}\Big)\||\V{m}|^\eh\beta_{L,\Lambda}\|^2
\|\Id\Gamma(\omega)^\eh\Psi(\V{x})\|^2
\\
&\quad+6\|(|\V{m}|^\eh\vee|\V{m}|^{\nf{3}{4}})\beta_{L,\Lambda}\|^2
\|(1+\Id\Gamma(\omega))^\eh\Psi(\V{x})\|\|\Id\Gamma(\omega)^\eh\Psi(\V{x})\|.
\end{align*}
\end{lem}

\begin{proof}
In view of \eqref{rba} it is clear how to estimate the first two terms on the right hand side of 
\eqref{defvKL}. The third term can be dealt with by Lemma~\ref{lemtertius} where we choose
$\beta:=\beta':=e_{\V{x}}\beta_{L,\Lambda}$.
\end{proof}

\begin{lem}\label{lemvx2}
Let $0\le K\le L\le\Lambda<\Lambda'<\infty$ and $\Psi\in\QGV$. Then, for a.e. $\V{x}\in\Geb$,
\begin{align*}
&\big|\mathfrak{v}_{K,L,\Lambda,\V{x}}[\Psi]-\mathfrak{v}_{K,L,\Lambda',\V{x}}[\Psi]\big|
\\
&\le\|\nabla\Psi(\V{x})+i\vp(e_{\V{x}}\V{m}\beta_{K,L})\Psi(\V{x})\|
\||\V{m}|^\eh\beta_{\Lambda,\Lambda'}\|\|\Id\Gamma(\omega)^\eh\Psi(\V{x})\|
\\
&\quad+2\||\V{m}|^\eh\beta_{L,\infty}\|\||\V{m}|^\eh\beta_{\Lambda,\Lambda'}\|
\|\Id\Gamma(\omega)^\eh\Psi(\V{x})\|^2
\\
&\quad+12\|(|\V{m}|^\eh\vee|\V{m}|^{\nf{3}{4}})\beta_{L,\infty}\|
\|(|\V{m}|^\eh\vee|\V{m}|^{\nf{3}{4}})\beta_{\Lambda,\Lambda'}\|
\\
&\quad\quad\cdot\|(1+\Id\Gamma(\omega))^\eh\Psi(\V{x})\|\|\Id\Gamma(\omega)^\eh\Psi(\V{x})\|.
\end{align*}
\end{lem}

\begin{proof}
Thanks to \eqref{rba} it is again clear how to estimate the terms in the second and third lines of
\begin{align}\nonumber
&\mathfrak{v}_{K,L,\Lambda,\V{x}}[\Psi]-\mathfrak{v}_{K,L,\Lambda',\V{x}}[\Psi]
\\\nonumber
&=-\Re\Big[2\SPB{ia(e_{\V{x}}\V{m}\beta_{\Lambda,\Lambda'})\Psi(\V{x})}{\nabla\Psi(\V{x})
+i\vp(e_{\V{x}}\V{m}\beta_{K,L})\Psi(\V{x})}\Big]
\\\nonumber
&\quad+\|a(e_{\V{x}}\V{m}\beta_{L,\Lambda})\Psi(\V{x})\|^2
-\|a(e_{\V{x}}\V{m}\beta_{L,\Lambda'})\Psi(\V{x})\|^2
\\\nonumber
&\quad-\Re\Big[\SPn{\ad(e_{\V{x}}\V{m}\beta_{L,\Lambda})
\Psi(\V{x})}{a(e_{\V{x}}\V{m}\beta_{\Lambda,\Lambda'})\Psi(\V{x})}\Big]
\\\label{quartus}
&\quad-\Re\Big[\SPn{\ad(e_{\V{x}}\V{m}\beta_{\Lambda,\Lambda'})
\Psi(\V{x})}{a(e_{\V{x}}\V{m}\beta_{L,\Lambda'})\Psi(\V{x})}\Big],\quad\text{a.e. $\V{x}\in\Geb$.}
\end{align}
The bound \eqref{tertius} applies to the terms in the last two lines.
\end{proof}

\begin{proof}[Proof of Proposition~\ref{proprbUV1} and Proposition~\ref{proprbUV2}.]
Since the restriction of $V_-$ to $\Geb$ is infinitesimally form bounded with respect to the negative 
Dirichlet-Laplacian, a diamagnetic inequality (see, e.g., \cite{Matte2017}) implies
$$
\frac{1}{4}\int_\Geb\|\nabla\Psi(\V{x})+i\vp(e_{\V{x}}\V{m}\beta_{K,L})\Psi(\V{x})\|^2
\Id\V{x}-\int_{\Geb}V_-(\V{x})\|\Psi(\V{x})\|^2\Id\V{x}+c_{V_-}\|\Psi\|^2\ge0,
$$ 
for all $\Psi\in\QGV$ and some $c_{V_-}>0$. For all $\Psi\in L^2(\Geb,\fdom(\Id\Gamma(\omega)))$ 
and a.e. $\V{x}\in\Geb$, we further deduce from \eqref{rba} that
\begin{align*}
\frac{1}{2}\|\Id\Gamma(\omega)^\eh\Psi(\V{x})\|^2
+\frac{1}{2}\SPb{\Psi(\V{x})}{\vp(2e_{\V{x}}f_{K}+e_{\V{x}}\V{m}^2\beta_{K,\Lambda})\Psi(\V{x})}&
\\
+\Big(2\||\V{m}|^{\mh}f_{K}\|^2+\frac{1}{2}\||\V{m}|^{\nf{3}{2}}\beta_{K,L}\|^2\Big)\|\Psi(\V{x})\|^2&\ge0.
\end{align*}
Combining these remarks we find, for all $\Psi\in{\QGV}$,
\begin{align}\nonumber
&\frac{1}{2}\int_{\Geb}\Big(\frac{1}{2}\|\nabla\Psi(\V{x})+i\vp(e_{\V{x}}\V{m}\beta_{K,L})\Psi(\V{x})\|^2
+\|\Id\Gamma(\omega)^\eh\Psi(\V{x})\|^2\Big)\Id\V{x}
\\\label{sarah1}
%&\le\hv_{\Geb,K,L}[\Psi]+\Big(2\||\V{m}|^{\mh}f_{K}\|^2+\frac{1}{2}\||\V{m}|^{\nf{3}{2}}\beta_{K,L}\|^2+\frac{1}{2}\|\V{m}\beta_{K,L}\|^2+c_{V_-}\Big)\|\Psi\|^2.
&\le\hv_{\Geb,K,L}[\Psi]+\Big(2\||\V{m}|^{\mh}f_{K}\|^2+\|(|\V{m}|^{\nf{3}{2}}\vee|\V{m}|)\beta_{K,L}\|^2
+c_{V_-}\Big)\|\Psi\|^2.
\end{align}
Proposition~\ref{proprbUV1} and Proposition~\ref{proprbUV2} are now direct consequences of
\eqref{sarah1} together with Lemma~\ref{lemvx1} and Lemma~\ref{lemvx2}, respectively.
\end{proof}

The above relative bounds can also be used to study the dependence on $\ee\eta$:

\begin{lem}\label{lembdvmu}
In the situation of Lemma~\ref{lempert}, let $0\le\UV<\infty$ and define 
$\hat{\mathfrak{v}}_{\UV,\V{x}}$ in the same way as 
$\mathfrak{v}_{\UV,\V{x}}:=\mathfrak{v}_{0,0,\UV,\V{x}}$ in \eqref{defvKL} but with 
the symbol $\hat{\beta}$ put in place of $\beta$ everywhere. Then
there exists a universal constant $c>0$ such that, for all $\Psi\in\QGV$ and a.e. $\V{x}\in\Geb$,
\begin{align}\nonumber
&\big|\mathfrak{v}_{\UV,\V{x}}[\Psi]-\hat{\mathfrak{v}}_{\UV,\V{x}}[\Psi]\big|
\\\label{bdvmu}
&\le\frac{d_\UV}{2}\|\nabla\Psi(\V{x})\|^2+cd_\UV(1\vee|\ee|\vee|\hat{\ee}|)
\big(\|\Id\Gamma(\omega)^\eh\Psi(\V{x})\|^2+\|\Psi(\V{x})\|^2\big).
\end{align}
\end{lem}

\begin{proof}
The bound \eqref{bdvmu} follows in a straightforward fashion from \eqref{rba}, 
a computation analogous to \eqref{quartus}, and Lemma~\ref{lemtertius}. 
\end{proof}

%%%%%%%%%%%%%%%%%%%%%%%%%%%%%%%%%%%%%%%%%%%%%%%%
%%%%%%%%%%%%%%%%%%%%%%%%%%%%%%%%%%%%%%%%%%%%%%%%
%%%%%%%%%%%%%%%%%%%%%%%%%%%%%%%%%%%%%%%%%%%%%%%%

\section{A Simple Lemma on Resolvents of Positive Operators}\label{appeasyresolvent}

\noindent
The next lemma implies a strengthened version of a well-known criterion for norm resolvent 
convergence of semi-bounded self-adjoint operators.

\begin{lem}\label{lemeasyresolvent}
Let $A$ and $B$ be strictly positive self-adjoint operators in some Hilbert space $\sK$
with lower bounds $\alpha>0$ and $\beta>0$, respectively.
Let $\mathfrak{a}$ and $\mathfrak{b}$ be the associated quadratic forms and assume
that $\dom(\mathfrak{a})=\dom(\mathfrak{b})$. Finally, assume there exist $q\in(0,1)$ and $\zeta\ge0$
such that 
\begin{align}\label{bedAB}
|\mathfrak{a}[\psi]-\mathfrak{b}[\psi]|&\le q\mathfrak{a}[\psi]+q\zeta\|\psi\|^2,
\quad\psi\in\dom(\mathfrak{a}).
\end{align}
Then
\begin{align}\label{easyres1}
\sup_{\psi\in\dom(\mathfrak{a}):\|\psi\|=1}\big\|A^{\eh}(A^{-1}-B^{-1})A^\eh\psi\big\|
\le\frac{q}{1-q}\Big(1+\frac{\zeta}{\alpha}\Big)\Big(1+\frac{q\zeta}{\beta}\Big).
\end{align}
If $M$ is a non-negative self-adjoint operator in $\sK$ with associated quadratic
form $\mathfrak{m}$ such that $\dom(\mathfrak{a})\subset\dom(\mathfrak{m})$ and
$$
\frac{1}{c}\mathfrak{m}[\psi]-c\|\psi\|^2\le\mathfrak{a}[\psi],\quad\psi\in\dom(\mathfrak{a}),
$$ 
for some $c>0$, then
\begin{align*}
\sup_{\phi\in\dom(\mathfrak{m}):\|\phi\|=1}\big\|M^{\eh}(A^{-1}-B^{-1})M^\eh\phi\big\|
\le\frac{cq+c^2q}{1-q}\Big(1+\frac{\zeta}{\alpha}\Big)\Big(1+\frac{q\zeta}{\beta}\Big).
\end{align*}
\end{lem}

\begin{proof}
In view of \eqref{bedAB} and its consequence
$$
\mathfrak{a}[\psi]\le\frac{1}{1-q}\mathfrak{b}[\psi]+\frac{q\zeta}{1-q}\|\psi\|^2,
\quad\psi\in\dom(\mathfrak{a}),
$$
the operators $C:=B^\eh A^{\mh}$ and $D:=A^\eh B^\mh$ are in $\LO(\sK)$ with
$\|C^*C-\id\|\le q(1+\zeta/\alpha)$ and $\|DD^*\|=\|D^*D\|\le(1+q\zeta/\beta)/(1-q)$. 
Furthermore, $A^{\mh}B^\eh\subset C^*$ and $B^{\mh}A^\eh\subset D^*$.
Since the range of $B^{-1}$ is $\dom(B)\subset\dom(A^{\eh})$,
\begin{align*}
A^{-1}-B^{-1}&=(A^{-1}B-\id)B^{-1}
\\
&=A^{\mh}(C^*C-\id)A^{\eh}B^{-1}=A^{\mh}(C^*C-\id)DD^*A^{\mh}.
\end{align*}
Finally, the assumptions on $M$ entail $\|M^\eh A^\mh\|^2\le c+c^2$.
If $\phi\in\dom(\mathfrak{m})$ is normalized, so that in particular
$A^{\mh}M^\eh\phi\in\dom(\mathfrak{a})$, then we infer from the previous bound and \eqref{easyres1} that
$$
\|M^{\eh}(A^{-1}-B^{-1})M^\eh\phi\|\le(c+c^2)^\eh\frac{q}{1-q}
\Big(1+\frac{\zeta}{\alpha}\Big)\Big(1+\frac{q\zeta}{\beta}\Big)\|A^{\mh}M^\eh\phi\|.
$$
We conclude by using that $A^{\mh}M^\eh\subset(M^\eh A^\mh)^*$.
\end{proof}

%%%%%%%%%%%%%%%%%%%%%%%%%%%%%%%%%%%%%%%%%%%%%%%%
%%%%%%%%%%%%%%%%%%%%%%%%%%%%%%%%%%%%%%%%%%%%%%%%
%%%%%%%%%%%%%%%%%%%%%%%%%%%%%%%%%%%%%%%%%%%%%%%%

\section[Feynman-Kac Formulas for Dirichlet Realizations]{Feynman-Kac Formulas for Dirichlet Realizations of Nelson Hamiltonians}\label{appFKD}

\noindent
At the end of this appendix we prove the Feynman-Kac formulas for proper open subsets
$\Geb\subset\RR^3$ asserted in Theorem~\ref{thmFKD}. Departing from the known formulas
in the case $\Geb=\RR^3$, this can be done by a standard procedure for Schr\"{o}dinger operators 
originating from \cite{Simon1978Adv}; see also \cite[App.~B]{BHL2000} where
Schr\"{o}dinger operators with classical magnetic fields are treated. 
In \cite{Matte2019} this procedure is carried through in a slightly abstracted setting 
also covering models of nonrelativistic quantum field theory. All we do here is 
verifying the hypotheses of the next lemma, which is a special case of \cite[Lem.~3.4]{Matte2019}.

We suppose that $A_{\RR^3}$ and $A_\Geb$ are self-adjoint operators in $\HR:=L^2(\RR^3,\sF)$ 
and its subspace $\HR_\Geb:=1_\Geb L^2(\RR^d,\sF)$, respectively, which are semi-bounded from 
below. We denote their quadratic forms by $\mathfrak{a}_{\RR^3}$ and $\mathfrak{a}_\Geb$ and 
suppose that these forms are defined on $\sQ_{\RR^3}$ and $\QGV$, respectively.
We further assume these two forms to be related as described in the following paragraph:

We pick a sequence of compact sets $K_\ell$, $\ell\in\NN$, exhausting $\Geb$, i.e.,
$$
K_\ell\subset\mr{K}_{\ell+1}, \;\,\ell\in\NN,\quad\text{and}\quad 
\bigcup_{\ell\in\NN}K_\ell=\Geb.
$$
Furthermore, we pick cutoff functions $\vt_\ell\in C_0^\infty(\RR^d)$ with $\vt_\ell=1$ on $K_\ell$,
$\vt_\ell=0$ on $K_{\ell+1}^c$, and $0\le\vt_\ell\le1$, for all $\ell\in\NN$.
As in \cite{Simon1978Adv} we put 
\begin{align}\label{defYsGinfty}
Y(\V{x})&:=\left\{\begin{array}{ll}
\frac{1}{\dist(\V{x},\Geb^c)^{3}}+\sum_{\ell=1}^\infty|\nabla\vt_\ell|^2(\V{x}),&\V{x}\in\Geb,
\\
\infty,&\V{x}\in\Geb^c.
\end{array}\right.
\end{align}
The numerical function $Y$ defines a closed form in $\HR$ with domain
$$
\fdom(Y)=\bigg\{\Psi\in L^2(\RR^3,\sF)\,\bigg|\:
\int_{\RR^3}Y(\V{x})\|\Psi(\V{x})\|^2\Id\V{x}<\infty\bigg\}\subset\HR_{\Geb}.
$$
(In general this domain is not dense.) We further set 
$$
\sD_Y:=\sQ_{\RR^3}\cap\fdom(Y)\subset\HR_{\Geb}.
$$
We now fix $t>0$ in the rest of this appendix and assume:
\begin{enumerate}
\item[(a)] $\sD_Y\subset\QGV$. 
\item[(b)] The closure of $\sD_Y$ with respect to the form norm of $A_\Geb$ equals $\QGV$.
\item[(c)] $\mathfrak{a}_\Geb[\Psi]=\mathfrak{a}_{\RR^d}[\Psi]$, for all $\Psi\in\sD_Y$.
\item[(d)] For all $\V{x}\in\RR^d$, there exists a
strongly measurable map $M_t(\V{x}):\Omega\to\LO(\sF)$ such that, for all $\Psi\in\HR$,
\begin{align}\label{domMt}
\|M_t(\V{x})\|\|\Psi(\V{B}^{\V{x}}_t)\|\in L^1(\Omega,\PP),
\end{align}
and, for all bounded and continuous functions $v:\RR^d\to\RR$,
\begin{align}\label{FKabsRRd}
(e^{-t(A_{\RR^3}+v)}\Psi)(\V{x})=\EE\big[e^{-\int_0^tv(\V{B}^{\V{x}}_s)\Id s}
M_t(\V{x})\Psi(\V{B}^{\V{x}}_t)\big],\quad\text{a.e. $\V{x}\in\RR^3$.}
\end{align}
\end{enumerate}

\begin{lem}\label{lemFKGeb}
In the situation described above, let $\Psi\in\HR_\sG$. Then
\begin{align}\label{FKabs}
(e^{-tA_\Geb}\Psi)(\V{x})&=\EE\big[1_{\{\tau_\Geb(\V{x})>t\}}
M_t(\V{x})\Psi(\V{B}^{\V{x}}_t)\big],\quad\text{a.e. $\V{x}\in\RR^3$,}
\end{align}
with $\tau_\Geb(\V{x})$ defined as in \eqref{firstentryBGebc}.
\end{lem}

Now we apply the previous lemma to the Nelson operators and their non-Fock versions. Recall that 
the right hand sides of our Feynman-Kac formulas are defined in \eqref{defTGL} and \eqref{defwtTGL}.

\begin{thm}\label{prop-FK-bd-dom}
Let $\UV\in[0,\infty]$ and $\Psi\in L^2(\Geb,\sF)$. Then, for a.e. $\V{x}\in\Geb$,
\begin{align}\label{FKGeb}
\big(e^{-t\NV_{\Geb,\UV}}\Psi\big)(\V{x})&=\big(T_{\Geb,\UV,t}\Psi\big)(\V{x}),\quad
\big(e^{-t\HV_{\Geb,\UV}}\Psi\big)(\V{x})=\big(\wt{T}_{\Geb,\UV,t}\Psi\big)(\V{x}).
\end{align}
\end{thm}

\begin{proof}
When we apply Lemma~\ref{lemFKGeb} we can substitute 
$(\hv_{\RR^3,\UV},\hv_{\Geb,\UV})$ with $0\le\UV\le\infty$ or
$(\nv_{\RR^3,\UV},\nv_{\Geb,\UV})$ with $0\le\UV<\infty$ for
the pair of forms $(\mathfrak{a}_{\RR^\nu},\mathfrak{a}_\Geb)$.
We consider the forms and operators associated with $\Geb$
as forms and operators in $\HR_\Geb$ in the canonical way; elements
of $L^2(\Geb,\sF)$ are extended by $0$ to $\RR^3$.
Then the Feynman-Kac formulas derived in \cite{MatteMoeller2017} play the role of the
postulated relation \eqref{FKabsRRd} with the obvious choices of $M_t(\V{x})$. The 
integrability condition \eqref{domMt} is valid by virtue of \eqref{kashmir} and \eqref{WinLp}. 

It remains to check the conditions (a), (b), and (c). To this end we recall that, by \eqref{rbcfnv} and 
\eqref{pernille0}, the form norms of $\HV_{\Geb,\UV}$ with $0\le\UV\le\infty$ and of
$\NV_{\Geb,\UV}$ with $0\le\UV<\infty$ are all equivalent to the norm associated with
$\mathfrak{q}_{*}[\Psi]:=\cfG[\Psi]+\|\Psi\|^2$, $\Psi\in{\QGV}$. Here $\cfG$ is
given by \eqref{compform}. Therefore, we can replace the form norm of $A_\Geb$ in (b)
by $\mathfrak{q}_{*}^{1/2}$. But then (a) and (b) are (simple) special cases of \cite[Prop.~5.13]{Matte2019}.
(To prove (a) we have to approximate $\Psi\in\sD_Y$ in the norm $\mathfrak{q}_*^{1/2}$ by elements of 
$\sD(\Geb,\fdom(\Id\Gamma(1\vee\omega)))$, and
the function $Y$ is introduced to ensure that 
$\mathfrak{q}_{*}[\vt_\ell\Psi-\Psi]\to0$, as $\ell\to\infty$.)

Let $\UV$ be finite. Obviously, $\nv_{\RR^3,\UV}[\Phi]=\nv_{\Geb,\UV}[\Phi]$, for all
$\Phi\in\sD(\Geb,\fdom(\Id\Gamma(1\vee\omega)))$, where
$\sD(\Geb,\fdom(\Id\Gamma(1\vee\omega)))$ is a core for 
$\nv_{\Geb,\UV}$ by definition. Since we know by now that $\sD_Y\subset\QGV$, we conclude that 
$\nv_{\RR^3,\UV}[\Psi]=\nv_{\Geb,\UV}[\Psi]$, for every $\Psi\in\sD_Y$.
In the same way we see that 
$\hv_{\RR^3,\UV}[\Psi]=\hv_{\Geb,\UV}[\Psi]$ for all $\UV\in[0,\infty]$.
Thus, (c) is satisfied as well and Lemma~\ref{lemFKGeb} implies \eqref{FKGeb}
in all cases considered at present.

So far we excluded the renormalized Nelson operator, because its form domain is not known explicitly.
To extend the result to $\NV_{\Geb,\infty}$ we recall that $\NV_{\Geb,\UV}$
converges to $\NV_{\Geb,\infty}$ in strong resolvent sense, as $\UV\to\infty$. Hence, it suffices to set
$\UV_n:=n$, $n\in\NN$, and show that
\begin{align*}
\sup_{\V{x}\in\RR^3}\EE\Big[\|W_{\UV_n,t}(\V{x})-W_{\infty,t}(\V{x})\|^p\Big]
\xrightarrow{\;\;n\to\infty\;\;}0,
\end{align*}
for all $p,t>0$, which is done in \cite[Prop.~5.6]{MatteMoeller2017}.
\end{proof}

%%%%%%%%%%%%%%%%%%%%%%%%%%%%%%%%%%%%%%%%%%%%%%%
%%%%%%%%%%%%%%%%%%%%%%%%%%%%%%%%%%%%%%%%%%%%%%%
%%%%%%%%%%%%%%%%%%%%%%%%%%%%%%%%%%%%%%%%%%%%%%%

\section{Proving Exponential Localization of Spectral Subspaces}\label{appexploc}

\noindent
As promised in Subsection~\ref{ssecexploc}, we present detailed proofs of our $L^2$-exponential
localization estimates in this appendix. To this end we shall proceed along the lines of 
\cite{BFS1998b,Griesemer2004} with a few modifications necessary to derive \eqref{exot2} and
to cover weight functions $F$ that increase faster than linearly.

We start with a few bounds on the decay properties of resolvents:

\begin{lem}\label{lemresexploc}
Let $0\le K<\UV\le\infty$ and $F:\RR^3\to\RR$ be locally Lipschitz continuous and bounded from 
below. Assume that the domain of the maximal operator of 
multiplication with $|\nabla F|$ in $L^2(\Geb,\sF)$ contains ${\QGV}$. Assume further that 
\begin{align}\label{preben}
\inf\bigg\{(1-\epsilon)\hv_{\Geb,K,\UV}[\Psi]-\frac{1}{2}\int_{\Geb}|\nabla F(\V{x})|^2
\|\Psi(\V{x})\|^2\Id\V{x}\,\bigg|\:\Psi\in{\QGV},\,\|\Psi\|=1\bigg\}\ge\delta,
\end{align}
for some $\delta>0$ and some $\epsilon\ge0$. Let $z\in\CC$ with $\Re[z]\le0$.
Then the range of $(\HV_{\Geb,K,\UV}-z)^{-1}e^{-F}$ is contained in the domain of $e^F$ and
\begin{align}\label{resexploc}
\|e^F(\HV_{\Geb,K,\UV}-z)^{-1}e^{-F}\|\le\frac{1}{\delta}.
\end{align}
Furthermore, if $\epsilon>0$, then $e^F$ maps the range of 
$(\HV_{\Geb,K,\UV}-z)^{-1}e^{-F}$ into ${\QGV}$ and
\begin{align}\label{resexploc2a}
\|(\HV_{\Geb,K,\UV})^\eh
e^F(\HV_{\Geb,K,\UV}-z)^{-1}e^{-F}\|\le\frac{1}{\epsilon\delta^\eh}.
\end{align}
In particular, if $a,\delta>0$, $\hv_{\Geb,K,\UV}\ge\delta+a^2/2$, $\cK\subset\RR^3$ is 
compact, and $d_\cK(\V{x}):=\dist(\V{x},\cK)$, $\V{x}\in\RR^3$, then the range of 
$(\HV_{\Geb,K,\UV}-z)^{-1}1_{\cK\cap\Geb}$ is contained in $\dom(e^{ad_\cK})$ and
\begin{align}\label{resexploc2}
\|e^{ad_\cK}(\HV_{\Geb,K,\UV}-z)^{-1}1_{\cK\cap\Geb}\|\le\frac{1}{\delta}.
\end{align}
\end{lem}

\begin{proof}
To start with with we assume in addition that $F$ is smooth and bounded with a bounded 
derivative. To derive \eqref{resexploc} from \eqref{exp5} with $f=e^F$ we could just copy the
corresponding arguments in \cite[pp.~326/7]{Griesemer2004}. Since we are also interested in the
bound \eqref{resexploc2a}, we have to extend these arguments slightly: Let $\Re[z]\le0$. As in 
\cite{Griesemer2004} we infer from \eqref{exp5} that $e^F$ maps ${\QGV}$ into itself and
\begin{align*}
\Re\SPn{\Psi}{(e^F\HV_{\Geb,K,\UV}e^{-F}-z)\Psi}&\ge\hv_{\Geb,K,\UV}[\Psi]
-\frac{1}{2}\int_{\Geb}|\nabla F(\V{x})|^2\|\Psi(\V{x})\|^2\Id\V{x},
\end{align*}
for all $\Psi\in e^F\dom(\HV_{\Geb,K,\UV})$. In conjunction with \eqref{preben} and the
Cauchy-Schwarz inequality this yields
\begin{align*}
\|\Psi\|\|e^F(\HV_{\Geb,K,\UV}-z)e^{-F}\Psi\|
&\ge\delta\|\Psi\|^2+\epsilon\|(\HV_{\Geb,K,\UV})^\eh\Psi\|^2
\\
&\ge\delta\|\Psi\|^2+\epsilon\delta^\eh\|\Psi\|\|(\HV_{\Geb,K,\UV})^\eh\Psi\|,
\end{align*}
for $\Psi$ as above. Choosing $\Psi:=e^F(\HV_{\Geb,K,\UV}-z)^{-1}e^{-F}\Phi$, for arbitrary
$\Phi\in L^2(\Geb,\sF)$, we obtain
\begin{align*}
\delta\|e^F(\HV_{\Geb,K,\UV}-z)^{-1}e^{-F}\Phi\|+\epsilon\delta^\eh
\|(\HV_{\Geb,K,\UV})^\eh e^F(\HV_{\Geb,K,\UV}-z)^{-1}e^{-F}\Phi\|&\le\|\Phi\|,
\end{align*}
which proves \eqref{resexploc} and \eqref{resexploc2} under the additional conditions on $F$.

For general $F$ as in the statement, let $F_{n,\ve}$, $\ve\in(0,1]$, denote a standard
mollification of $F_n:=F\wedge n$, where $n\in\NN$. Let $\Phi\in L^2(\Geb,\sF)$ and
put $\Upsilon:=(\HV_{\Geb,\UV}-z)^{-1}e^{-F}\Phi$. Since $F_n-F\le0$, it is then clear that
\begin{align*}
\int_{\Geb}e^{2F_{n,\ve}(\V{x})}\|\Upsilon(\V{x})\|^2\Id\V{x}
&\le\frac{1}{\delta}\|e^{F_{n,\ve}-F_n}\Phi\|^2,\quad\ve\in(0,1],\,n\in\NN.
\end{align*}
By means of the dominated convergence theorem we can now pass to the limit $\ve\downarrow0$
in the previous inequality. Afterwards we let $n\to\infty$ by monotone convergence. This shows
that $\Upsilon\in\dom(e^F)$ and finishes the proof of \eqref{resexploc}. It is now also clear that
$e^{F_{n,\ve}}\Upsilon\to e^{F_n}\Upsilon$, $\ve\downarrow0$, 
and $e^{F_n}\Upsilon\to e^F\Upsilon$, $n\to\infty$, in $L^2(\Geb,\sF)$, 
which together with \eqref{resexploc2a} implies the bounds
\begin{align*}
\|(\HV_{\Geb,K,\UV})^{\eh}(m^{-1}\HV_{\Geb,K,\UV}+1)^\mh e^F\Upsilon\|
&\le\frac{1}{\epsilon\delta^\eh}\|\Phi\|,\quad m\in\NN,
\end{align*}
provided that $\epsilon>0$. To pass to the limit $m\to\infty$, we now apply the monotone 
convergence theorem in a spectral representation of $\HV_{\Geb,K,\UV}$.
\end{proof}

In the main text we applied the next remark and the succeeding example to prove absence of 
ground states of the infrared singular Nelson operator:

\begin{rem}\label{remNresexploc}
Suppose the hypotheses of Lemma~\ref{lemresexploc} are fulfilled for some $0<K<\UV\le\infty$.
Since $\NV_{\Geb,\UV}$ and $\HV_{\Geb,K,\UV}$ are unitarily equivalent via the Gross 
transformation (recall Proposition~\ref{propGrosstrafo} and Remark~\ref{remdefrenNelson}) and since 
$G_{K,\UV}e^F=e^FG_{K,\UV}$, the bounds \eqref{resexploc} and \eqref{resexploc2a} still hold true
when $\NV_{\Geb,\UV}$ is put in place of $\HV_{\Geb,K,\UV}$. 
\end{rem}

\begin{ex}\label{exresexpabs}
Let $a,\zeta>0$ satisfy $a^2/2<\zeta$ and assume that $F:\RR^3\to[0,\infty)$ is Lipschitz continuous 
such that $|\nabla F|\le a$ holds a.e. on $\RR^3$. Then the range of
$(\NV_{\Geb,\infty}-E_{\Geb,\infty}+\zeta)^{-1}e^{-F}$ is contained in the domain of $e^F$ and
\begin{align}\label{resexploc2aN}
\|e^F(\NV_{\Geb,\infty}-E_{\Geb,\infty}+\zeta)^{-1}e^{-F}\|&\le\frac{1}{\zeta-a^2/2}.
\end{align}
To verify this bound we apply Remark~\ref{remNresexploc} with $V-E_{\Geb,\infty}+\zeta$ put in place
of $V$. Then \eqref{preben} is obviously fulfilled with $\epsilon=0$, $\delta=\zeta-a^2/2$
and, for instance, $K=1$.
\end{ex}

After these preparations we are in a position to prove our $L^2$-exponential localization estimates:

\begin{proof}[Proof of Proposition~\ref{propexploc}.]
Instead of treating $\HV_{\Geb',\UV}$ and $\NV_{\Geb',\UV}$ separately, we prove the
proposition for the operator $\HV_{\Geb',K,\UV}$ with arbitrary $0\le K<\UV\le\infty$.
Recall that $\HV_{\Geb',\UV}=\HV_{\Geb',0,\UV}$. 
To obtain the proposition for the Nelson operator we exploit the unitary equivalence
$\NV_{\Geb',\UV}=G_{K,\UV}^*\HV_{\Geb',K,\UV}G_{K,\infty}$,
which holds for all $0<K<\UV\le\infty$, and the obvious fact that
$G_{K,\infty}^*e^{(1-\ve)F}=e^{(1-\ve)F}G_{K,\infty}^*$ on the domain of $e^{(1-\ve)F}$.
Notice also that, for every $R>0$, the operators
$\HV_{\Geb_R,\UV}$, $\HV_{\Geb_R,K,\UV}$, and $\NV_{\Geb_R,\UV}$ all have the same
spectrum, since the latter two are unitarily equivalent and $\HV_{\Geb_R,K,\UV}$ converges
to $\HV_{\Geb_R,\UV}$ in the norm resolvent sense, as $K\downarrow0$.
In particular, the conditions \eqref{defDelta}, \eqref{bedFa}, and \eqref{bedFb} are the same for 
all these operators.

So let $0\le K<\UV\le\infty$. We shall employ the IMS partition of unity 
$\chi_{0,R}^2+\chi_{1,R}^2=1$ introduced in Ex.~\ref{exIMS}.
Similarly as in \cite{Griesemer2004} we define a comparison operator 
\begin{align}\label{cassius1}
Y_{\Geb'}^{R}:=\HV_{\Geb',K,\UV}+(\Th{R}{\UV}-E_{\Geb,\UV})\chi_{0,R}^2
+\frac{1}{2}|\nabla F|^2\chi_{0,R}^2.
\end{align}
Let $\Psi\in\sQ_{\Geb'}$. Then it is straightforward to verify that 
$\chi_{1,R}\Psi$ is not only in $\sQ_{\Geb'}$, but also in $\sQ_{\Geb_R'}$ when considered
as a function on $\Geb_R'$, and that
\begin{align}\label{cassius2}
\hv_{\Geb',K,\UV}[\chi_{1,R}\Psi]=\hv_{\Geb_R',K,\UV}[\chi_{1,R}\Psi]
&\ge E_{\Geb_R',\UV}\|\chi_{1,R}\Psi\|^2\ge E_{\Geb_R,\UV}\|\chi_{1,R}\Psi\|^2.
\end{align}
In the last step we took into account that $\Geb_R'\subset\Geb_R$. Notice that
the previous estimation also holds in the case $\Geb_R'=\emptyset$, where 
$E_{\Geb_R',\UV}=\infty$, since then $\chi_{1,R}\Psi=0$; here we employ the 
usual convention $\infty\cdot0:=0$. Likewise,
\begin{align}\nonumber
\hv_{\Geb',K,\UV}[\chi_{1,R}\Psi]&\ge E_{\Geb_R',\UV}^0\|\chi_{1,R}\Psi\|^2
+\int_{\Geb'}V(\V{x})\|(\chi_{1,R}\Psi)(\V{x})\|^2\Id\V{x}
\\\label{cassius3}
&\ge E_{\Geb_R,\UV}^0\|\chi_{1,R}\Psi\|^2
+\int_{\Geb'}V(\V{x})\|(\chi_{1,R}\Psi)(\V{x})\|^2\Id\V{x},
\\\label{cassius4}
\hv_{\Geb',K,\UV}[\chi_{0,R}\Psi]&\ge E_{\Geb',\UV}\|\chi_{0,R}\Psi\|^2\ge
 E_{\Geb,\UV}\|\chi_{0,R}\Psi\|^2.
\end{align}
By virtue of the IMS localization formula \eqref{IMSq} and \eqref{cassius1} through\eqref{cassius4}
we obtain the following bounds in the sense of quadratic forms on $\sQ_{\Geb'}$,
\begin{align}\label{exp99a}
Y_{\Geb'}^{R}-\lambda&\ge\Th{R}{\UV}-\lambda-\frac{4}{R^2}=\Delta,
\\\label{exp99b}
Y_{\Geb'}^{R}-\lambda&\ge\frac{1}{2}|\nabla F|^2+\Delta\chi_{0,R}^2\ge\frac{1}{2}|\nabla F|^2.
\end{align}
To get \eqref{exp99b} we applied \eqref{cassius2} in the case where \eqref{bedFa} holds,
and \eqref{cassius3} in the case where \eqref{bedFb} is satisfied.
Let $\ve\in(0,1)$ and set $F_\ve:=(1-\ve)F$.
Combining \eqref{exp99a} and \eqref{exp99b} we then find
\begin{align*}
\Big(1-\frac{\ve}{2}(2-\ve)\Big)\Big(Y_{\Geb'}^{R}-\lambda-\frac{\ve\Delta}{5}\Big)
&\ge\frac{1}{2}|\nabla F_\ve|^2-(1-\ve)^2\frac{\ve\Delta}{5}
+\frac{\ve}{2}(2-\ve)\Big(1-\frac{\ve}{5}\Big)\Delta
\\
&\ge\frac{1}{2}|\nabla F_\ve|^2-\frac{\ve\Delta}{5}+\frac{\ve}{2}\Big(1-\frac{1}{5}\Big)\Delta
\\
&\ge\frac{1}{2}|\nabla F_\ve|^2+\delta,
\end{align*}
as quadratic forms on $\sQ_{\Geb'}$, with
\begin{align*}
\delta:=\min\Big\{1 \ , \ \frac{\ve\Delta}{5}\Big\}.
\end{align*}
Putting the expression $V+(\Th{R}{\UV}-E_{\Geb,\UV}+|\nabla F|^2/2)\chi_{0,R}^2-\lambda-\delta$ 
in place of $V$ in \eqref{resexploc} and \eqref{resexploc2a}, we then see that
\begin{align}\label{exp42}
\|e^{F_\ve}(Y_{\Geb'}^{R}-z)^{-1}e^{-F_\ve}\|&\le\frac{1}{\delta},
\\\label{exp43}
\big\|(Y_{\Geb'}^R-\lambda-\delta)^\eh e^{F_\ve}(Y_{\Geb'}^{R}-z)^{-1}e^{-F_\ve}\big\|
&\le\frac{2}{\ve\delta^\eh},
\end{align}
for all $z\in\CC$ with $\Re[z]\le\lambda+\delta$.

Pick some $\vr\in C_0^\infty(\RR,\RR)$ with $0\le\vr\le1$, $\vr(y)=1$ for $|y|\le1/2$, and
$\vr(y)=0$ for $|y|\ge1$. Furthermore, we define $\theta\in C_0^\infty(\RR,\RR)$ by setting
$\theta(x)=1$ for $x\in J:=[E_{\Geb,\UV},\lambda]$, $\theta(x)=\vr(x-E_{\Geb,\UV})$ for
$x\le E_{\Geb,\UV}$, and $\theta(x)=\vr([x-\lambda]/\delta)$ for $x\ge\lambda$. 
(We may assume without loss of generality that $\lambda\ge E_{\Geb',\UV}\ge E_{\Geb,\UV}$,
for otherwise the statement of Proposition~\ref{propexploc} is trivial.)
Next, we define an
extension of $\theta$ to $\CC$ by $\tilde{\theta}(z):=(\theta(x)+\theta'(x)iy)\vr(y)$, where
as usual $z=x+iy$ with $x,y\in\RR$. Notice that $|\partial_{\bar{z}}\tilde{\theta}(z)|/|y|$ 
with $\partial_{\bar{z}}=(\partial_x+i\partial_y)/2$ is integrable on $\CC\setminus\RR$ and
\begin{align}
\int_{\CC\setminus\RR}\frac{|\partial_{\bar{z}}\tilde{\theta}(z)|}{|y|}\Id x\Id y&\le
c\max\{\|\vr'\|_\infty^2,\|\vr''\|_\infty\}
\Big(\frac{1}{\delta}+1+\lambda+\delta-E_{\Geb,\UV}\Big),
\end{align}  
where $\lambda+\delta<\Th{R}{\UV}$ and $c>0$ is some universal constant. 
We shall employ the Helffer-Sj\"{o}strand formula,
\begin{align}\label{HSformula}
\theta(A)&=-\frac{1}{\pi}\int_{\CC\setminus\RR}(A-z)^{-1}\partial_{\bar{z}}\tilde{\theta}(z)\Id x\Id y,
\end{align}
valid for any self-adjoint operator $A$ in some Hilbert space. It follows from the formula for
the fundamental solution to the Cauchy-Riemann operator $\partial_{\bar{z}}$
and the spectral calculus. The main idea \cite{BFS1998b} is to exploit the following key relation
entailed by \eqref{exp99b}  and \eqref{HSformula},
\begin{align*}
1_{J}(\HV_{\Geb',K,\UV})&=\frac{1}{\pi}\int_{\CC\setminus\RR}\Big((Y_{\Geb'}^{R}-z)^{-1}
-(\HV_{\Geb',K,\UV}-z)^{-1}\Big)1_{J}(\HV_{\Geb',K,\UV})\partial_{\bar{z}}\tilde{\theta}(z)\Id x\Id y.
\end{align*}
Multiplying it by $e^{F_\ve\wedge n}$ and by 
$$
Z_m:=(Y_{\Geb'}^{R}-\lambda-\delta)^\eh
(m^{-1}Y_{\Geb'}^{R}-m^{-1}\lambda-m^{-1}\delta+1)^\mh,
$$
applying the second resolvent identity, and using \eqref{exp42} and \eqref{exp43}, we find
\begin{align*}
\max&\Big\{{\delta}\|e^{F_\ve\wedge n}1_{(-\infty,\lambda]}(\HV_{\Geb',K,\UV})\| \ , \
\frac{\ve\delta^\eh}{2}\|Z_me^{F_\ve\wedge n}1_{(-\infty,\lambda]}(\HV_{\Geb',K,\UV})\|\Big\}
\\
&\le\frac{1}{\pi}\Big(\Th{R}{\UV}-E_{\Geb,\UV}+\frac{1}{2}\|\chi_{0,R}|\nabla F|\|_\infty^2\Big)
\|e^{F_\ve}1_{B_{2R}}\|_\infty
\int_{\CC\setminus\RR}\frac{|\partial_{\bar{z}}\tilde{\theta}(z)|}{|y|}\Id x\Id y
\\
&\le \frac{c'}{\delta}\|e^{F_\ve}1_{B_{2R}}\|(1+\|1_{B_{2R}}|\nabla F|\|_\infty^2)
(\Th{R}{\UV}-E_{\Geb,\UV}+1)^2,\quad m,n\in\NN,
\end{align*}
with another universal constant $c'>0$.
Together with a limiting argument this entails the analogues of \eqref{exot1} and \eqref{exot2}
for the operator $\HV_{\Geb',K,\UV}$. Here we also take
\begin{align*}
\|(\HV_{\Geb',K,\UV}-E_{\Geb',\UV})^\eh(Y_{\Geb'}^{R}-\lambda-\delta)^\mh\|^2
&\le\frac{5}{4}\Big(1+\frac{\Th{R}{\UV}-E_{\Geb,\UV}}{\Delta}\Big)
\end{align*}
into account, which follows from the bounds
$\HV_{\Geb',K,\UV}-E_{\Geb',\UV}\le Y_{\Geb'}^{R}-E_{\Geb,\UV}$ 
and $\lambda+\delta<\Th{R}{\UV}$ as well as from \eqref{exp99a}.
\end{proof}

%%%%%%%%%%%%%%%%%%%%%%%%%%%%%%%%%%%%%%%%%%%%%%%
%%%%%%%%%%%%%%%%%%%%%%%%%%%%%%%%%%%%%%%%%%%%%%%
%%%%%%%%%%%%%%%%%%%%%%%%%%%%%%%%%%%%%%%%%%%%%%%

\section{A Compactness Criterion in Fock Space}\label{appcpt}

\noindent
For the reader's convenience we now explain the arguments used in \cite{Matte2016} to
obtain the compactness result of  Proposition~\ref{propcpt}. The following
proof combines Kolmogorov's characterization of compact subsets in $L^p(\RR^d)$
with observations from \cite{GLL2001}.

\begin{proof}[Proof of Proposition~\ref{propcpt}.]
In this proof we shall only treat the case of the Hilbert space $L^2(\RR^3,\sF)$ explicitly. 
To deal with $\sF$ alone we simply
have to ignore everything related to the variable $\V{x}$ in what follows.
According to the canonical isomorphism 
$L^2(\RR^3,\sF)=\bigoplus_{\ell=0}^\infty L^2(\RR^3,L^2_{\mathrm{sym}}(\RR^{3\ell}))$, 
we write every $\Psi\in L^2(\RR^3,\sF)$ as $\Psi=(\Psi^{(\ell)})_{\ell\in\NN_0}$.

{\em Step~1.} First, we pick $\ell\in\NN$ and $\delta\in(0,1]$ and show that
$\{(\Gamma(\vr_\delta)\Phi_\iota)^{(\ell)}:\iota\in\sI\}$ is relatively compact, where
\begin{align*}
(\Gamma(\vr_\delta)\Phi_{\iota})^{(\ell)}(\V{x},\V{k}_1,\ldots,\V{k}_\ell)
:=\bigg(\prod_{j=1}^\ell\vr_\delta(\V{k}_j)\bigg)
\Phi_{\iota}^{(\ell)}(\V{x},\V{k}_1,\ldots,\V{k}_\ell).
\end{align*}
Exploiting the permutation symmetry of $\Phi_{\iota}^{(\ell)}$ in the variables
$(\V{k}_1,\ldots,\V{k}_\ell)=:\V{k}_{[\ell]}$ and the obvious fact that, if $|(\V{x},\V{k}_{[\ell]})|\ge R$,
then at least one of the $\ell+1$ component vectors in $(\V{x},\V{k}_{[\ell]})$ must have norm
$\ge R/(\ell+1)$, we find
\begin{align*}
&\sup_{\iota\in\sI}\int_{\RR^{3(\ell+1)}}1_{\{|(\V{x},\V{k}_{[\ell]})|\ge R\}}|(\Gamma(\vr_\delta)
\Phi_{\iota})^{(\ell)}(\V{x},\V{k}_{[\ell]})|^2\Id(\V{x},\V{k}_{[\ell]})
\\
&\le\sup_{\iota\in\sI}\|1_{\{|\hat{\V{x}}|\ge R/(\ell+1)\}}\Phi_\iota\|^2+
\sup_{\iota\in\sI}\int_{\RR^3}1_{\{|\V{k}|\ge R/(\ell+1)\}}\|(a\Phi_\iota)(\V{k})\|^2\Id\V{k}
\xrightarrow{\;\;R\to\infty\;\;}0.
\end{align*}
Here we used \eqref{cpt1a} and \eqref{cpt2a} in the last step. For every
$\Psi^{(\ell)}\in L^2(\RR^{3(\ell+1)})$, we now put
$(S_{(\V{y},\V{h}_{[\ell]})}\Psi^{(\ell)})(\V{x},\V{k}_{[\ell]})
:=\Psi^{(\ell)}(\V{x}+\V{y},\V{k}_{[\ell]}+\V{h}_{[\ell]})$. Furthermore, we abbreviate 
$S^{(1)}_{\V{h}_j}:=S_{(0,\V{h}_{j},0,...,0)}$, i.e., 
$S^{(1)}_{\V{h}_j}$ shifts the variable $\V{k}_1$ by $\V{h}_j$.
By a telescopic summation and the permutation symmetry of $\Phi_{\iota}^{(\ell)}$ 
in its last $\ell$ variables, we then obtain, for all $\V{y}\in\RR^3$ and
$\V{h}_{[\ell]}=(\V{h}_1,\ldots,\V{h}_\ell)\in\RR^{3\ell}$,
\begin{align*}
&\|(\Gamma(\vr_\delta)\Phi_{\iota})^{(\ell)}
-S_{(\V{y},\V{h}_{[\ell]})}(\Gamma(\vr_\delta)\Phi_{\iota})^{(\ell)}\|
\\
&\le\Big(\int_{\RR^3}\|\Phi_{\iota}(\V{x})-\Phi_{\iota}(\V{x}+\V{y})\|^2\Id\V{x}\Big)^\eh
+\sum_{j=1}^\ell\|(\Gamma(\vr_\delta)\Phi_{\iota})^{(\ell)}-S_{\V{h}_{j}}^{(1)}
(\Gamma(\vr_\delta)\Phi_{\iota})^{(\ell)}\|,
\end{align*}
where the sum is $\le\sqrt{\ell}\max_{j=1}^\ell\triangle_\delta(\V{h}_j)^{\eh}$.
From \eqref{cpt1b} and \eqref{cpt2b} it now follows that
$\|(\Gamma(\vr_\delta)\Phi_{\iota})^{(\ell)}
-S_{(\V{y},\V{h}_{[\ell]})}(\Gamma(\vr_\delta)\Phi_{\iota})^{(\ell)}\|\to0$, as
$(\V{y},\V{h}_{[\ell]})\to0$, uniformly in $\iota\in\sI$.
Altogether this implies that the bounded set $\{(\Gamma(\vr_\delta)\Phi_\iota)^{(\ell)}:\iota\in\sI\}$ is 
relatively compact in $L^2(\RR^{3(\ell+1)})$. Furthermore, it follows directly from
\eqref{cpt2a} and \eqref{cpt2b} that the bounded set
$\{{\Phi}^{(0)}_{\iota}:\iota\in\sI\}$ is relatively compact in $L^2(\RR^{3})$.

{\em Step~2.}
Now let $\{\Phi_n\}_{n\in\NN}$ be a sequence in $\{\Phi_\iota\}_{\iota\in\sI}$.
Since it is bounded, is contains a weakly converging subsequence which we again denote
by $\{\Phi_n\}_{n\in\NN}$ for simplicity. Let $\Phi_\infty$ be its weak limit. Then
$\|\Phi_\infty\|\le\liminf_{n\to\infty}\|\Phi_n\|$. 
We shall find natural numbers $n_1<n_2<\ldots\,$ such that
 $\|\Phi_\infty\|^2\ge\|\Phi_{n_s}\|^2-1/s$, $s\in\NN$, which implies
$\|\Phi_{n_s}\|\to\|\Phi_\infty\|$, $s\to\infty$, and hence $\Phi_{n_s}\to\Phi_\infty$ strongly.

Let $s\in\NN$. Pick $m\in\NN$ and $\delta\in(0,1]$ such that
$N(0,\infty)/(m+1)\le1/4s$ and $N(0,2\delta)\le1/4s$.
After a $(m+1)$-fold iterative selection of subsequences, employing Step~1 successively for 
$\ell\in\{0,\ldots,m\}$, we find natural numbers $\kappa_1<\kappa_2<\ldots\,$ such that every 
sequence $\{(\Gamma(\vr_\delta)\Phi_{\kappa_i})^{(\ell)}\}_{i\in\NN}$ with $\ell\in\{0,\ldots,m\}$
converges strongly to its weak limit $(\Gamma(\vr_\delta)\Phi_\infty)^{(\ell)}$. 
Let $p_{m}$ be the orthogonal projection in $L^2(\RR^3,\sF)$ onto the subspace 
$L^2(\RR^3)\oplus\bigoplus_{\ell=1}^m L^2(\RR^3,L^2_{\mathrm{sym}}(\RR^{3\ell}))$ and
$p_m^\bot=1-p_m$. Then
\begin{align*}
\|\Phi_\infty\|^2&\ge\|p_{m}\Gamma(\vr_\delta)\Phi_\infty\|^2
=\lim_{i\to\infty}\|p_{m}\Gamma(\vr_\delta)\Phi_{\kappa_i}\|^2
=\lim_{i\to\infty}\SPn{\Phi_{\kappa_i}}{p_m\Gamma(\vr_\delta^2)\Phi_{\kappa_i}},
\end{align*}
where we used that $p_m$ and $\Gamma(\vr_\delta)$ commute.
Let $B_{2\delta}$ denote the open ball about $0$ of radius $2\delta$ in $\RR^3$.
Since $\vr_\delta=1$ on $B_{2\delta}^c$, we have
$1-\Gamma(\vr_\delta^2)\le1-\Gamma(1_{B_{2\delta}^c})$, where in each subspace
$L^2_{\mathrm{sym}}(\RR^{3\ell})$ the operator $\Gamma(1_{B_{2\delta}^c})$ acts by
multiplication with the characteristic function of the kartesian product 
$\times_{j=1}^\ell B_{2\delta}^c$.
Of course, if $\V{k}_{[\ell]}=(\V{k}_1,\ldots,\V{k}_\ell)$ is in the complement of 
$\times_{j=1}^\ell B_{2\delta}^c$, then $1_{B_{2\delta}}(\V{k}_j)=1$, for at least
one $j\in\{1,\ldots,\ell\}$. Therefore, 
$1-\Gamma(\vr_\delta^2)\le\Id\Gamma(1_{B_{2\delta}})$, whence
\begin{align*}
\SPn{\Phi_{\kappa_i}}{p_m\Gamma(\vr_\delta^2)\Phi_{\kappa_i}}
&=\|\Phi_{\kappa_i}\|^2-\|p_m^\bot\Phi_{\kappa_i}\|^2-\SPn{\Phi_{\kappa_i}}{
p_m(1-\Gamma(\vr_\delta^2))\Phi_{\kappa_i}}
\\
&\ge\|\Phi_{\kappa_i}\|^2-\frac{N(0,\infty)}{m+1}
-\sup_{\iota\in\sI}\|\Id\Gamma(1_{B_{2\delta}})^\eh\Phi_\iota\|^2,\quad i\in\NN.
\end{align*}
Here the last supremum is equal to $N(0,2\delta)$ in view of \eqref{fordGak}. Putting all these
remarks together we see that
$\|\Phi_\infty\|^2\ge\limsup_{i\to\infty}\|\Phi_{\kappa_i}\|^2-1/2s$. 

It is now clear how to find the above indices $n_1<n_2<\ldots\,$ and we conclude.
\end{proof}

%%%%%%%%%%%%%%%%%%%%%%%%%%%%%%%%%%%%%%%%%%%%%%%
%%%%%%%%%%%%%%%%%%%%%%%%%%%%%%%%%%%%%%%%%%%%%%%
%%%%%%%%%%%%%%%%%%%%%%%%%%%%%%%%%%%%%%%%%%%%%%%

\section{Domination of Inverse Gaussians of Field Operators}\label{appinvGauss}

\noindent
Here we prove the relative bounds employed in Remarks~\ref{introremGauss1} 
and~\ref{introremGauss2}.

\begin{lem}\label{leminvGauss}
Let $\vk:\RR^3\to\RR$ be measurable and a.e. strictly positive and let $f\in L^2(\RR^3)$
and $\alpha>1$ satisfiy $\vk^\mh f\in L^2(\RR^3)$ and
$4\alpha\|(\vk^\mh\vee1)f\|^2<1$. Then every 
$\psi\in\dom(e^{\Id\Gamma(\vk/(\alpha-1))})$ is in the domain of  $e^{\vp(f)^2}$ and 
\begin{align*}
\|e^{\vp(f)^2}\psi\|&\le\frac{1}{\sqrt{1-4\alpha\|(\vk^\mh\vee1)f\|^2}}
\|e^{\Id\Gamma(\vk/(\alpha-1))}\psi\|.
\end{align*}
\end{lem}

\begin{proof}
The proof is based on the normal ordering 
\begin{align}\label{karl1}
\vp(f)^{n}\phi&=\sum_{{k,\ell,m\in\NN_0:\atop2k+\ell+m=n}}\frac{n!}{k!\ell!m!}
\frac{\|f\|^{2k}}{2^k}\ad(f)^\ell a(f)^m\phi,
\end{align}
valid for all $\phi\in\dom(\Id\Gamma(\vk)^{\nf{n}{2}})$, as well as on the following relative bound 
implied by the Cauchy-Schwarz inequality (somewhat similarly to \eqref{myresluger1}),
\begin{align}\label{karl2}
\|a(f)^m\phi\|&\le\|\vk^\mh f\|^m\|\Id\Gamma(\vk)^{\nf{m}{2}}\phi\|,
\end{align}
for all $m\in\NN$ and $\phi$ in the domain of $\Id\Gamma(\vk)^{\nf{m}{2}}$.

Let $\lambda>0$ and $\psi\in\bigcap_{j=1}^\infty\dom(\Id\Gamma(\vk)^{\nf{j}{2}})$.
Applying \eqref{karl1} and \eqref{karl2} we then find
\begin{align*}
S_f(\psi)&:=\sum_{N=0}^\infty\frac{2^N}{N!}\SPn{\psi}{\vp(f)^{2N}\psi}
\\
&=\sum_{N=0}^\infty\frac{2^N}{N!}\sum_{{k,\ell,m\in\NN_0:\atop2k+\ell+m=2N}}
\frac{(2N)!}{k!\ell!m!}\frac{\|f\|^{2k}}{2^k}\SPn{a(f)^\ell\psi}{a(f)^m\psi}
\\
&\le\sum_{N=0}^\infty\frac{(2\alpha)^{N}}{N!}\|(\vk^\mh\vee1)f\|^{2N}
\\
&\quad\cdot
\sum_{{k,\ell,m\in\NN_0:\atop2k+\ell+m=2N}}
\frac{(2N)!}{(2^kk!)\ell!m!}\frac{1}{\alpha^k(\alpha\lambda)^{\nf{\ell}{2}}
(\alpha\lambda)^{\nf{m}{2}}}
\|\Id\Gamma(\lambda\vk)^{\nf{\ell}{2}}\psi\|\|\Id\Gamma(\lambda\vk)^{\nf{m}{2}}\psi\|.
\end{align*}
Next, we employ the bound $((2p)!)^\eh\le2^pp!$ with $p=N$ and $p=k$ to get
\begin{align*}
\frac{((2N)!)^\eh}{N!2^kk!}\le\frac{2^{N}}{((2k)!)^\eh}.
\end{align*}
This implies
\begin{align*}
S_f(\psi)&\le\sum_{N=0}^\infty(4\alpha)^{N}\|(\vk^\mh\vee1)f\|^{2N}
C_{2N}^{\alpha,\lambda}(\vk,\psi),
\end{align*}
where we abbreviate
\begin{equation*}
C_{2N}^{\alpha,\lambda}(\vk,\psi):=\!\!\sum_{{k,\ell,m\in\NN_0:\atop2k+\ell+m=2N}}\!\!
\Big(\frac{(2N)!}{(2k)!\ell!m!}\Big)^\eh
\frac{\|(\ell!)^\mh\Id\Gamma(\lambda\vk)^{\nf{\ell}{2}}\psi\|
\|(m!)^\mh\Id\Gamma(\lambda\vk)^{\nf{m}{2}}\psi\|}{
\alpha^k(\alpha\lambda)^{\nf{\ell}{2}}(\alpha\lambda)^{\nf{m}{2}}}.
\end{equation*}
With the help of the Cauchy-Schwarz inequality and the multinomial theorem we obtain
\begin{align*}
C_{2N}^{\alpha,\lambda}(\vk,\psi)&\le
\Bigg(\sum_{{k,\ell,m\in\NN_0:\atop2k+\ell+m=2N}}\frac{(2N)!}{(2k)!\ell!m!}
\frac{1}{\alpha^{2k}(\alpha\lambda)^\ell(\alpha\lambda)^m}\Bigg)^{\eh}
\\
&\quad\cdot
\Bigg(\sum_{{k,\ell,m\in\NN_0:\atop2k+\ell+m=2N}}
\frac{\SPn{\psi}{\Id\Gamma(\lambda\vk)^\ell\psi}}{\ell!}
\frac{\SPn{\psi}{\Id\Gamma(\lambda\vk)^m\psi}}{m!}\Bigg)^\eh
\\
&\le\Big\{\frac{1}{\alpha}+\frac{2}{\alpha\lambda}\Big\}^{N}
\SPB{\psi}{\sum_{j=1}^{2N}\frac{1}{j!}\Id\Gamma(\lambda\vk)^j\psi}.
\end{align*}
We now choose $\lambda:=2/(\alpha-1)$ so that the curly bracket $\{\cdots\}$ in the last line
of the previous estimation equals $1$.
We further assume in addition that $\psi\in\dom(e^{\Id\Gamma(\vk/(\alpha-1))})$ and let
$\nu_\psi$ denote the spectral measure of $\vp(f)$ associated with $\psi$. Putting the
above remarks together we then conclude
\begin{align*}
\int_{\RR}e^{2t^2}\Id\nu_\psi(t)&=\sum_{N=0}^\infty\frac{2^N}{N!}
\int_{\RR}t^{2N}\Id\nu_\psi\le\frac{\|e^{\Id\Gamma(\vk/(\alpha-1))}\psi\|^2}{
1-4\alpha\|(\vk^\mh\vee1)f\|^2},
\end{align*}
where we applied Fubini's theorem for non-negative functions in the first step.
\end{proof}

%%%%%%%%%%%%%%%%%%%%%%%%%%%%%%%%%%%%%%%%%%%%%%%
%%%%%%%%%%%%%%%%%%%%%%%%%%%%%%%%%%%%%%%%%%%%%%%
%%%%%%%%%%%%%%%%%%%%%%%%%%%%%%%%%%%%%%%%%%%%%%%

\section*{Acknowledgments}

F.H. thanks Aalborg University for their kind hospitality. F.H. is financially supported by 
a Grant-in-Aid for Science Research ((B)16H03942) from the Japan Society for the Promotion of Science. 

O.M. thanks Kyushu University for their kind hospitality. O.M. was supported by the
VILLUM Foundation via the project grant ``Spectral Analysis of Large Particle Systems'' (VKR023170) 
 during the early phase of work on this article. 
 
F.H. and O.M. are grateful for support by the Danish Agency for Science, Technology and 
Innovation via the International Network Programme grant ``Exciting Polarons'' (5132-00122B).

%%%%%%%%%%%%%%%%%%%%%%%%%%%%%%%%%%%%%%%%%%%%%%%
%%%%%%%%%%%%%%%%%%%%%%%%%%%%%%%%%%%%%%%%%%%%%%%
%%%%%%%%%%%%%%%%%%%%%%%%%%%%%%%%%%%%%%%%%%%%%%%


\begin{thebibliography}{42}

\bibitem{AizenmanSimon1982}
M.~Aizenman and B.~Simon,
\newblock Brownian motion and Harnack's inequality for Schr\"odinger operators,
\newblock {\em Comm. Pure Appl. Math.} \textbf{35} (1982) 209--271. 

\bibitem{Ammari2000}
Z.~Ammari, Asymptotic completeness for a renormalized nonrelativistic Hamiltonian in 
quantum field theory: the Nelson model,
\newblock {\em Math. Phys. Anal. Geom.} \textbf{3} (2000) 217--285 .

\bibitem{Arai2001}
A.~Arai, Ground state of the massless Nelson model without infrared cutoff in a non-Fock 
representation, {\em Rev. Math. Phys.} \textbf{13} (2001) 1075--1094. 

\bibitem{BFS1998b}
V.~Bach, J.~Fr{\"o}hlich and I.~M.~Sigal, 
\newblock Quantum electrodynamics of confined nonrelativistic particles,
\newblock {\em Adv. Math.} \textbf{137} (1998) 299--395. 

\bibitem{BFS1999}
V.~Bach, J.~Fr{\"o}hlich and I.~M.~Sigal, 
\newblock Spectral analysis for systems of atoms and molecules coupled to the
  quantized radiation field,
\newblock {\em Commun. Math. Phys.} \textbf{207} (1999) 249--290. 

\bibitem{BDP2011}
S.~Bachmann, D.~A.~Deckert and A.~Pizzo, 
\newblock {The mass shell of the {N}elson model without cutoffs},
\newblock {\em J. Funct. Anal.} \textbf{263} (2012) {1224--1282.}

\bibitem{Berezanskii1986}
{Yu.~M.~Berezanski\u\i,}
{\em Selfadjoint operators in spaces of functions of infinitely many variables},
{Translations of Mathematical Monographs}, Vol. {63} ({American Mathematical Society, 
Providence, Rhode Island}, {1986}).

\bibitem{BetzHiroshima2009}
V.~Betz and F.~Hiroshima, {Gibbs measures with double stochastic integrals on a path space},
{\em Infin. Dimens. Anal. Quantum Probab. Relat. Top.} \textbf{12} (2009) {135--152.} 

\bibitem{BetzSpohn2005}
V.~Betz and H.~Spohn, {A central limit theorem for {G}ibbs measures relative to {B}rownian motion},
{\em Probab. Theory Relat. Fields} \textbf{131} (2005) {459--478.} 

\bibitem{BHLMS2002}
V.~Betz, F.~Hiroshima, J.~L\H{o}rinczi, R.~A.~Minlos and H.~Spohn, 
Ground state properties of the Nelson Hamiltonian: A Gibbs measure-based approach,
{\em Rev. Math. Phys.} \textbf{14} (2002) 173--198. 

\bibitem{BHL2000}
K.~Broderix, D.~Hundertmark and H.~Leschke, 
\newblock Continuity properties of Schr\"odinger semigroups with magnetic fields,
\newblock {\em Rev. Math. Phys.} \textbf{12} (2000) 181--225. 

\bibitem{Cannon1971}
J.~T.~Cannon, Quantum field theoretic properties of a model of Nelson: domain and eigenvector
stability for perturbed linear operators,
\newblock {\em J. Funct. Anal.} \textbf{8} (1971) 101--152. 

\bibitem{Dam2018}
T.~N.~Dam,  Absence of ground states in the translation invariant massless Nelson model,
{\em Ann. Henri Poincar\'{e}} \textbf{21} (2020) 2655--2679.

\bibitem{DamHinrichs2021}
T.~N.~Dam and B.~Hinrichs, Absence of ground states in the renormalized massless translation-invariant Nelson model,
arXiv:1909.07661v2 (2021, preprint) 35 pp.

\bibitem{DerezinskiGerard1999}
J.~Derezi\'{n}ski and C.~G\'{e}rard, 
Asymptotic completeness in quantum field theory. 
Massive Pauli-Fierz Hamiltonians, {\em Rev. Math. Phys.} \textbf{11} (1999) {383--450.} 

\bibitem{DerezinskiGerard2004}
J.~Derezi\'{n}ski and C.~G\'{e}rard, 
Scattering theory of infrared divergent Pauli-Fierz Hamiltonians,
{\em Ann. Henri Poincar\'{e}} \textbf{5} (2004) 523--577.  

\bibitem{Faris1972}
W.~G.~Faris, Invariant cones and uniqueness of the ground state for Fermion systems,
\newblock {\em J. Math. Phys.} \textbf{13} (1972) 1285--1290. 

\bibitem{FarisSimon1975}
W.~G.~Faris and B.~Simon,  Degenerate and non-degenerate ground states for Schr\"{o}dinger operators,
{\em Duke Math. J.} \textbf{42} (1975) 559--567.

\bibitem{Froehlich1974}
J.~Fr{\"o}hlich, 
\newblock {Existence of dressed one electron states in a class of persistent models},
\newblock {\em Fortschritte Phys.} \textbf{22} (1974) 159--198.

\bibitem{Gerard2000}
C.~G\'{e}rard, 
\newblock {On the existence of ground states for massless Pauli-Fierz Hamiltonians},
\newblock {\em Ann. Henri Poincar\'{e}} \textbf{1} (2000) 443--459. 

\bibitem{GHPS2009}
C.~G\'{e}rard, F.~Hiroshima, A.~Panati and A.~ Suzuki, 
{Infrared divergence of a scalar quantum field model on a pseudo {R}iemannian manifold},
{\em Interdiscip. Inform. Sci.} \textbf{15} (2009) {399--421.} 

\bibitem{GHPS2012}
C.~G\'{e}rard, F.~Hiroshima, A.~Panati and A.~ Suzuki, 
{Absence of ground state for the Nelson model on static space-times},
{\em J. Funct. Anal.} \textbf{262} (2012) 273--299. 

\bibitem{Griesemer2004}
M.~Griesemer, 
{Exponential decay and ionization thresholds in non-relativistic quantum electrodynamics},
{\em J. Funct. Anal.} \textbf{210} (2004)  321--340. 

\bibitem{GLL2001}
M.~Griesemer, E.~H.~Lieb and M.~Loss, 
\newblock Ground states in non-relativistic quantum electrodynamics,
\newblock {\em Invent. Math.} \textbf{145} (2001)  557--595. 

\bibitem{GriesemerWuensch2017}
M.~Griesemer and A.~W\"{u}nsch, On the domain of the Nelson Hamiltonian,
{\em J. Math. Phys.} \textbf{59} (2018) {042111, 21 pp.} 

\bibitem{Gross1972}
L.~Gross,  {Existence and uniqueness of physical ground states},
\newblock {\em J. Funct. Anal.} \textbf{10} (1972) 52--109. 

\bibitem{GHL2014}
M.~Gubinelli, F.~Hiroshima and J.~L\H{o}rinczi, 
Ultraviolet renormalization of the Nelson Hamiltonian through functional integration,
{\em J. Funct. Anal.} \textbf{267} (2014) 3125--3153. 

\bibitem{GMM2017}
B.~G\"{u}neysu, O.~Matte and J.~S.~M{\o}ller,
Stochastic differential equations for models of non-relativistic matter interacting
with quantized radiation fields,
{\em Probab. Theory Relat. Fields} \textbf{167} (2017) 817--915.

\bibitem{HackenbrochThalmaier1994}
W.~Hackenbroch and A.~Thalmaier, 
\newblock {\em Stochastische Analysis},
\newblock (Teubner, Stuttgart, 1994).

\bibitem{HaslerHerbst2008}
D.~Hasler and I.~Herbst, 
\newblock On the self-adjointness and domain of Pauli-Fierz type Hamiltonians,
\newblock {\em Rev. Math. Phys.}  \textbf{20} (2008) 787--800. 

\bibitem{HillePhillips1957}
E.~Hille and R.~S.~Phillips,  {\em Functional analysis and semigroups},
American Mathematical Society Colloquium Publications, Vol. XXXI (American Mathematical Society,
Providence, Rhode Island, 1957).

\bibitem{Hirokawa2006}
M.~Hirokawa, Infrared catastrophe for Nelson's model -- Non-existence of ground state and 
soft-boson divergence, {\em Publ. RIMS} \textbf{42} (2006) 897--922. 

\bibitem{HHS2005}
M.~Hirokawa, F.~Hiroshima and H.~Spohn, Ground state for point particles interacting through a massless 
scalar Bose field, {\em Adv. Math.} \textbf{191} (2005) 339--392. 

\bibitem{Hiroshima2000esa}
F.~Hiroshima, 
\newblock {Essential self-adjointness of translation-invariant quantum
field models for arbitrary coupling constants},
\newblock {\em Commun. Math. Phys.} \textbf{211} (2000) {585--613.} 

\bibitem{Hiroshima2002}
F.~Hiroshima, 
\newblock Self-adjointness of the Pauli-Fierz Hamiltonian for arbitrary values of coupling constants,
\newblock {\em Ann. Henri Poincar\'{e}} \textbf{3} (2002) 171--201. 

\bibitem{Hiroshima2004}
{F.~Hiroshima,} 
{Analysis of ground states of atoms interacting with a quantized radiation field},
in {\em Topics in the theory of {S}chr\"odinger operators}, ed. H.~Araki
({World Scientific, River Edge, NJ}, {2004}), pp. {145--272}.

\bibitem{Hiroshima2014}
{F.~Hiroshima,}
Functional integral approach to semi-relativistic Pauli-Fierz models,
{\em Adv. Math.} \textbf{259} (2014) 784--840. 

\bibitem{HiroshimaLorincziTakaesu2012}
F.~Hiroshima, J.~L\H{o}rinczi and T.~Takaesu, {A probabilistic representation of the ground state 
expectation of fractional powers of the boson number operator},
{\em J. Math. Anal. Appl.} \textbf{395} (2012) {437--447.} 

\bibitem{KMS2011}
M.~K\"{o}nenberg, O.~Matte and E.~Stockmeyer, 
Existence of ground states of hydrogen-like atoms in relativistic QED I: The semi-relativistic 
Pauli-Fierz operator, {\em Rev. Math. Phys.} \textbf{23} (2011) 375--407.

\bibitem{Lampart2021}
J.~Lampart, The resolvent of the Nelson Hamiltonian improves positivity,
{\em Math. Phys. Anal. Geom.} \textbf{24} (2021) Art.~2, 17 pp. 

\bibitem{LampartSchmidt2019}
J.~Lampart and J.~Schmidt, On Nelson-type Hamiltonians and abstract boundary conditions,
{\em Commun. Math. Phys.} \textbf{367} (2019) 629--663. 

\bibitem{LHB2011}
J.~L\H{o}rinczi, F.~Hiroshima and V.~Betz, 
\newblock {\em Feynman-Kac-type theorems and Gibbs measures on path space},
\newblock De Gruyter Studies in Mathematics, Vol. {34} (Walter de Gruyter \& Co., Berlin-Boston, 2011).

\bibitem{LorincziMinlosSpohn2002}
J.~L{\H{o}}rinczi, R.~A.~Minlos and H.~Spohn, {The infrared behaviour in {N}elson's model of a quantum
particle coupled to a massless scalar field},
{\em Ann. Henri Poincar\'{e}} \textbf{3} (2002) {269--295.}

\bibitem{LorincziMinlosSpohn2002b}
J.~L{\H{o}}rinczi, R.~A.~Minlos and H.~Spohn,
{Infrared regular representation of the three-dimensional
massless {N}elson model}, {\em Lett. Math. Phys.} \textbf{59} (2002) {189--198.} 

\bibitem{Matte2016}
O.~Matte, Continuity properties of the semi-group and its integral kernel in non-relativistic QED,
{\em Rev. Math. Phys.} \textbf{28} (2016) 1650011, 90 pp. 

\bibitem{Matte2017}
O.~Matte,
Pauli-Fierz type operators with singular electromagnetic potentials on general domains,
{\em Math. Phys. Anal. Geom.} \textbf{20} (2017) Art.~18, 41 pp. 

\bibitem{Matte2019}
O.~Matte,
Feynman-Kac formulas for Dirichlet-Pauli-Fierz operators with singular coefficients,
arXiv:1906.07616  (2019, preprint) 48 pp.

\bibitem{MatteMoeller2017}
O.~Matte and J.~S.~M{\o}ller,  Feynman-Kac formulas for the ultraviolet renormalized Nelson model.
{\em Ast\'{e}risque} \textbf{404} (2018) vi+110 pp.

\bibitem{Miyao2019}
T.~Miyao, On the semigroup generated by the renormalized Nelson Hamiltonian. {\em J. Funct. Anal.}
\textbf{276} (2019) 1948--1977. 

\bibitem{Nelson1964proc}
E.~Nelson, Schr\"{o}dinger particles interacting with a quantized scalar field, in
{\em Analysis in Function Space: Proceedings of a conference on the theory and application of 
analysis in function space},
eds. W.~T.~Martin and I.~Segal, Dedham, Mass., June 1963
(MIT Press, Cambridge, Mass., 1964), pp. 87--121.

\bibitem{Nelson1964}
E.~Nelson, 
{Interaction of nonrelativistic particles with a quantized scalar field},
{\em J. Math. Phys.} \textbf{5} (1964) 1190--1197. 

\bibitem{Panati2009}
A.~Panati, {Existence and nonexistence of a ground state for the massless {N}elson model under 
binding condition}, {\em Rep. Math. Phys.} \textbf{63} (2009) {305--330.} 

\bibitem{Parthasarathy1992}
K.~R.~Parthasarathy,  {\em An introduction to quantum stochastic calculus},
Monographs in Mathematics, Vol. 85 (Birkh\"{a}user, Basel, 1992).

\bibitem{Posilicano2020}
A.~Posilicano, On the self-adjointness of $H+A^*+A$,
{\em Math. Phys. Anal. Geom.} \textbf{23} (2020) Art.~37, 31 pp.

\bibitem{ReedSimonI}
M.~Reed and B.~Simon, 
\newblock \textit{Methods of modern mathematical physics, {I}: Functional analysis}, second edition.
\newblock (Academic Press [Harcourt Brace Jovanovich Publishers], New York, 1980).

\bibitem{Sasaki2005}
I.~Sasaki, Ground state of the massless Nelson model in a non-Fock representation,
{\em J. Math. Phys.} \textbf{46} (2005) 102107, 12 pp.

\bibitem{Schmidt2020}
J.~Schmidt, The massless Nelson Hamiltonian and its domain, in
{\em Mathematical challenges of zero-range physics}, ed. A.~Michelangeli,
Rome, July 2018, Springer INdAM Series, Vol.~42 (Springer, Cham, 2021), pp. 57--80.

\bibitem{Simon1974}
B.~Simon, 
\newblock {\em The $P(\phi)_2$ Euclidean (quantum) field theory},
\newblock (Princeton University Press, Princeton, New Jersey, 1974).

\bibitem{Simon1978}
B.~Simon, 
\newblock {A canonical decomposition for quadratic forms with applications to monotone 
convergence theorems},
\newblock {\em J. Funct. Anal.} \textbf{28} (1978) 377--385. 

\bibitem{Simon1978Adv}
B.~Simon, 
\newblock {Classical boundary conditions as a technical tool in modern mathematical physics},
{\em Adv. Math.} \textbf{30} (1978) 268--281. 

\bibitem{Spohn1998}
H.~Spohn, 
\newblock{Ground state of a quantum particle coupled to a scalar Bose field},
\newblock{\em Lett. Math. Phys.} \textbf{44} (1998) 9--16. 
 
\end{thebibliography}
\end{document}